%% file: main.tex
\documentclass[11pt]{article}
\usepackage{palatino}

\input{headers}

\begin{document}

\title{Unclonable Non-Interactive Zero-Knowledge}
\author{Ruta Jawale\thanks{University of Illinois at Urbana-Champaign, USA. Email:\{jawale2,dakshita\}@illinois.edu} \and Dakshita Khurana$^{\ast}$}
\date{}
\maketitle 

\begin{abstract}
A non-interactive ZK (NIZK) proof enables verification of NP statements without revealing secrets about them. However, an adversary that obtains a NIZK proof may be able to clone this proof and distribute arbitrarily many copies of it to various entities: this is inevitable for any proof that takes the form of a classical string. In this paper, we ask whether it is possible to rely on quantum information in order to build NIZK proof systems that are impossible to clone.

We define and construct {\em unclonable non-interactive zero-knowledge arguments (of knowledge)} for NP, addressing a question first posed by Aaronson (CCC 2009). Besides satisfying the zero-knowledge and argument of knowledge properties, these proofs additionally satisfy unclonability. Very roughly, this ensures that no adversary can split an honestly generated proof of membership of an instance $x$ in an NP language $\mathcal{L}$ and distribute copies to multiple entities that all obtain accepting proofs of membership of $x$ in $\mathcal{L}$. Our result has applications to {\em unclonable signatures of knowledge}, which we define and construct in this work; these \emph{non-interactively} prevent replay attacks.
\end{abstract}

\newpage
\tableofcontents

\newpage
\input{intro}
\input{overview}

\input{prelim}

\input{unclonable-nizk}
\input{rom-unclonable-nizk}
\input{unclonable-sigs}
\input{revocable-anonymous-credentials}

\input{acks}

\addcontentsline{toc}{section}{References}
\bibliographystyle{alpha}

\bibliography{refs}

\appendix
\input{udef-app}

\end{document}

%% file: headers.tex
\usepackage[dvipsnames]{xcolor}
\usepackage{amsthm}
\usepackage{fullpage}
\usepackage{mathtools}
\usepackage{graphicx}

\usepackage{graphicx}
\usepackage{amsfonts}
\usepackage{amsthm}
\usepackage{framed}
\usepackage{amsmath}
\usepackage{amssymb}

\newcommand{\negpar}[1][-1em]{%
  \ifvmode\else\par\fi
  {\parindent=#1\leavevmode}\ignorespaces
}

\usepackage{tabularx}
\usepackage{caption}
\usepackage{subcaption}

\usepackage{varioref}
\usepackage{hyperref}
\usepackage{cleveref}
\hypersetup{
    colorlinks=true,
    citecolor=Rhodamine,
    linkcolor=Rhodamine,
    filecolor=Rhodamine,      
    urlcolor=blue,
}

\usepackage[]{algorithm2e}
\usepackage{mathtools}

\crefname{section}{Section}{Sections}
\crefname{appendix}{Appendix}{Appendices}

\crefname{figure}{Figure}{Figures}

\crefname{equation}{Equation}{Equations}

\newtheorem{theorem}{Theorem}[section]
\crefname{theorem}{Theorem}{Theorems}

\theoremstyle{definition}
\newtheorem{definition}[theorem]{Definition}
\crefname{definition}{Definition}{Definitions}

\newtheorem{corollary}[theorem]{Corollary}
\crefname{corollary}{Corollary}{Corollaries}

\crefname{proposition}{Proposition}{Propositions}

\newtheorem{lemma}[theorem]{Lemma}
\crefname{lemma}{Lemma}{Lemmas}

\newtheorem{claim}[theorem]{Claim}
\crefname{claim}{Claim}{Claims}

\newtheoremstyle{myremark}     {10pt}{10pt}{}{}{\scshape}{.}{.5em}{}
\theoremstyle{myremark}
\newtheorem{remark}{Remark}[section]
\Crefname{remark}{Remark}{Remarks}

\newcommand{\poly}{\ensuremath{\mathsf{poly}}\xspace}

\newcommand{\negl}{\ensuremath{\mathsf{negl}}\xspace}

\newcommand{\crs}{\ensuremath{\mathsf{crs}}}

\newcommand{\aux}{\ensuremath{\mathsf{aux}}\xspace}
\newcommand{\Gen}{\ensuremath{\mathsf{Gen}}\xspace}
\newcommand{\NoteGen}{\ensuremath{\mathsf{NoteGen}}\xspace}

\newcommand{\tensor}{\otimes}

\newcommand\urand{\stackrel{\mathclap{\mbox{\text{\tiny\$}}}}{\gets}}

\allowdisplaybreaks[1]

\newcommand{\Ch}{\mathsf{Ch}}
\newcommand{\NP}{\mathsf{NP}}

  \newcommand{\cA}{\ensuremath{{\mathcal A}}\xspace}
  \newcommand{\cB}{\ensuremath{{\mathcal B}}\xspace}
  
  \newcommand{\cD}{\ensuremath{{\mathcal D}}\xspace}
  \newcommand{\cE}{\ensuremath{{\mathcal E}}\xspace}

  \newcommand{\cI}{\ensuremath{{\mathcal I}}\xspace}
  \newcommand{\cJ}{\ensuremath{{\mathcal J}}\xspace}
  \newcommand{\cM}{\ensuremath{{\mathcal M}}\xspace}
  
  \newcommand{\cO}{\ensuremath{{\mathcal O}}\xspace}
  \newcommand{\cP}{\ensuremath{{\mathcal P}}\xspace}
  
  \newcommand{\cR}{\ensuremath{{\mathcal R}}\xspace}
  \newcommand{\cS}{\ensuremath{{\mathcal S}}\xspace}

  \newcommand{\cW}{\ensuremath{{\mathcal W}}\xspace}
  \newcommand{\cX}{\ensuremath{{\mathcal X}}\xspace}
  \newcommand{\cY}{\ensuremath{{\mathcal Y}}\xspace}

\newcommand{\zo}{\ensuremath{{\{0,1\}}}\xspace}
\newcommand{\state}{\ensuremath{\mathsf{st}}}

\newcommand{\st}{\: : \:}
\newcommand{\cL}{\ensuremath{{\mathcal L}}\xspace}
\newcommand{\Dec}{\mathsf{Dec}}
\newcommand{\Enc}{\mathsf{Enc}}

\newcommand{\sP}{\mathsf{P}}
\newcommand{\sV}{\mathsf{V}}

\newcommand{\Sim}{\mathsf{Sim}}

\newcommand{\Com}{\mathsf{Com}}

\newcommand{\td}{\mathsf{td}}

\newcommand{\bbN}{\mathbb{N}}

\usepackage{amsmath}
\usepackage{array}
\usepackage{amssymb}
\usepackage{amsfonts}
\usepackage{bm}
\usepackage{xspace}
\usepackage{float}
\usepackage{graphicx}
\usepackage{enumitem}
\usepackage{color}
\usepackage{comment}
\usepackage{appendix}
\usepackage{mathtools}
\usepackage{bm}

\usepackage[utf8]{inputenc}
\usepackage{balance}
\usepackage{csquotes}
\usepackage{microtype}
\usepackage{breqn}

\usepackage{arydshln}
\usepackage{booktabs}
\usepackage{nicefrac}
\usepackage{enumitem}
\PassOptionsToPackage{hyphens}{url}
\usepackage{url}

\newcommand{\commentout}[1]{}

\newcommand{\Ver}{\ensuremath{\mathsf{Ver}}\xspace}

\newcommand{\Prove}{\ensuremath{\mathsf{Prove}}\xspace}

\newcommand{\Verify}{\ensuremath{\mathsf{Verify}}\xspace}
\newcommand{\Sign}{\ensuremath{\mathsf{Sign}}\xspace}

\newcommand{\Setup}{\ensuremath{\mathsf{Setup}}\xspace}

\newcommand{\sk}{\ensuremath{\mathsf{sk}}\xspace}
\newcommand{\pk}{\ensuremath{\mathsf{pk}}\xspace}

\newcommand{\ct}{\ensuremath{\mathsf{ct}}\xspace}

\newcommand{\id}{\mathsf{id}}

\usepackage[nointegrals]{wasysym}

\usepackage{array,multirow,booktabs,adjustbox,graphicx}

\newcommand{\Ext}{\ensuremath{\mathsf{Ext}}}

\usepackage{physics}
\usepackage{tikz}
\usetikzlibrary{quantikz2}

\usepackage{scrextend}

\usepackage{ marvosym }

\newcommand{\Revoke}{\ensuremath{\mathsf{Revoke}}}
\newcommand{\revnotice}{\ensuremath{\mathsf{revnotice}}}
\newcommand{\VerRevoke}{\ensuremath{\mathsf{VerRevoke}}}
\newcommand{\Issue}{\ensuremath{\mathsf{Issue}}}
\newcommand{\cred}{\ensuremath{\mathsf{cred}}}
\newcommand{\access}{\ensuremath{\mathsf{access}}}
\newcommand{\nym}{\ensuremath{\mathsf{nym}}}
\newcommand{\IGen}{\ensuremath{\mathsf{IssuerKeyGen}}}
\newcommand{\VerCred}{\ensuremath{\mathsf{VerifyCred}}}

%% file: intro.tex
\section{Introduction}

Zero-knowledge (ZK)~\cite{GMR89}
proofs 
allow a prover to convince a verifier about the truth of an (NP) statement, without revealing secrets about it.
These 
are among the most widely used cryptographic primitives, with a rich history of study.

\paragraph{Enhancing Zero-knowledge.}
ZK proofs for NP are typically defined via the simulation paradigm.
A simulator is a polynomial-time algorithm that mimics the interaction of an adversarial verifier with an honest prover, given only the statement, i.e., $x \in \cL$, for an instance $x$ of an $\mathsf{NP}$ language $\cL$.
A protocol satisfies zero-knowledge if it admits a simulator that generates a view for the verifier, which is indistinguishable from the real view generated by an honest prover.
This captures the intuition that any information obtained by a verifier upon observing an honestly generated proof, could have been generated by the verifier ``on its own'' by running the simulator.

Despite being widely useful and popular, there are desirable properties of proof systems that (standard) simulation-based security does not capture.
For example, consider (distributions over) instances $x$ of an NP language $\cL$ where it is hard to find an NP witness $w$ corresponding to a given instance $x$.
In an ``ideal'' world, given just the description of one such NP statement $x \in \cL$, it is difficult for an adversary to find an NP witness $w$, and therefore to output {\em any} proofs of membership of $x \in \cL$.
And yet, upon obtaining a {\em single proof} of membership of $x \in \cL$, it may suddenly become feasible for an adversary to make many copies of this proof, thereby generating {\em several} correct proofs of membership of $x \in \cL$.

Unfortunately, this attack is inevitable for classical non-interactive proofs: given any proof string, an adversary can always make multiple copies of it.
And yet, there is hope to prevent such an attack quantumly, by relying on the {\em no-cloning} principle.

Indeed, a recent series of exciting works have combined cryptography with the no-cloning principle to develop quantum money~\cite{Wie83,AaronsonC13,FGH+12,Zhandry19,Kan18}, quantum tokens for digital signatures~\cite{BS16}, quantum copy-protection~\cite{Aar09,AL21,ALL+21,CLLZ21}, unclonable encryption~\cite{Got03,BL20,AK21,MST21,AKLLZ22}, unclonable decryption~\cite{GZ20}, one-out-of-many unclonable security~\cite{KN23}, and more.
In this work, we combine zero-knowledge and unclonability to address a question first posed by Aaronson~\cite{Aar09}:
\begin{center}
    \emph{Can we construct unclonable quantum proofs?\\ How do these proofs relate to quantum money or copy-protection?}
\end{center}

\subsection{Our Results}
We define and construct unclonable non-interactive zero-knowledge argument of knowledge (NIZKAoK). We obtain a construction in the common reference string (CRS) model, as well as one in the quantum(-accessible) random oracle model (QROM).
The CRS model allows a trusted third-party to set up a structured string that is provided to both the prover and verifier.
On the other hand, the QROM allows both parties quantum access to a truly random function $\cO$.

In what follows, we describe our contributions in more detail.

\subsubsection{Definitional Contributions}

Before discussing how we formalize the concept of unclonability for NIZKs, it will be helpful to define hard distributions over NP instance-witness pairs.

\paragraph{Hard Distributions over Instance-Witness Pairs.}
Informally, an efficiently samplable distribution over instance-witness pairs of a language $\cL$ is a ``hard'' distribution if given an instance sampled randomly from this distribution, it is hard to find a witness.
Then, unclonable security requires that no adversary given an instance $x$ sampled randomly from the distribution, together with an honestly generated proof, can output {\em two accepting proofs} of membership of $x \in \cL$.

More specifically, a hard distribution $(\cX, \cW)$
over $R_{\cL}$
satisfies the following:
for any polynomial-sized (quantum) circuit family $\{C_\lambda\}_{\lambda \in \mathbb{N}}$,
\vspace{-2mm}
\begin{equation*}
  \Pr_{(x,w) \leftarrow (\cX_\lambda,\cW_\lambda)}[C_\lambda(x) \in R_{\cL}(x)] \leq \mathsf{negl}(\lambda).
\end{equation*}

For the sake of simplifying our subsequent discussions and definitions, let us fix a $\NP$ language $\cL$ with corresponding relation $\cR$.
Let $(\cX, \cW)$ be some hard distribution over $\cR$.

\paragraph{A Weaker Definition: Unclonable Security.}
For NIZKs satisfying standard completeness, soundness and ZK, we define a simple, natural variant of unclonable security as follows. 
Informally, a proof system satisfies unclonable security if, given an honest proof for an instance and witness pair $(x, w)$ sampled from a hard distribution $(\cX, \cW)$, no adversary can produce two proofs that verify with respect to $x$ except with negligible probability.
\vspace{-2mm}
\begin{definition} (Unclonable Security of NIZK).
\label{def:uncnizk}
A NIZK proof  $(\mathsf{Setup},\mathsf{Prove},\mathsf{Verify})$ satisfies unclonable security if for every language $\mathcal{L}$ and every hard distribution $(\mathcal{X},\mathcal{W})$ over $R_{\mathcal{L}}$, for every poly-sized quantum circuit family $\{C_\lambda\}_{\lambda \in \mathbb{N}}$,
\vspace{-4mm}
$$\Pr_{(x,w) \leftarrow (\mathcal{X}_\lambda,\mathcal{W}_\lambda)}\Bigg[
\mathsf{Verify}(\mathsf{crs},x,\pi_1) = 1 \bigwedge 
\mathsf{Verify}(\mathsf{crs},x,\pi_2) = 1
\Bigg|
\substack{(\mathsf{crs}, \td) \leftarrow \mathsf{Setup}(1^\lambda)\\
\pi \leftarrow \mathsf{Prove}(\mathsf{crs},x,w)\\
\pi_1, \pi_2 \leftarrow C_\lambda(x, \pi)
}
\Bigg]
\leq \mathsf{negl}(\lambda).
\vspace{-1mm}
$$
\end{definition}

In the definition above, we aim to capture the intuition that one of the two proofs output by the adversary can be the honest proof they received, but the adversary cannot output any other correct proof for the same statement. Of course, such a proof is easy to generate if the adversary is able to find the witness $w$ for $x$, which is exactly why we require hardness of the distribution $(\cX,\cW)$ to make the definition non-trivial.

We also remark that unclonable security of proofs {\em necessitates} that the proof $\pi$ keep hidden any witnesses $w$ certifying membership of $x$ in $\cL$, as otherwise an adversary can always clone the proof $\pi$ by generating (from scratch) another proof for $x$ given the witness $w$.

\paragraph{A Stronger Definition: Unclonable Extractability.}
We can further strengthen the definition above to 
require that any adversary generating two (or more) accepting proofs of membership of $x \in \cL$ given a single proof, must have generated one of the two proofs ``from scratch'' and must therefore ``know'' a valid witness $w$ for $x$.
This will remove the need to refer to hard languages.

In more detail, we will say that a proof system satisfies {\em unclonable extractability} if, 
from any adversary $\cA$ that on input a single proof of membership of $x \in \cL$ outputs two proofs for $x$,
then we can extract a valid witness $w$ from $\cA$ for at least one of these statements with high probability.
Our (still, simplified) definition of unclonable extractability is as follows.

\begin{definition}[Unclonable Extractability.]
    A proof $(\mathsf{Setup},\mathsf{Prove},\mathsf{Verify})$ satisfies unclonable security 
    there exists a QPT extractor $\mathcal{E}$ which is an oracle-aided circuit such that 
    for every language $\cL$ with corresponding relation $\cR_\cL$ and 
    for every non-uniform polynomial-time quantum adversary $\mathcal{A}$,
    for every instance-witness pair $(x,w) \in \cR_\cL$ 
    and $\lambda = \lambda(|x|)$,
    such that there is a polynomial $p(\cdot)$ satisfying:
    \begin{equation*}
    \Pr\Bigg[
    \mathsf{Verify}(\mathsf{crs},x,\pi_1) = 1 \bigwedge 
    \mathsf{Verify}(\mathsf{crs},x,\pi_2) = 1
    \Bigg|
    \substack{(\mathsf{crs}, \td) \leftarrow \mathsf{Setup}(1^\lambda)\\
    \pi \leftarrow \mathsf{Prove}(\mathsf{crs},x,w)\\
    \pi_1, \pi_2 \leftarrow \cA_\lambda(\crs, x, \pi, z)
    }
    \Bigg]
    \geq \frac{1}{p(\lambda)},
    \end{equation*} 
    there is also a polynomial $q(\cdot)$ such that
    \begin{equation*}
        \Pr[(x,w_\cA) \in \cR_\cL| w_\cA \leftarrow \mathcal{E}^{\cA}(x)] \geq \frac{1}{q(\lambda)}.
    \end{equation*}
\end{definition}

In fact, in the technical sections, we further generalize this definition to consider a setting where the adversary obtains an even larger number (say $k-1$) input proofs on instances $x_1, \ldots, x_{k-1}$, and outputs $k$ or more proofs. Then we require the extraction of an NP witness corresponding to any proofs that attempt to ``clone'' honestly generated proofs (i.e. the adversary outputs two or more proofs w.r.t. the same instance $x_i \in \{x_1, \ldots, x_{k-1}\}$). All our theorem statements hold w.r.t. this general definition.
Finally, we also consider definitions and constructions in the quantum-accessible random oracle model (QROM); these are natural generalizations of the definitions above, so we do not discuss them here.

We also show that the latter definition of unclonable extractability implies the former, i.e. unclonable security. Informally, this follows because the extractor guaranteed by the definition of extractability is able to obtain a witness $w$ for $x$ from any adversary, which contradicts hardness of the distribution $(\cX,\cW)$.
We refer the reader to \cref{app:defs-reduct} for a formal proof of this claim.

Moreover, we can generically boost the unclonable-extractor’s success probability from $1/q(\lambda)$ to $1 - \negl(\lambda)$ with respect to a security parameter $\lambda$. For details, see \cref{sec:unc-defs-crs} and \cref{sec:unc-defs-qro}.

\subsubsection{Realizations of Unclonable NIZK, and Relationship with Quantum Money}
We obtain realizations of unclonable NIZKs in both the common reference string (CRS) and the quantum random oracle (QRO) models, assuming public-key quantum money mini-scheme and other (post-quantum) standard assumptions. We summarize these results below. 

\begin{theorem}[Informal]
    Assuming public-key quantum money mini-scheme, public-key encryption, perfectly binding and computationally hiding commitments, and adaptively sound NIZK arguments for $\NP$, there exists an unclonable-extractable NIZK argument of knowledge scheme in the CRS model.
\end{theorem}

Adaptively sound NIZK arguments for $\NP$ exist assuming the polynomial quantum hardness of LWE \cite{PS19}.

\begin{theorem}[Informal]
    Assuming public-key quantum money mini-scheme and honest verifier zero-knowledge arguments of knowledge sigma protocols for $\NP$, there exists an unclonable-extractable NIZK argument of knowledge scheme in the QROM.
\end{theorem}

\paragraph{Is Quantum Money {\em necessary} for Unclonable NIZKs?}
Our work builds unclonable NIZKs for NP by relying on any (public-key) quantum money scheme (mini-scheme), in conjuction with other assumptions such as NIZKs for NP.
Since constructions of public-key quantum money mini-scheme are only known based on post-quantum indistinguishability obfuscation~\cite{AaronsonC13,Zhandry19a}, it is natural to wonder whether 
the reliance on quantum money is inherent.
We show that this is indeed the case, by proving that unclonable NIZKs
in fact imply public-key quantum money mini-scheme.

\begin{theorem}[Informal]
    \label{thm:unizk-qmoney-intro}
    Unclonable NIZK arguments for NP imply public-key quantum money mini-scheme.
\end{theorem}

\subsubsection{Applications}

\noindent \textbf{Unclonable Signatures of Knowledge.}
A (classical) signature scheme asserts that a message $m$ has been signed on behalf of a public key $\pk$.
However, in order for this signature to be authenticated, the public key $\pk$ must be proven trustworthy through a certification chain rooted at a trusted public key $\mathsf{PK}$. However, as \cite{CL06} argue, this reveals too much information; it should be sufficient for the recipient to only know that \emph{there exists} a public key $\pk$ with a chain of trust from $\mathsf{PK}$. To solve this problem, \cite{CL06} propose \emph{signatures of knowledge} which allow a signer to sign \emph{on behalf of an instance $x$ of an $\NP$-hard language} without revealing its corresponding witness $w$. Such signatures provide an anonymity guarantee by hiding the $\pk$ of the sender.

While this is ideal for many applications, anonymity presents the following downside: a receiver cannot determine whether they were the intended recipient of this signature. 
In particular, anonymous signatures are more susceptible to \emph{replay attacks}.
Replay attacks are a form of passive attack whereby an adversary observes a signature and retains a copy. The adversary then leverages this signature, either at a later point in time or to a different party, to impersonate the original signer. The privacy and financial consequences of replay attacks are steep. They can lead to data breach attacks which cost millions of dollars annually and world-wide~\cite{IBM}.

In this work, we construct a signature of knowledge scheme which is the first \emph{non-interactive} signature in the CRS model that is \emph{naturally secure against replay attacks}. 
Non-interactive, replay attack secure signatures have seen a lot of recent interest including a line of works in the bounded quantum storage model \cite{BS23b} and the quantum random oracle model \cite{BS23a}. 
Our construction is in the CRS model and relies on the quantum average-case hardness of $\NP$ problems, plausible cryptographic assumptions, and the axioms of quantum mechanics.
We accomplish this by defining \emph{unclonable signatures of knowledge}: if an adversary, given a signature of a message $m$ with respect to an instance $x$, can produce two signatures for $m$ which verify with respect to the same instance $x$, then our extractor is able to extract a witness for $x$.

\begin{theorem}[Informal]
    Assuming public-key quantum money mini-scheme, public-key encryption, perfectly binding and computationally hiding commitments, and simulation-sound NIZK arguments for $\NP$, there exists an unclonable-extractable signature of knowledge in the CRS model.
\end{theorem}

Our construction involves showing that an existing compiler can be augmented using unclonable NIZKs to construct unclonable signatures of knowledge. 
The authors of \cite{CL06} construct signatures of knowledge from CPA secure dense cryptosystems~\cite{SP92,SCP00} and simulation-sound NIZKs for $\NP$~\cite{Sahai99,SCOPS01}.
Signatures of knowledge are signature schemes in the CRS model for which we associate an instance $x$ in a language $\cL$. This signature is simulatable, so there exists a simulator which can create valid signatures without knowledge of a witness for $x$. Additionally, the signature is extractable which means there is an extractor which is given a trapdoor for the CRS and a signature, and is able to produce a witness for $x$.
We show that, by switching the simulation-sound NIZKs for unclonable simulation-extractable NIZKs (and slightly modifying the compiler), we can construct unclonable signatures of knowledge.\\

\noindent \textbf{Relationship with Revocation.}
A recent exciting line of work obtains \emph{certified deletion} for time-lock puzzles~\cite{EC:Unruh14}, non-local games~\cite{PhysRevA.97.032324}, information-theoretic proofs of deletion with partial security~\cite{CRW20}, encryption schemes~\cite{10.1007/978-3-030-64381-2_4,BK22},  device-independent security of one-time pad encryption with certified deletion~\cite{https://doi.org/10.48550/arxiv.2011.12704}, public-key encryption with certified deletion~\cite{10.1007/978-3-030-92062-3_21}, commitments and zero-knowledge with certified everlasting hiding~\cite{cryptoeprint:2021:1315}, and fully-homomorphic encryption with
certified deletion~\cite{cryptoeprint:2022:295,BK22,BKP23,BGGKMRR23,APV23}.
While certified everlasting deletion of secrets has been explored in the context of {\em interactive} zero-knowledge proofs~\cite{cryptoeprint:2021:1315}, there are no existing proposals for {\em non-interactive} ZK satisfying variants of certified deletion. Our work provides a pathway to building such proofs.

In this work, we construct a quantum \emph{revocable/unclonable anonymous credentials} protocol in which the issuer of credentials uses a pseudonym to anonymize themselves, receivers of credentials do not require any trusted setup, and the issuer has the ability to remove access from other users. Our work follows a line of work on (classical) revocation for anonymous credentials schemes using NIZK~\cite{BCCKLS09,CKS10,AN11}.

In particular, our construction involves noting that NIZK proof systems that are unclonable can also be viewed as supporting a form of certified deletion/revocation, where in order to delete, an adversary must simply return the entire proof. In other words, the (quantum) certificate of deletion is the proof itself, and this certificate can be verified by running the NIZK verification procedure on the proof.
The unclonability guarantee implies that an adversary cannot keep with itself or later have the ability to generate {\em another proof} for the same instance $x$. 
In the other direction, in order to offer certifiable deletion, a NIZK must necessarily be unclonable.
To see why, note that if there was an adversary who could clone the NIZK, we could use this adversary to obtain two copies, and provably delete one of them. Even though the challenger for the certifiable deletion game would be convinced that its proof was deleted, we would still be left with another correct proof.

\subsection{Related Works}

This work was built upon the foundations of and novel concepts introduced by prior literature. 
We will briefly touch upon some notable such results in this section.\\

\noindent\textbf{Unclonable Encryptions.}
\emph{Unclonable encryption}~\cite{Got03,BL20,AK21,MST21,AKLLZ22} imagines an interaction between three parties in which one party receives a quantum ciphertext and splits this ciphertext in some manner between the two remaining parties. At some later point, the key of the encryption scheme is revealed, yet both parties should not be able to simultaneouly recover the underlying message. 
While our proof systems share the ideology of unclonability, we do not have a similar game-based definition of security. This is mainly due to proof systems offering more structure which can take advantage of to express unclonability in terms of simulators and extractors.\\

\noindent\textbf{Signature Tokens.}
Prior work~\cite{BDS17} defines and constructs \emph{signature tokens} which are signatures which involve a quantum signing token which can only be used once before it becomes inert. The setting they consider is where a client wishes to delegate the signing process to a server, but does not wish the server to be able to sign more than one message. They rely on quantum money~\cite{AaronsonC13} and the no-cloning principle to ensure the signature can only be computed once.
For our unclonable signatures of knowledge result, we focus on the setting where a client wishes to authenticate themselves to a server and wants to prevent an adversary from simultaneously, or later, masquerading as them.
\\

\noindent\textbf{One-shot Signatures.}
The authors of~\cite{AGKZ20} introduce the notion of \emph{one-shot signatures} which extend the concept of signature tokens to a scenario where the client and server only exchange classical information to create a one-use quantum signature token.
They show that these signatures can be plausibly constructed in the CRS model from post-quantum indistinguishability obfuscation.
Unless additional measures for security, which we discussed in our applications section, are employed, classical communication can be easily copied and replayed at a later point.
In contrast, we prevent an adversary from simultaneously, or later, authenticating with the client's identity.\\

\noindent\textbf{Post-quantum Fiat-Shamir.}
Our QROM results are heavily inspired by the recent post-quantum Fiat-Shamir result~\cite{LiuZ19} which proves the post-quantum security of NIZKs in the compressed quantum(-accessible) random oracle model (compressed QROM).
These classical NIZKs are the result of applying Fiat-Shamir 
to post-quantum sigma protocols which are HVZKAoKs.
We further extend, and crucially rely upon, their novel proof techniques to prove extractability (for AoK) and programmability (for ZK) to achieve extractability and programmability for some protocols which output quantum proofs.

\subsection{Concurrent Works}


\noindent \textbf{Unclonable Commitments and Proofs.} A recent, concurrent work~\cite{GMR23} defines and constructs unclonable commitments and interactive unclonable proofs. They additionally construct commitments in the QROM that are unclonable with respect to any verification procedures, and they show that it is impossible to have  (interactive) proofs with the same properties.
The authors also observe a similar relationship between non-interactive unclonable proofs and public-key quantum money via unclonable commitments. They also briefly mention a connection between unclonable commitments and unclonable credentials.

In contrast, we define unclonable-extractable proofs which we construct in the \emph{non-interactive} setting in \emph{both} the crs model and the QROM. We also show a relationship between non-interactive unclonable-extractable proofs and quantum money in \emph{both} the crs model and the QROM. Our work also \emph{formalizes} the relationship between unclonable-extractable proofs and unclonable \emph{anonymous} credentials.

%% file: overview.tex
\section{Technical Overview}
\label{sec:tech}

In this section, we give a high-level overview of our construction and the techniques underlying our main results. 

\subsection{Unclonable Extractable NIZKs in the CRS Model}

Our construction assumes the existence of public-key encryption, classical bit commitments where honestly generated commitment strings are perfectly binding, along with
\begin{itemize}
\item {\em Public-key quantum money mini-scheme} (which is known assuming post-quantum $i\mathcal{O}$ and injective OWFs~\cite{Zhandry19a}).
At a high level, public-key quantum money mini-scheme consists of two algorithms: $\mathsf{Gen}$ and $\mathsf{Ver}$. $\mathsf{Gen}$ on input a security parameter, outputs a (possibly mixed-state) quantum banknote $\rho_\$$ along with a classical serial number $s$. 
$\mathsf{Ver}$ is public, takes a quantum money banknote, and
outputs either a classical serial number $s$, or $\bot$ indicating that its input is an invalid
banknote. The security guarantee is that no efficient adversary given an honest banknote $\rho_\$$ can output two notes $\rho_{\$, 0}$ and $\rho_{\$, 1}$ that both pass the verification and have serial numbers equal to that of $\rho_\$$. 
\item {\em Post-quantum NIZKs for NP}, which are known assuming the post-quantum hardness of LWE. These satisfy (besides completeness) (1) soundness, i.e., no efficient prover can generate accepting proofs for false NP statements, and (2) zero-knowledge, i.e., the verifier obtains no information from an honestly generated proof beyond what it could have generated on {\em its own} given the NP statement itself.
\end{itemize}

\noindent{\bf Construction.} Given these primitives, the algorithms $(\mathsf{Setup}, \mathsf{Prove}, \mathsf{Verify})$ of the unclonable extractable NIZK are as follows.

\vspace{2mm}
\noindent {\underline{\textsc{Setup}}$(1^\lambda)$}:
The setup algorithm 
samples a public key $\mathsf{pk}$ of a public-key encryption, the common reference string $\mathsf{crs}$ of a classical (post-quantum) NIZK for NP, along with a perfectly binding, computationally hiding classical commitment to $0^\lambda$ with uniform randomness $t$, i.e. $c = \mathsf{Com}(0^\lambda;t)$.
It outputs $(\mathsf{pk}, \mathsf{crs}, c)$.

\vspace{2mm}
\noindent {\underline{\textsc{Prove}}}:
Given the CRS $(\mathsf{pk}, \mathsf{crs}, c)$,
instance $x$ and witness $w$, output $(\rho_\$,s,ct,\pi)$ where
\begin{itemize}
\item The state $\rho_\$ \leftarrow \mathsf{Gen}$ is generated as a quantum banknote  with associated serial number $s$.
\item The ciphertext $ct = \mathsf{Enc}_{\mathsf{pk}}(w;u)$ is an encryption of the witness $w$ with randomness $u$.
\item The proof string $\pi$ is a (post-quantum) NIZK for the following statement using witness $(w, u)$: \vspace{-2mm}
$$
\textsf{EITHER }\left(\exists w, u: ct = \mathsf{Enc}_\pk(w;u) \wedge R_L(x,w) =1\right)\textsf{ OR }\left(\exists r: c = \mathsf{Com}(s;r)\right),\vspace{-1mm}$$
~~where we recall that $\pk$ and $c$ were a part of the CRS output by the Setup algorithm.
\end{itemize}
\noindent {\underline{\textsc{Verify}}}:
Given CRS $(\mathsf{pk}, \mathsf{crs}, c)$, instance $x$ and proof $(\rho_\$, s, ct, \pi)$, 
check that (1) $\mathsf{Ver}(\rho_\$)$ outputs $s$ and (2) $\pi$ is an accepting NIZK argument of the statement above.

\vspace{2mm}
\noindent \textbf{Analysis.}
Completeness, soundness/argument of knowledge and ZK for this construction follow relatively easily, so we focus on unclonable extractability in this overview.
Recall that unclonable extractability requires that no adversary, given an honestly generated proof for $x\in \mathcal{L}$, can split this into {\em two accepting proofs} for $x \in \mathcal{L}$ (as long as it is hard to find a witness for $x$).
Towards a contradiction, suppose an adversary splits a proof into 2 accepting proofs $(\rho_{\$, 0}, s_1, ct_1, \pi_1)$, $(\rho_{\$, 1}, s_2, ct_2, \pi_2)$. Then,
\begin{itemize}
\item If $s_1 = s_2 = s$, the adversary given one bank note with serial number $s$ generated two valid banknotes $\rho_{\$, 0}$ and $\rho_{\$, 1}$ that both have the same serial number $s$. This contradicts the security of quantum money.
\item Otherwise, there is a $b \in \{1, 2\}$ such that $s_b \neq s$. Then, consider an indistinguishable hybrid where the adversary obtains a simulated proof generated {\em without witness $w$} as follows: (1) sample quantum banknote $\rho_\$$ with serial number $s$, (2) sample public key $\mathsf{pk}$ along with secret key $\mathsf{sk}$, (3) generate $c = \mathsf{Com}(s;t)$, $ct = \mathsf{Enc}_{\mathsf{pk}}(0;u)$, (4) generate proof $\pi$ using witness $t$ (since $c= \mathsf{Com}(s;t)$) instead of using witness $w$. 
Send common reference string $(\mathsf{pk},\mathsf{crs},c)$ and proof $(\rho_\$,s,ct,\pi)$ to the adversary.
Now, the proof that the adversary generates with $s_b \neq s$ {\em must} contain $\mathsf{ct}_b = \mathsf{Enc}_{\mathsf{pk}}(w;u)$, since $c$ being generated as a commitment to $s \neq s_b$ along with the perfect binding property implies that $(\not\exists r: c = \mathsf{Com}(s_b;r))$.
That is, given instance $x$, the adversary can be used to compute a witness $w$ for $x$ by decrypting ciphertext $\mathsf{ct}_b$, thereby contradicting the hardness of the distribution.
\end{itemize}

Our technical construction in \cref{sec:unizk-crs}, while conceptually the same, is formalized slightly differently. It uses NIZKs with an enhanced simulation-extraction property, which can be generically constructed from NIZK (see \cref{sec:simext-nizk-crs}).
Having constructed unclonable extractable arguments in the CRS model, in the next section, we analyze a construction of unclonable extractable arguments in the QROM.

\subsection{Unclonable Extractable NIZK in the QROM}

We now turn our attention to the QRO setting 
in which we demonstrate a protocol which is provably unclonable. 
Our construction assumes the existence of public-key quantum money mini-scheme and a \emph{post-quantum sigma protocol for NP}. 
A sigma protocol $(\sP, \sV)$ is an interactive three-message honest-verifier protocol: the prover sends a commitment message, the verifier sends a uniformly random challenge, and the prover replies by opening its commitment at the locations specified by the random challenge. 

\vspace{0.2cm}

\noindent \textbf{Construction.} The algorithms $(\textsc{Prove}, \textsc{Verify})$ of the unclonable extractable NIZK in the QROM are as follows.\\

\noindent {\underline{\textsc{Prove}}}:
Given an instance $x$ and witness $w$, output $(\rho_\$,s,\alpha,\beta,\gamma)$ where
\begin{itemize}
    \item The quantum banknote $\rho_\$$ is generated alongside associated serial number $s$.
    \item $\sP$ is run to compute the sigma protocol's commitment message as $\alpha$ given $(x, w)$ as input.
    \item The random oracle is queried on input $(\alpha, s, x)$ in order to obtain a challenge $\beta$.
    \item $\sP$ is run, given as input $(x, w, \alpha, \beta)$ and its previous internal state, to compute the sigma protocol's commitment openings as $\gamma$.
\end{itemize}
\noindent {\underline{\textsc{Verify}}}:
Given instance $x$ and proof $(\rho_\$, s, \alpha, \beta, \gamma)$, 
check that (1) the quantum money verifier accepts $(\rho_\$, s)$, (2) the random oracle on input $(\alpha, s, x)$ outputs $\beta$, and (3) $\sV$ accepts the transcript $(\alpha, \beta, \gamma)$ with respect to $x$.\\

\noindent \textbf{Analysis.}
Since the completeness, argument of knowledge and zero-knowledge properties are easy to show, we focus on unclonable extractability. Suppose an adversary was able to provide two accepting proofs $\pi_1 = (\rho_{\$, 0}, s_1, \alpha_1, \beta_1, \gamma_1)$ and $\pi_2 = (\rho_{\$, 1}, s_2, \alpha_2, \beta_2, \gamma_2)$ for an instance $x$ for which it received an honestly generated proof $\pi = (\rho_\$, s, \alpha, \beta, \gamma)$.
Then,
\begin{itemize}
    \item Suppose $s_1 = s_2 = s$. In this case, the adversary given one bank note with serial number $s$ generated two valid banknotes $\rho_{\$, 0}$ and $\rho_{\$, 1}$ that both have the same serial number $s$. This contradicts the security of quantum money.
    
    \item Otherwise, there is a $b \in [1, 2]$ such that $s_b \ne s$. 
    By the zero-knowledge property of the underlying HVZK sigma protocol, this event also occurs when the proof $\pi$ that the adversary is given is replaced with a simulated proof. Specifically, we build a reduction that locally programs the random oracle at location $(\alpha, s, x)$ in order to generate a simulated proof for the adversary. 
    Since the adversary's own proof for $s_b \neq s$ is generated by making a distinct query $(\alpha_b, s_b, x) \neq (\alpha, s, x)$, the programming on $(\alpha, s, x)$ does not affect the knowledge extractor for the adversary's proof, which simply rewinds the (quantum) random oracle to extract a witness for $x$, following~\cite{LiuZ19}. This allows us to obtain a contradiction, showing that our protocol must be unclonable.
\end{itemize}

\subsection{Unclonable NIZKs imply Quantum Money Mini-Scheme}

Finally, we discuss why unclonable NIZKs satisfying even the weaker definition of unclonable security (i.e., w.r.t. hard distributions) imply public-key quantum money mini-scheme. Given an unclonable NIZK, we build a public-key quantum money mini-scheme as follows.\\

\noindent \textbf{Construction.}
Let $(\cX, \cW)$ be a hard distribution over a language $\cL \in \NP$. Let $\Pi = (\Setup, \Prove, \Verify)$ be an unclonable NIZK protocol for $\cL$.

\vspace{2mm}
\noindent {\underline{\textsc{Gen}}$(1^\lambda)$}:
Sample $(x, w) \leftarrow (\cX, \cW)$, $\crs \leftarrow \Setup(1^\lambda, x)$, and an unclonable NIZK proof $\pi$ as $\Prove(\crs, x, w)$. Output a (possibly mixed-state) quantum banknote $\rho_\$ = \pi$, and associated serial number $s = (\crs, x)$.

\vspace{1mm}
\noindent {\underline{\textsc{Ver}}$(\rho_\$,s)$}:
Given a (possibly mixed-state) quantum banknote $\rho_\$$ and a classical serial number $s$ as input, parse $\rho_\$ = \pi$ and $s = (\crs, x)$, and output the result of $\Verify(\crs, x, \pi)$.\\

\noindent \textbf{Analysis.}
The correctness of the quantum money scheme follows from the completeness of the unclonable NIZK $\Pi$. We will now argue that this quantum money scheme is unforgeable. Suppose an adversary $\cA$ given a quantum banknote and classical serial number $(\rho_\$, s)$ was able to output two banknotes $(\rho_{\$, 0}, \rho_{\$, 1})$ both of which are accepted with respect to $s$. 
We can use $\cA$ to define a reduction to the uncloneability of our NIZK $\Pi$ as follows:
\begin{itemize}
    \item The NIZK uncloneability challenger outputs a hard instance-witness pair $(x, w)$, a common reference string $\crs$, and an unclonable NIZK $\pi$ to the reduction.
    
    \item The reduction outputs a banknote $(\rho_\$,s)$ to the adversary, where $\rho_\$ = \pi$ and $s = (\crs, x)$. It receives two quantum banknotes $(\rho_{\$, 0}, \rho_{\$, 1})$ from $\cA$, and finally outputs two proofs $(\pi_0, \pi_1)$ where $\pi_0 = \rho_{\$, 0}$ and $\pi_1 = \rho_{\$, 1}$.
\end{itemize}
If $\cA$ succeeds in breaking unforgeability, then the quantum money verifier accepts both banknotes $(\rho_{\$, 0} = \pi_0, \rho_{\$, 1} = \pi_1)$, with respect to the same serial number $s = (\crs, x)$. By syntax of the verification algorithm, this essentially means that both {\em proofs} $(\pi_0, \pi_1)$ are accepting proofs for membership of the same instance $x \in \cL$, w.r.t. $\crs$, leading to a break in the unclonability of NIZK.

\subsection{Unclonable Signatures of Knowledge}

Informally, a signature of knowledge has the following property: if an adversary, given a signature of a message $m$ with respect to an instance $x$, can produce two signatures for $m$ which verify with respect to the same instance $x$, then the adversary {\em must know} (and our extractor will be able to extract) a witness for $x$.

We obtain unclonable signatures of knowledge assuming the existence of an unclonable extractable \emph{simulation-extractable} NIZK for $\NP$.
Simulation-extractability states that an adversary which is provided any number of simulated proofs for instance and witness pairs of their choosing, cannot produce an accepting proof $\pi$ for an instance $x$ which they have not queried before and where extraction fails to find an accepting witness $w$. Our unclonable extractable NIZK for $\NP$ in the CRS model can, with some extra work, be upgraded to simulation-extractable.

We informally describe the construction of signatures of knowledge from such a NIZK below.

\vspace{0.2cm}

\noindent \textbf{Construction.} 
Let $(\Setup, \sP, \sV)$ be non-interactive simulation-extractable, adaptive multi-theorem computational zero-knowledge, unclonable-extractable protocol for $\NP$.
Let $\cR$ be the $\NP$ relation corresponding to $\cL$.

\vspace{1mm}
\noindent {\underline{\textsc{Setup}}}:
The setup algorithm samples a 
common reference string $\crs$ of an unclonable-extractable simulation-extractable NIZK for $\NP$.
It outputs $\mathsf{crs}$.

\vspace{1mm}
\noindent {\underline{\textsc{Sign}}}:
Given the CRS $\crs$, instance $x$, witness $w$, and message $m$, output signature $\pi$ where
\begin{itemize}
    \item The proof string $\pi$ is an unclonable-extractable simulation-extractable NIZK with tag $m$ using witness $w$ of the following statement: \vspace{-2mm}
    $$\vspace{-2mm}
    \left(\exists w: (x, w) \in \cR\right).$$
\end{itemize}
\vspace{1mm}
\noindent {\underline{\textsc{Verify}}}:
Given CRS $\crs$, instance $x$, message $m$, and signature $\pi$, check that $\pi$ is an accepting NIZK proof with tag $m$ of the statement above.

\noindent \textbf{Analysis.}
The simulatability (extractability) property follows from the zero-knowledge (resp. simulation-extractability) properties of the NIZK. 
Suppose an adversary $\cA$ given a signature $\sigma$ was able to forge two signatures $\sigma_1 = \pi_1$ and $\sigma_2 = \pi_2$, and, yet, our extractor was to fail to extract a witness $w$ from $\cA$. Then,
\begin{itemize}
    \item Either both proofs $\pi_1$ and $\pi_2$ are accepting proofs for membership of the same instance w.r.t. $\crs$. However, this contradicts the unclonability of the NIZK.

    \item Otherwise there exists a proof $\pi_i$ (where $i \in \{1,2\}$) for an instance which $\cA$ has not previously seen a proof for. We can switch to a hybrid where our signatures contain simulated proofs for the NIZK. But now, we have that the verifier accepts a proof for an instance which $\cA$ has not seen a simulated proof for and, yet, we cannot extract a witness from $\cA$. This contradicts the simulation extractability of the NIZK.
\end{itemize}

\paragraph{Roadmap.} In \cref{sec:crscons}, we define and construct unclonable NIZKs in the CRS model, and in Section \ref{sec:qro}, in the QROM. Along the way, we also show that unclonable NIZKs imply quantum money (in the CRS and QRO model respectively). Later, we show how to define and construct unclonable signatures of knowledge from unclonable NIZKs in the CRS model.

%% file: prelim.tex
\section{Preliminaries}

\subsection{Post-Quantum Commitments and Encryption}

\begin{definition}[Post-Quantum Commitments]
\label{def:com}
$\Com$ is a post-quantum commitment scheme if it has the following syntax and properties.

\noindent \textbf{Syntax.}
\begin{itemize}
    \item $c \gets \Com(m; r)$: The polynomial-time algorithm $\Com$ on input a message $m$ and randomness $r \in \zo^{r(\lambda)}$ outputs commitment a $c$.
\end{itemize}

\noindent \textbf{Properties.}
\begin{itemize}
    \item \textbf{Perfectly Binding}: For every $\lambda \in \bbN^+$ and every $m, m', r, r'$ such that $m \ne m'$,
    \begin{equation*}
        \Com(m; r) \ne \Com(m'; r').
    \end{equation*}

    \item \textbf{Computational Hiding}: There exists a negligible function $\negl(\cdot)$ for every polynomial-size quantum circuit $\cD$, every sufficiently large $\lambda \in \bbN^+$, and every $m, m'$,
    \begin{equation*}
        \left\vert \Pr_{\substack{r \urand \zo^{r(\lambda)}, \: c \gets \Com(m; r)}}[\cD(c) = 1] -  \Pr_{\substack{r \urand \zo^{r(\lambda)}, \: c' \gets \Com(m'; r)}}[\cD(c') = 1] \right\vert \le \negl(\lambda).
    \end{equation*}
\end{itemize}
\end{definition}

\begin{theorem}[Post-Quantum Commitment]~\cite{LombardiS19}
    \label{thm:com}
    Assuming the polynomial quantum hardness of LWE, there exists a non-interactive commitment with perfect binding and computational hiding (\cref{def:com}).
\end{theorem}

\begin{definition}[Post-Quantum Public-Key Encryption]
\label{def:enc}
$(\Gen, \Enc, \Dec)$ is a post-quantum public-key encryption scheme if it has the following syntax and properties.

\noindent \textbf{Syntax.}
\begin{itemize}
    \item $(\pk, \sk) \gets \Gen(1^\lambda)$: The polynomial-time algorithm $\Gen$ on input security parameter $1^\lambda$ outputs a public key $\pk$ and a secret key $\sk$.
    \item $c \gets \Enc(\pk, m; r)$: The polynomial-time algorithm $\Enc$ on input a public key $\pk$, message $m$ and randomness $r \in \zo^{r(\lambda)}$ outputs a ciphertext $c$.
    \item $m \gets \Dec(\sk, c)$: The polynomial-time algorithm $\Dec$ on input a secret key $\sk$ and a ciphertext $c$ outputs a message $m$.
\end{itemize}

\noindent \textbf{Properties.}
\begin{itemize}
    \item \textbf{Perfect Correctness}: For every $\lambda \in \bbN^+$ and every $m, r$,
    \begin{equation*}
        \Pr_{\substack{(\pk, \sk) \gets \Gen(1^\lambda)}}[\Dec(\sk, \Enc(\pk, m; r)) = m] = 1.
    \end{equation*}

    \item \textbf{Indistinguishability under Chosen-Plaintext (IND-CPA) Secure}: There exists a negligible function $\negl(\cdot)$ such that for every polynomial-size quantum circuit $\cA = (\cA_0, \cA_1)$ and every sufficiently large $\lambda \in \bbN^+$
    \begin{equation*}
        \left\vert \Pr_{\substack{(\pk, \sk) \gets \Gen(1^\lambda) \\  (m_0, m_1, \zeta) \gets \cA_0(1^\lambda, \pk) \\ c \gets \Enc(\pk, m_0)}}[\cA_1(1^\lambda, c, \zeta) = 1] - \Pr_{\substack{(\pk, \sk) \gets \Gen(1^\lambda) \\  (m_0, m_1, \zeta) \gets \cA_0(1^\lambda, \pk) \\ c \gets \Enc(\pk, m_1)}}[\cA_1(1^\lambda, c, \zeta) = 1]\right\vert \le \negl(\lambda).
    \end{equation*}
\end{itemize}
\end{definition}

\subsection{Sigma protocols}

\begin{definition}[Post-Quantum Sigma Protocol for $\NP$]~\cite{LiuZ19}
\label{def:sigma}
Let $\NP$ relation $\cR$ with corresponding language $\cL$ be given such that they can be indexed by a security parameter $\lambda \in \bbN$.

$\Pi = (\sP = (\sP.\Com, \sP.\Prove), \sV = (\sV.\Ch, \sV.\Ver))$ is a post-quantum sigma protocol if it has the following syntax and properties.

\noindent \textbf{Syntax.} The input $1^\lambda$ is left out when it is clear from context.
\begin{itemize}
    \item $(\alpha, \state) \gets \sP.\Com(1^\lambda, x, w)$: The probabilistic polynomial-size circuit $\sP.\Com$ on input an instance and witness pair $(x, w) \in \cL_\lambda$ outputs a commitment $\alpha$ and an internal prover state $\state$.
    \item $\beta \gets \sV.\Ch(1^\lambda, x, \alpha)$: The probabilistic polynomial-size circuit $\sV.\Ch$ on input an instance $x$ outputs a uniformly random challenge $\beta$.
    \item $\gamma \gets \sP.\Prove(1^\lambda, x, w, \state, \beta)$: The probabilistic polynomial-size circuit $\sP.\Prove$ on input an instance and witness pair $(x, w) \in \cL_\lambda$, an internal prover state $\state$, and a challenge $\beta$ outputs the partial opening (to $\alpha$ as indicated by $\beta$) $\gamma$.
    \item $\sV.\Ver(1^\lambda, x, \alpha, \beta, \gamma) \in \zo$: The probabilistic polynomial-size circuit $\sV.\Ver$ on input an instance $x$, a commitment $\alpha$, a challenge $\beta$, and a partial opening $\gamma$ outputs $1$ iff $\gamma$ is a valid opening to $\alpha$ at locations indicated by $\beta$.
\end{itemize}

\noindent \textbf{Properties.}
\begin{itemize}
    \item {\bf Perfect Completeness.}
    For every $\lambda \in \bbN$ and every $(x, w) \in \cR_\lambda$,
    \begin{equation*}
        \Pr_{\substack{(\alpha, \state) \gets \sP.\Com(x, w) \\ \beta \gets \sV.\Ch(x, \alpha) \\ \gamma \gets \sP.\Prove(x, w, \state, \beta)}}[\sV.\Verify( x, \alpha, \beta, \gamma) = 1] = 1
    \end{equation*}

    \item {\bf Computational Honest-Verifier Zero-Knowledge with Quantum Simulator.}
    There exists a quantum polynomial-size circuit $\Sim$ and a negligible function $\negl(\cdot)$ such that
    for every polynomial-size quantum circuit $\cD$, every sufficiently large $\lambda \in \bbN$, and every $(x, w) \in \cR_\lambda$,
    \begin{equation*}
        \left\vert \Pr_{\substack{ (\alpha, \state) \gets \sP.\Com(x, w) \\ \beta \gets \sV.\Ch( x, \alpha) \\ \gamma \gets \sP.\Prove(x, w, \state, \beta)}}[\cD(x, \alpha, \beta, \gamma) = 1] - \Pr_{\substack{(\alpha, \beta, \gamma) \gets \Sim(1^\lambda, x)}}[\cD(x, \alpha, \beta, \gamma) = 1]\right\vert \le \negl(\lambda).
    \end{equation*}

    \item {\bf Argument of Knowledge with Quantum Extractor.}
    There exists an oracle-aided quantum polynomial-size circuit $\Ext$, a constant $c$, a polynomial $p(\cdot)$, and negligible functions $\negl_0(\cdot)$, $\negl_1(\cdot)$ such that for every polynomial-size quantum circuit  $\cA = (\cA_0, \cA_1)$ where
    \begin{itemize}
        \item $\cA_0(x)$ is a unitary $U_x$ followed by a measurement and
        \item $\cA_1(x, \ket{\state}, \beta)$ is a unitary $V_{x,\beta}$ onto the state $\ket{\state}$ followed by a measurement,
    \end{itemize}
    and every $x$ with associated $\lambda \in \bbN$ satisfying
    \begin{equation*}
        \Pr_{\substack{(\alpha, \ket{\state}) \gets \cA_0(x) \\ \beta \gets \zo^\lambda \\ \gamma \gets \cA_1(x, \ket{\state}, \beta)}}[\sV.\Ver(x, \alpha, \beta, \gamma) = 1] \ge \negl_0(\lambda)
    \end{equation*}
    we have
    \begin{align*}
        &\Pr[(x, \Ext^{\cA(x)}(x)) \in \cR_\lambda]  \ge \frac{1}{p(\lambda)} \cdot \left( \Pr_{\substack{(\alpha, \ket{\state}) \gets \cA_0(x) \\ \beta \gets \zo^\lambda \\ \gamma \gets \cA_1(x, \ket{\state}, \beta)}}[\sV.\Ver(x, \alpha, \beta, \gamma) = 1] - \negl_0(\lambda)\right)^c - \negl_1(\lambda).
    \end{align*}
    When we say $\Ext$ has oracle access to $\cA(x)$, we mean that $\Ext$ has oracle access to both unitaries $U_x, V_{x,\beta}$ and their inverses $U_x^\dagger, V_{x, \beta}^\dagger$.

    \item {\bf Unpredictable Commitment.}
    There exists a negligible function $\negl(\cdot)$ such that for every sufficiently large $\lambda \in \bbN$ and every $(x, w) \in \cR_\lambda$,
    \begin{equation*}
        \Pr_{\substack{ (\alpha, \state) \gets \sP.\Com(x, w) \\ (\alpha', \state') \gets \sP.\Com(x, w)}}[\alpha = \alpha'] \le \negl(\lambda).
    \end{equation*}
\end{itemize}

We note that the unpredictable commitment property in the definition above may appear to be an unusual requirement, but this property is w.l.o.g. for post quantum sigma protocols as shown in~\cite{LiuZ19}.
In particular, any sigma protocol which does not have unpredictable commitments, can be modified into one that does: the prover can append a random string $r$ to the end of their commitment message $\alpha$, and the verifier can ignore this appended string $r$ when they perform their checks.
\end{definition}

\subsection{NIZKs in the CRS model}

We consider the common reference string model.

\begin{definition}[Post-Quantum (Quantum) NIZK for $\NP$ in the CRS Model]
\label{def:nizk-crs}
Let $\NP$ relation $\cR$ with corresponding language $\cL$ be given such that they can be indexed by a security parameter $\lambda \in \bbN$.

$\Pi = (\Setup, \sP, \sV)$ is a non-interactive post-quantum (quantum) zero-knowledge argument for $\NP$ in the CRS model if it has the following syntax and properties.

\noindent \textbf{Syntax.}
The input $1^\lambda$ is left out when it is clear from context.
\begin{itemize}
    \item $(\crs, \td) \gets \Setup(1^\lambda)$: The probabilistic polynomial-size circuit $\Setup$ on input $1^\lambda$ outputs a common reference string $\crs$ and a trapdoor $\td$.
    \item $\pi \gets \sP(1^\lambda, \crs, x, w)$: The probabilistic (quantum) polynomial-size circuit $\sP$ on input a common reference string $\crs$ and instance and witness pair $(x, w) \in \cR_\lambda$, outputs a proof $\pi$.
    \item $\sV(1^\lambda, \crs, x, \pi) \in \zo$: The probabilistic (quantum) polynomial-size circuit $\sV$ on input a common reference string $\crs$, an instance $x$, and a proof $\pi$ outputs $1$ iff $\pi$ is a valid proof for $x$.
\end{itemize}

\noindent \textbf{Properties.}
\begin{itemize}
    \item {\bf Perfect Completeness.}
    For every $\lambda \in \bbN$ and every $(x, w) \in \cR_\lambda$,
    \begin{equation*}
        \Pr_{\substack{(\crs, \td) \gets \Setup(1^\lambda) \\ \pi \gets \sP(\crs, x, w)}}[\sV(\crs, x, \pi) = 1] = 1.
    \end{equation*}

    \item {\bf Adaptive Computational Soundness.}
    There exists a negligible function $\negl(\cdot)$ such that for every polynomial-size quantum circuit $\cA$ and every sufficiently large $\lambda \in \bbN$,
    \begin{equation*}
        \Pr_{\substack{(\crs, \td) \gets \Setup(1^\lambda) \\ (x, \pi) \gets \cA(\crs)}}[\sV(x, \crs, \pi) = 1 \: \land \: x \not\in \cL_\lambda]  \le \negl(\lambda).
    \end{equation*}

    \item {\bf Adaptive Computational Zero-Knowledge.}
    There exists a probabilistic (quantum) polynomial-size circuit $\Sim = (\Sim_0, \Sim_1)$ and a negligible function $\negl(\cdot)$ such that for every polynomial-size quantum circuit $\cA$, every polynomial-size quantum circuit $\cD$, and every sufficiently large $\lambda \in \bbN$,
    \begin{equation*}
        \left\vert \Pr_{\substack{(\crs, \td) \gets \Setup(1^\lambda) \\ (x, w, \zeta) \gets \cA(\crs) \\ \pi \gets \sP(\crs, x, w)}}[\cD(\crs, x, \pi, \zeta) = 1] - \Pr_{\substack{(\crs, \td) \gets \Sim_0(1^\lambda) \\ (x, w, \zeta) \gets \cA(\crs) \\ \pi \gets \Sim_1(\crs, \td, x)}}[\cD(\crs, x, \pi, \zeta) = 1]\right\vert \le \negl(\lambda).
    \end{equation*}
\end{itemize}
\end{definition}

\begin{theorem}[Post-Quantum NIZK argument for $\NP$ in the CRS Model]~\cite{PS19}
    \label{thm:nizk-crs}
    Assuming the polynomial quantum hardness of LWE, there exists a non-interactive adaptively computationally sound, adaptively computationally zero-knowledge argument for $\NP$ in the common reference string model (\cref{def:nizk-crs}).
\end{theorem}

\begin{definition}[Post-Quantum (Quantum) Simulation-Sound NIZK for $\NP$ in CRS Model]
\label{def:simsound-nizk-crs}
Let $\NP$ relation $\cR$ with corresponding language $\cL$ be given such that they can be indexed by a security parameter $\lambda \in \bbN$.

$\Pi = (\Setup, \sP, \sV)$ is a post-quantum (quantum) non-interactive simulation-sound, adaptive multi-theorem computational zero-knowledge protocol for $\NP$ in the CRS model if it has the following syntax and properties.

\begin{itemize}
    \item $\Pi$ is a post-quantum (quantum) non-interactive zero-knowledge argument for $\NP$ in the CRS model (\cref{def:nizk-crs}).

    \item {\bf Adaptive Multi-Theorem Computational Zero-Knowledge.}~\cite{FLS90}
    There exists a probabilistic (quantum) polynomial-size circuit $\Sim = (\Sim_0, \Sim_1)$~\footnote{$\Sim_1$ ignores the second term (a witness $w$) in the queries it receives from $\cA$.} and a negligible function $\negl(\cdot)$ such that for every polynomial-size quantum circuit $\cA$, every polynomial-size quantum circuit $\cD$, and every sufficiently large $\lambda \in \bbN$,
    \begin{equation*}
        \left\vert \Pr_{\substack{(\crs, \td) \gets \Setup(1^\lambda)}}[\cA^{\sP(\crs, \cdot, \cdot)}(\crs) = 1] - \Pr_{\substack{(\crs, \td) \gets \Sim_0(1^\lambda)}}[\cA^{\Sim_1(\crs, \td, \cdot)}(\crs) = 1]\right\vert \le \negl(\lambda).
    \end{equation*}

    \item {\bf Simulation Soundness.}~\cite{Sahai99,SCOPS01}
    Let $\Sim = (\Sim_0, \Sim_1)$ be the simulator given by the adaptive multi-theorem computational zero-knowledge property. There exists a negligible function $\negl(\cdot)$ such that for every oracle-aided polynomial-size quantum circuit $\cA$ and every sufficiently large $\lambda \in \bbN$,
    \begin{equation*}
        \Pr_{\substack{(\crs, \td) \gets \Sim_0(1^\lambda) \\ (x, \pi) \gets \cA^{\Sim_1(\crs, \td, \cdot)}(\crs)}}[\sV(\crs, x, \pi) = 1 \wedge x \not\in Q \wedge x \not\in \cL] \le \negl(\lambda),
    \end{equation*}
    where $Q$ is the list of queries from $\cA$ to $\Sim_1$.
\end{itemize}
\end{definition}

\begin{remark}
    In \cref{def:simsound-nizk-crs}, adaptive multi-theorem computational zero-knowledge implies adaptive computational zero-knowledge.
\end{remark}

\begin{remark}
    As defined in \cref{def:simsound-nizk-crs}, a simulation-sound zero-knowledge protocol has adaptive computational soundness (\cref{def:nizk-crs}). 
\end{remark}

\begin{theorem}[Simulation Sound Compiler]~\cite{SCOPS01}
    \label{thm:simsound-compiler}
    Given one-way functions and a single-theorem NIZK proof system for $\NP$, then there exists a non-interactive simulation sound, adaptively multi-theorem computationally zero-knowledge proof for $\NP$ in the common reference string model (\cref{def:simsound-nizk-crs}).
\end{theorem}

\begin{corollary}[Post-Quantum Simulation Sound NIZK for $\NP$]
    \label{cor:simsound-nizk}
    Assuming the polynomial quantum hardness of LWE, there exists a post-quantum non-interactive simulation sound, adaptively multi-theorem computationally zero-knowledge proof for $\NP$ in the common reference string model (\cref{def:simsound-nizk-crs}).
\end{corollary}

\begin{proof}
    This follows from \cref{thm:nizk-crs} and \cref{thm:simsound-compiler}.
\end{proof}

\subsection{NIZKs in the QRO model}

We now consider the quantum random oracle model. For sake of completeness, we briefly outline a definition for a quantum random oracle.

\begin{definition}
    \label{def:qro}
    A quantum random oracle $\cO$ is a random function which support quantum queries and allows for the following accesses:
    \begin{itemize}
        \item \textbf{Query Access.} On input a message, $\cO$ outputs a uniformly random value. This is the usual access provided. When quantum access may be invoked, we denote the oracle as $\ket{\cO}$.

        \item \noindent \textbf{Programmability Access.} Given programmability access, $\cO$ can be set to output a specified value on a specified input. An arbitrary number of distinct points can be programmed.

        \item \textbf{Extractability Access.} Given extractability access, specific queries to $\ket{\cO}$ can be read.
    \end{itemize}
\end{definition}

\begin{definition}[(Quantum) Post-Quantum NIZKAoK for $\NP$ in QROM]~\cite{LiuZ19}
\label{def:nizkpok-qro}
Let $\cO$ be a random oracle. Let $\NP$ relation $\cR$ with corresponding language $\cL$ be given such that they can be indexed by a security parameter $\lambda \in \bbN$.

$\Pi = (\sP, \sV)$ is a (quantum) non-interactive zero-knowledge argument of knowledge protocol with respect to a random oracle if it has the following syntax and properties.

\noindent \textbf{Syntax.}
The input $1^\lambda$ is left out when it is clear from context.
\begin{itemize}
    \item $\pi \gets \sP^\cO(1^\lambda, x, w)$: The random oracle-aided (quantum) probabilistic polynomial-size circuit $\sP$ on input an instance and witness pair $(x, w) \in \cR_\lambda$, outputs a proof $\pi$.
    \item $\sV^\cO(1^\lambda, x, \pi) \in \zo$: The random oracle-aided (quantum) probabilistic polynomial-size circuit $\sV$ on input an instance $x$ and a proof $\pi$, outputs $1$ iff $\pi$ is a valid proof for $x$.
\end{itemize}

\noindent \textbf{Properties.}
\begin{itemize}
    \item {\bf Perfect Completeness.}
    For every $\lambda \in \bbN$ and every $(x, w) \in \cR_\lambda$,
    \begin{equation*}
        \Pr_{\substack{\cO \\ \pi \gets \sP^\cO(x, w)}}[\sV^\cO(x, \pi) = 1] = 1.
    \end{equation*}

    \item {\bf Zero-Knowledge with Quantum Simulator.}
    There exists a quantum polynomial-size circuit $\Sim$ which ignores its second input and a negligible function $\negl(\cdot)$ such that for every oracle-aided polynomial-size quantum circuit $\cD$ which is limited to making queries $(x, \omega) \in \cR_\lambda$ on input $1^\lambda$, and every sufficiently large $\lambda \in \bbN$,
    \begin{equation*}
        \left\vert \Pr[\cD^{\Sim, \ket{\cO_\Sim}}(1^\lambda) = 1] - \Pr_\cO[\cD^{\sP^\cO, \ket{\cO}}(1^\lambda) = 1]\right\vert \le \negl(\lambda)
    \end{equation*}
    where $\Sim$ simulates the random oracle $\ket{\cO_\Sim}$.

    \item {\bf Argument of Knowledge with Quantum Extractor.}
    There exists an oracle-aided quantum polynomial-size circuit extractor $\Ext$ that simulates a random oracle $\ket{\cO_\Ext}$, a constant $c$, a polynomial $p(\cdot)$, and negligible functions $\negl_0(\cdot)$, $\negl_1(\cdot)$ such that for every polynomial-size quantum circuit $\cA$ and every $x$ with associated $\lambda \in \bbN$ satisfying
    \begin{equation*}
        \Pr_{\substack{\cO \\ \pi \gets \cA^{\ket{\cO}}(x)}}[\sV^\cO(x, \pi) = 1] \ge \negl_0(\lambda)
    \end{equation*}
    we have
    \begin{align*}
        &\Pr[(x, \Ext^{\cA^{\ket{\cO_\Ext}}(x)}(x)) \in \cR_\lambda]  \ge \frac{1}{p(\lambda)} \cdot \left( \Pr_{\substack{\cO \\ \pi \gets \cA^{\ket{\cO}}(x)}}[\sV^\cO(x, \pi) = 1] - \negl_0(\lambda)\right)^c - \negl_1(\lambda).
    \end{align*}
\end{itemize}
\end{definition}

\begin{theorem}[NIZKAoK in QROM~\cite{Unruh17,LiuZ19}]
    \label{thm:pq-fs}
    Let $\Pi$ be a post-quantum sigma protocol (\cref{def:sigma}).
    The Fiat-Shamir heuristic applied to $\Pi$ yields a classical post-quantum NIZKAoK in the QROM (\cref{def:nizkpok-qro}).
\end{theorem}

\subsection{Quantum Money}

\begin{definition}[Public Key Quantum Money Mini-Scheme]~\cite{AaronsonC13,Zhandry19a}
\label{def:qmoney}
$(\Gen, \Ver)$ is a public key quantum money scheme if it has the following syntax and properties.

\noindent \textbf{Syntax.}
\begin{itemize}
    \item $(\rho_\$, s) \gets \Gen(1^\lambda)$: The quantum polynomial-time algorithm $\Gen$ on input security parameter $1^\lambda$ outputs a (possibly mixed-state) quantum banknote $\rho_\$$ (if pure-state, denoted $\ket{\$}$) along with a classical serial number $s$.

    \item $\Ver(\rho_\$, s) \in \zo$: The quantum polynomial-time algorithm $\Ver$ on input a (possibly mixed-state) quantum banknote $\rho_\$$ (if pure-state, denoted $\ket{\$}$) and a classical serial number $s$ outputs $1$ or $0$.
\end{itemize}

\noindent \textbf{Properties.}
\begin{itemize}
    \item \textbf{Perfect Correctness}: For every $\lambda \in \bbN^+$,
    \begin{equation*}
        \Pr_{(\rho_\$, s) \gets \Gen(1^\lambda)}[\Ver(\rho_\$, s) = 1] = 1.
    \end{equation*}

    \item \textbf{Unforgeable}: There exists a negligible function $\negl(\cdot)$ such that for every sufficiently large $\lambda \in \bbN^+$ and every polynomial-size quantum circuit $\cA$,
    \begin{equation*}
        \Pr_{\substack{(\rho_\$, s) \gets \Gen(1^\lambda) \\ (\rho_{\$, 0}, s_0, \rho_{\$, 1}, s_1) \gets \cA(\rho_\$, s)}}[s_0 = s_1 = s \: \land \: \Ver(\rho_{\$, 0}, s_0) = 1 \: \land \: \Ver(\rho_{\$, 1}, s_1) = 1] \le \negl(\lambda).
    \end{equation*}

    \item \textbf{Unpredictable Serial Numbers}: There exists a negligible function $\negl(\cdot)$ such that for every sufficiently large $\lambda \in \bbN$,
    \begin{equation*}
        \Pr_{\substack{(\rho_\$, s) \gets \Gen(1^\lambda) \\ (\rho_\$', s') \gets \Gen(1^\lambda)}}[s = s'] \le \negl(\lambda).
    \end{equation*}
\end{itemize}
\end{definition}

\begin{remark}[Unpredictable Serial Numbers]
    The unpredictable serial numbers property follows, w.l.o.g., from unforgeability. We will briefly outline the reduction. Say that $\Gen$ produced two quantum banknotes $\rho_\$$ and $\rho_\$'$ which had the same serial number $s$ with noticeable probability. Then an adversary $\cA$ that receives $(\rho_\$, s)$ from $\Gen$ could run $\Gen$ again to produce $(\rho_\$', s)$ with noticeable probability. This means that $\cA$ would have produced two quantum banknotes $\rho_\$$ and $\rho_\$'$ which $\Verify$ would accept with respect to the same serial number that $\cA$ received, $s$. 
\end{remark}

\begin{theorem}[Quantum Money from Subspace Hiding Obfuscation~\cite{AaronsonC13,Zhandry19a}]
\label{thm:qmoney}
If injective one-way functions and post-quantum iO exist, then public-key quantum money exists (\cref{def:qmoney}).
\end{theorem}

\begin{definition}[Public Key Quantum Money Mini-Scheme in QROM]
\label{def:qmoney-qro}
$(\Gen, \Ver)$ is a public key quantum money scheme with respect to a quantum random oracle $\cO$ if it has the following syntax and properties.

\noindent \textbf{Syntax.}
\begin{itemize}
    \item $(\rho_\$, s) \gets \Gen^\cO(1^\lambda)$: The random oracle-aided quantum polynomial-time algorithm $\Gen$ on input a security parameter $1^\lambda$ outputs a (possibly mixed-state) quantum banknote $\rho_\$$ (if pure-state, denoted $\ket{\$}$) along with a classical serial number $s$.

    \item $\Ver^\cO(\rho_\$, s) \in \zo$: The random oracle-aided quantum polynomial-time algorithm $\Ver$ on input a (possibly mixed-state) quantum banknote $\rho_\$$ (if pure-state, denoted $\ket{\$}$) and a classical serial number $s$ outputs $1$ or $0$.
\end{itemize}

\noindent \textbf{Properties.}
\begin{itemize}
    \item \textbf{Perfect Correctness}: For every $\lambda \in \bbN^+$,
    \begin{equation*}
        \Pr_{(\rho_\$, s) \gets \Gen^\cO(1^\lambda)}[\Ver^\cO(\rho_\$, s) = 1] = 1.
    \end{equation*}

    \item \textbf{Unforgeable}: There exists a negligible function $\negl(\cdot)$ such that for every sufficiently large $\lambda \in \bbN^+$ and every random oracle-aided polynomial-size quantum circuit $\cA$,
    \begin{equation*}
        \Pr_{\substack{(\rho_\$, s) \gets \Gen^\cO(1^\lambda) \\ (\rho_{\$, 0}, s_0, \rho_{\$, 1}, s_1) \gets \cA^\cO(\rho_\$, s)}}[s_0 = s_1 = s \: \land \: \Ver^\cO(\rho_{\$, 0}, s_0) = 1 \: \land \: \Ver^\cO(\rho_{\$, 1}, s_1) = 1] \le \negl(\lambda).
    \end{equation*}

    \item \textbf{Unpredictable Serial Numbers}: There exists a negligible function $\negl(\cdot)$ such that for every sufficiently large $\lambda \in \bbN$,
    \begin{equation*}
        \Pr_{\substack{(\rho_\$, s) \gets \Gen^\cO(1^\lambda) \\ (\rho_\$', s') \gets \Gen^\cO(1^\lambda)}}[s = s'] \le \negl(\lambda).
    \end{equation*}
\end{itemize}
\end{definition}

The unpredicable serial number property is w.l.o.g., just as above.

\subsection{Quantum Signature of Knowledge}

\begin{definition}[Quantum SimExt-secure Signature~\cite{CL06}]
\label{def:sig}
Let $\NP$ relation $\cR$ with corresponding language $\cL$ be given such that they can be indexed by a security parameter $\lambda \in \bbN$. Let a message space $\cM$ be given such that it can be indexed by a security parameter $\lambda \in \bbN$.

$(\Setup, \Sign, \Verify)$ is a SimExt-secure quantum signature of knowledge of a witness with respect to $\cL$ and $\cM$ if it has the following syntax and properties.

\noindent \textbf{Syntax.}
The input $1^\lambda$ is left out when it is clear from context.
\begin{itemize}
    \item $(\crs, \td) \gets \Setup(1^\lambda)$: The probabilistic polynomial-time algorithm $\Setup$ on input $1^\lambda$ outputs a common reference string $\crs$ and a trapdoor $\td$.
    \item $\sigma \gets \Sign(1^\lambda, \crs, x, w, m)$: The polynomial-time quantum algorithm $\Sign$ on input a common reference string $\crs$, an instance and witness pair $(x, w) \in \cR_\lambda$, and a message $m \in \cM_\lambda$, outputs a signature $\sigma$.
    \item $\Verify(1^\lambda, \crs, x, m, \sigma) \in \zo$: The polynomial-time quantum algorithm $\Verify$ on input a common reference string $\crs$, an instance $x$, a message $m \in \cM_\lambda$, and a signature $\sigma$, outputs $1$ iff $\sigma$ is a valid signature of $m$ with respect to $\crs$, $\cR_\lambda$, and $x$.
\end{itemize}

\noindent \textbf{Properties.}
\begin{itemize}
    \item {\bf Correctness}: For every sufficiently large $\lambda \in \bbN$, every $(x, w) \in \cR_\lambda$, and every $m \in \cM_\lambda$,
    \begin{equation*}
        \Pr_{\substack{(\crs, \td) \gets \Setup(1^\lambda) \\ \sigma \gets \Sign(\crs, x, w, m)}}[\Verify(\crs, x, m, \sigma) = 1] = 1.
    \end{equation*}

    \item {\bf Simulation}: There exists a quantum polynomial-size circuit simulator $\Sim = (\Sim_0, \Sim_1)$, where $\Sim_1$ ignores its second query input (a witness $w$), and a negligible function $\negl(\cdot)$ such that for every polynomial-size quantum circuit $\cA$
    and every sufficiently large $\lambda \in \bbN$,
    \begin{equation*}
        \left\vert \Pr_{(\crs, \td) \gets \Sim_0(1^\lambda)}[\cA^{\Sim_1(\crs, \td, \cdot, \cdot)}(\crs) = 1] - \Pr_{\substack{(\crs, \td) \gets \Setup(1^\lambda)}}[\cA^{\Sign(\crs, \cdot, \cdot, \cdot)}(\crs) = 1] \right\vert \le \negl(\lambda).
    \end{equation*}

    \item {\bf Extraction}: Let $\Sim = (\Sim_0, \Sim_1)$ be the simulator given by the simulation property. There exists a quantum polynomial-size circuit $\Ext$ and a negligible function $\negl(\cdot)$ such that for every oracle-aided polynomial-size quantum circuit $\cA$ 
    and every sufficiently large $\lambda \in \bbN$,
    \begin{align*}
        \Pr_{\substack{(\crs, \td) \gets \Sim_0(1^\lambda) \\ (x, m, \sigma) \gets \cA^{\Sim_1(\crs, \td, \cdot, \cdot)}(\crs) \\ w \gets \Ext(\crs, \td, x, m, \sigma)}}\left[\Verify(\crs, x, m, \sigma) = 1 \wedge (x, m) \not\in Q \wedge (x, w) \not\in \cR_\lambda \right] \le \negl(\lambda)
    \end{align*}
    where $Q$ is the list of queries from $\cA$ to $\Sim_1$.
\end{itemize}
\end{definition}

%% file: unclonable-nizk.tex
\section{Unclonable Non-Interactive Zero-Knowledge in the CRS Model}
\label{sec:crscons}

\subsection{Simulation-Extractable NIZK}

\label{sec:simext-nizk-crs}

\begin{definition}[Post-Quantum (Quantum) Simulation-Extractable NIZK for $\NP$ in CRS Model]
\label{def:simext-nizk-crs}
Let $\NP$ relation $\cR$ with corresponding language $\cL$ be given such that they can be indexed by a security parameter $\lambda \in \bbN$.

$\Pi = (\Setup, \sP, \sV)$ is a post-quantum (quantum) non-interactive simulation-extractable zero-knowledge argument for $\NP$ in the CRS model if it has the following syntax and properties.

\begin{itemize}
    \item $\Pi$ is a post-quantum (quantum) non-interactive simulation sound, adaptive multi-theorem computational zero-knowledge argument for $\NP$ in the CRS model (\cref{def:simsound-nizk-crs}).
    
    \item {\bf Simulation Extractability.}
    Let $\Sim = (\Sim_0, \Sim_1)$ be the simulator given by the adaptive multi-theorem computational zero-knowledge property. There exists a (quantum) polynomial-time circuit $\Ext$ and a negligible function $\negl(\cdot)$ such that for every oracle-aided polynomial-size quantum circuit $\cA$ and every $\lambda \in \bbN$,
    \begin{equation*}
        \Pr_{\substack{(\crs, \td) \gets \Sim_0(1^\lambda) \\ (x, \pi) \gets \cA^{\Sim_1(\crs, \td, \cdot)}(\crs) \\ w \gets \Ext(\crs, \td, x, \pi)}}[\sV(\crs, x, \pi) = 1 \wedge x \not\in Q \wedge (x, w) \not\in \cR] \le \negl(\lambda),
    \end{equation*}
    where $Q$ is the list of queries from $\cA$ to $\Sim_1$.
    \end{itemize}
\end{definition}

\begin{remark}
    As defined in \cref{def:simext-nizk-crs}, a simulation-extractable zero-knowledge protocol has simulation soundness~\cite{Sahai99,SCOPS01}, is an argument of knowledge, and has adaptive computational soundness (\cref{def:nizk-crs}).
\end{remark}

\begin{figure}[!ht]
\begin{framed}
\centering
\begin{minipage}{1.0\textwidth}
\begin{center}
    \underline{Simulation-Extractable Non-Interactive ZK for $\cL \in \NP$}
\end{center}

\vspace{2mm}

Let $\Pi = (\Setup, \sP, \sV)$ be a non-interactive simulation sound, adaptively multi-theorem computationally zero-knowledge protocol for $\NP$, and $(\Gen, \Enc, \Dec)$ be a post-quantum perfectly correct, IND-CPA secure encryption scheme. Let $\cR$ be the relation with respect to $\cL \in \NP$.

\vspace{2mm}
\noindent {\underline{\textsc{Setup}}$(1^\lambda)$}:
Compute $(\pk, \sk) \gets \Gen(1^\lambda)$, and $(\crs_\Pi, \td_\Pi) \gets \Pi.\Setup(1^\lambda)$. Output $(\crs = (\pk, \crs_\Pi), \td = (\sk, \td_\Pi))$.

\vspace{1mm}
\noindent {\underline{\textsc{Prove}}$(\crs, x, w)$}:
\begin{itemize}
    \item Compute $\ct = \Enc(\pk, w; r)$ for $r$ sampled uniformly at random.
    \item Let $x_\Pi = (\pk, x, \ct)$ be an instance of the following language $\cL_\Pi$:
    \begin{equation*}
    \hspace{-7mm}
        \{ (\pk, x, \ct) \st \exists (w, r) \st \ct = \Enc(\pk, w; r) \:\wedge\: (x, w) \in \cR \}.
    \end{equation*}
    \item Compute proof $\pi_\Pi \gets \Pi.\sP(\crs_\Pi, x_\Pi, (w, r))$ for language $\cL_\Pi$.
    \item Output $\pi = (\ct, \pi_\Pi)$.
\end{itemize}
\noindent {\underline{\textsc{Verify}}$(\crs, x, \pi)$}:
\begin{itemize}
    \item Output $\Pi.\sV(\crs_\Pi, x_\Pi, \pi_\Pi)$.
\end{itemize}

\end{minipage}
\end{framed}
\caption{Unclonable Non-Interactive Quantum Protocol for $\cL \in \NP$}
\label{fig:simext-nizk}
\end{figure}

\begin{theorem}[Post-Quantum Simulation-Extractable NIZK for $\NP$ in the CRS Model]
    \label{thm:simext-nizk-crs}
    Let $\NP$ relation $\cR$ with corresponding language $\cL$ be given.
    
    Let $\Pi = (\Setup, \sP, \sV)$ be a non-interactive post-quantum simulation sound, adaptively multi-theorem computationally zero-knowledge protocol for $\NP$ (\cref{def:simsound-nizk-crs}). Let $(\Gen, \Enc, \Dec)$ be a post-quantum perfectly correct, IND-CPA secure encryption scheme (\cref{def:enc}).
    
    $(\Setup, \sP, \sV)$ as defined in \cref{fig:simext-nizk} will be a non-interactive post-quantum simulation-extractable, adaptively multi-theorem computationally zero-knowledge argument for $\cL$ in the common reference string model (\cref{def:simext-nizk-crs}).
\end{theorem}

\begin{proof}
\noindent \textbf{Perfect Completeness.}
Completeness follows from the perfect completeness of $\Pi$.\\

\noindent \textbf{Adaptively Multi-theorem Computationally Zero-Knowledge.}
Let $\Pi.\Sim = (\Pi.\Sim_0, \Pi.\Sim_1)$ be the adaptive multi-theorem computationally zero-knowledge simulator of $\Pi$. We define $\Sim_0$ with oracle access to $\Pi.\Sim_0$ as follows:
\begin{addmargin}[2em]{2em} 
    \noindent \emph{Input}: $1^\lambda$.

    \noindent \textbf{(1)} Compute $(\pk, \sk) \gets \Gen(1^\lambda)$.
    
    \noindent \textbf{(2)} Send $1^\lambda$ to $\Pi.\Sim_0$. Receive $(\crs_\Pi, \td_\Pi)$ from $\Pi.\Sim_0$.

    \noindent \textbf{(3)} Output $(\crs=(\pk, \crs_\Pi), \td=(\sk, \td_\Pi))$.
\end{addmargin}
We define $\Sim_1$ with oracle access to $\Pi.\Sim_1$
as follows:
\begin{addmargin}[2em]{2em} 
    \noindent \emph{Input}: $\crs=(\pk, \crs_\Pi)$, $\td = (\sk, \td_\Pi)$, $x$.

    \noindent \textbf{(1)} Compute $\ct = \Enc(\pk, 0; r)$ for $r$ sampled uniformly at random.
    
    \noindent \textbf{(2)} Define $x_\Pi = (\pk, x, \ct)$.
    
    \noindent \textbf{(3)} Send $(\crs_\Pi, \td_\Pi, x_\Pi)$ to $\Pi.\Sim_1$. Receive $\pi_\Pi$.

    \noindent \textbf{(4)} Output $\pi = (\ct, \pi_\Pi)$.
\end{addmargin}

Let a polynomial $p(\cdot)$ and an oracle-aided polynomial-size quantum circuit $\cA$ be given such that
\begin{equation}
    \label{eqn:simsound-nizk1}
    \left\vert \Pr_{\substack{(\crs, \td) \gets \Setup(1^\lambda)}}[\cA^{\sP(\crs, \cdot, \cdot)}(\crs) = 1] - \Pr_{\substack{(\crs, \td) \gets \Sim_0(1^\lambda)}}[\cA^{\Sim_1(\crs, \td, \cdot)}(\crs) = 1]\right\vert \ge \frac{1}{p(\lambda)}.
\end{equation}
We will first switch the honest proofs for simulated proofs, using the adaptive multi-theorem zero-knowledge of $\Pi$. Later, we will see how we can switch the encryption of a valid witness to an encryption of $0$, by using the security of the encryption scheme.

Towards this end, we define an intermediary circuit $\cB=(\cB_0, \cB_1)$ which encrypts a valid witness, but provides simulated proofs through $\Pi.\Sim_1$.
We define $\cB_0$ to be equivalent to $\Sim_0$. We define $\cB_1$ with oracle access to $\Pi.\Sim_1$ as follows:
\begin{addmargin}[2em]{2em} 
    \noindent \emph{Input}: $\crs=(\pk, \crs_\Pi)$, $\td = (\sk, \td_\Pi)$, $x$, $w$.

    \noindent \textbf{(1)} Compute $\ct = \Enc(\pk, w; r)$ for $r$ sampled uniformly at random.
    
    \noindent \textbf{(2)} Define $x_\Pi = (\pk, x, \ct)$.
    
    \noindent \textbf{(3)} Send $(\crs_\Pi, \td_\Pi, x_\Pi)$ to $\Pi.\Sim_1$. Receive $\pi_\Pi$.

    \noindent \textbf{(4)} Output $\pi = (\ct, \pi_\Pi)$.
\end{addmargin}

\begin{claim}
    \label{claim:simsound-nizk}
    There exists a negligible function $\negl(\cdot)$ such that for every oracle-aided polynomial-size quantum circuit $\cA$,
    \begin{equation*}
        \left\vert \Pr_{\substack{(\crs, \td) \gets \Setup(1^\lambda)}}[\cA^{\sP(\crs, \cdot, \cdot)}(\crs) = 1] - \Pr_{\substack{(\crs, \td) \gets \cB_0(1^\lambda)}}[\cA^{\cB_1(\crs, \td, \cdot, \cdot)}(\crs) = 1]\right\vert \le \negl(\lambda).
    \end{equation*}
\end{claim}

We will later see a proof of \cref{claim:simsound-nizk}. For now, assuming that this claim holds, by \cref{eqn:simsound-nizk}, this claim, and a union bound, there exists a polynomial $p'(\cdot)$ such that
\begin{equation*}
    \label{eqn:simsound-nizk2}
    \left\vert \Pr_{\substack{(\crs, \td) \gets \cB_0(1^\lambda)}}[\cA^{\cB_1(\crs, \td, \cdot, \cdot)}(\crs) = 1] - \Pr_{\substack{(\crs, \td) \gets \Sim_0(1^\lambda)}}[\cA^{\Sim_1(\crs, \td, \cdot)}(\crs) = 1] \right\vert \ge \frac{1}{p'(\lambda)}.
\end{equation*}

We define a series of intermediary hybrids starting from encrypting all real witnesses to encrypting all zeros. The first intermediary hybrid switches the encryption sent in the last query from an encryption of a witness to an encryption of $0$. We continue switching the encryption in the second to last query and so on, until we've switched the first proof that the adversary makes.

Let $q(\cdot)$ be a polynomial denoting the maximum number of queries that $\cA$ makes.
By a union bound and \cref{eqn:simsound-nizk}, there must exist a hybrid indexed by $i$ (where we switch the ciphertext in the $i$th proof from encrypting a witness to encrypting $0$) where $\cA$ first distinguishes between the two ciphertexts with advantage $1/(p'(\lambda) q(\lambda))$. That is,
\begin{equation}
    \label{eqn:simsound-nizk}
    \left\vert \Pr_{\substack{(\crs, \td) \gets \Setup(1^\lambda)}}[\cA^{\Sim_1^{(i+1)}(\crs, \cdot, \cdot)}(\crs) = 1] - \Pr_{\substack{(\crs, \td) \gets \Setup(1^\lambda)}}[\cA^{\Sim_1^{(i)}(\crs, \td, \cdot)}(\crs) = 1]\right\vert \ge \frac{1}{p'(\lambda)q(\lambda)}.
\end{equation}
where $\Sim^{(j)}_1$ is a stateful algorithm which sends real proofs for the first $j-1$ queries and sends simulated proofs for the remaining queries.

We can use $\cA$ to define a reduction that breaks the IND-CPA security of the encryption scheme as follows:
\begin{addmargin}[2em]{2em} 
    \noindent \underline{Reduction}: to IND-CPA of encryption scheme given oracle access to $\cA$, $\Sim_0$, and $\Sim_1$.

    \noindent \textit{Hardwired with}: $i$.

    \noindent \textbf{(1)} Compute $(\pk, \sk) \gets \Gen(1^\lambda)$.

    \noindent \textbf{(2)} Compute $(\crs_\Pi, \td_\Pi) \gets \Pi.\Sim_0(1^\lambda)$.

    \noindent \textbf{(3)} Define $\crs = (\pk, \crs_\Pi)$ and $\td = (\sk, \td_\Pi)$.
    
    \noindent \textbf{(4)} Send $\crs$ to $\cA$. 
    
    \noindent \textbf{(5)} On the first $i-1$ queries $(x, w)$ from $\cA$: send $\pi \gets \cB_0(\crs, x, w)$ to $\cA$.

    \noindent \textbf{(6)} On the $i$th query $(x, w)$ from $\cA$: send $(w, 0)$ to the challenger, receive $\ct$ from the challenger, define $x_\Pi = (\pk, x, \ct)$, send $(\crs_\Pi, \td_\Pi, x_\Pi)$ to $\Pi.\Sim_1$, receive $\pi_\Pi$ from $\Pi.\Sim_1$, and send $\pi=(\ct, \pi_\Pi)$ to $\cA$.
    
    \noindent \textbf{(7)} On any queries $(x, w)$ after the $i$th: send $\pi \gets \Sim_1(\crs, \td, x)$ to $\cA$.

    \noindent \textbf{(8)} Output the result of $\cA$.
\end{addmargin}
The view of $\cA$ matches that of $\Sim_1^{(i+1)}$ or $\Sim_1^{(i)}$. As such, this reduction should have the same advantage at breaking the IND-CPA security of the encryption scheme. We reach a contradiction. Now, all that remains to prove that our earlier claim holds.

\begin{proof}[Proof of \cref{claim:simsound-nizk}]
    Let a polynomial $p(\cdot)$ and an oracle-aided polynomial-size quantum circuit $\cA$ be given such that
    \begin{equation*}
        \left\vert \Pr_{\substack{(\crs, \td) \gets \Setup(1^\lambda)}}[\cA^{\sP(\crs, \cdot, \cdot)}(\crs) = 1] - \Pr_{\substack{(\crs, \td) \gets \cB_0(1^\lambda)}}[\cA^{\cB_1(\crs, \td, \cdot, \cdot)}(\crs) = 1]\right\vert \ge \frac{1}{p(\lambda)}.
    \end{equation*}
    We define a reduction to the multi-theorem zero-knowledge property of $\Pi$ as follows:
    \begin{addmargin}[2em]{2em} 
        \noindent \underline{Reduction}: to multi-theorem zero-knowledge of $\Pi$ given oracle access to $\cA$.
    
        \noindent \textbf{(1)} Compute $(\pk, \sk) \gets \Gen(1^\lambda)$.
    
        \noindent \textbf{(2)} Receive (real or simulated) $\crs_\Pi$ from the challenger.
    
        \noindent \textbf{(3)} Send $\crs = (\pk, \crs_\Pi)$ to $\cA$. 
        
        \noindent \textbf{(4)} On query $(x, w)$ from $\cA$: compute $\ct = \Enc(\pk, w; r)$ for $r$ samples uniformly at random, send $x_\Pi = (\pk, x, \ct)$ to the challenger, receive (real or simulated) $\pi_\Pi$ from the challenger, send $\pi = (\ct, \pi_\Pi)$ to $\cA$.
    
        \noindent \textbf{(5)} Output the result of $\cA$.
    \end{addmargin}
    The view of $\cA$ matches that of $\Setup$ and $\sP$ or $\cB_0$ and $\cB_1$. As such, this reduction should have the same advantage at breaking the multi-theorem zero-knowledge property of $\Pi$. We reach a contradiction, hence our claim must be true.
\end{proof}

This concludes our proof. Hence our protocol must be multi-theorem zero-knowledge.\\

\noindent \textbf{Simulation Extractable.}
Let $\Pi.\Sim = (\Pi.\Sim_0, \Pi.\Sim_1)$ be the adaptive multi-theorem computationally zero-knowledge simulator of $\Pi$.
Let $\Sim = (\Sim_0, \Sim_1)$ be the simulator, with oracle access to $\Pi.\Sim$, as defined in the proof that \cref{fig:simext-nizk} is adaptive multi-theorem computational zero-knowledge.
We define $\Ext$ as follows:
\begin{addmargin}[2em]{2em} 
    \noindent \emph{Input}: $\crs=(\pk, \crs_\Pi)$, $\td = (\sk, \td_\Pi)$, $x$, $\pi = (\ct, \pi_\Pi)$.

    \noindent \textbf{(1)} Output $\Dec(\sk, \ct)$ as $w$.
\end{addmargin}

Let a polynomial $p(\cdot)$ and an oracle-aided polynomial-size quantum circuit $\cA$ be given such that
\begin{equation*}
    \Pr_{\substack{(\crs, \td) \gets \Sim_0(1^\lambda) \\ (x, \pi) \gets \cA^{\Sim_1(\crs, \td, \cdot)}(\crs) \\ w \gets \Ext(\crs, \td, x, \pi)}}[\sV(\crs, x, \pi) = 1 \wedge x \not\in Q \wedge (x, w) \not\in \cR] \ge \frac{1}{p(\lambda)},
\end{equation*}
where $Q$ is the list of queries from $\cA$ to $\Sim_1$. Since $\sV$ accepts the output of $\cA$, then $\Pi.\sV$ must accept $(\crs_\Pi, x_\Pi, \pi_\Pi)$. Since $x \not\in Q$, then $x_\Pi$ which contains $x$ must not have been sent as a query to $\Pi.\Sim_1$. 
By the definition of $\Ext$ and the perfect correctness of the encryption scheme, $x_\Pi \not\in \cL_\Pi$. Hence, we have that
\begin{equation*}
    \Pr_{\substack{(\crs, \td) \gets \Sim_0(1^\lambda) \\ (x, \pi) \gets \cA^{\Sim_1(\crs, \td, \cdot)}(\crs) \\ w \gets \Ext(\crs, \td, x, \pi)}}[\Pi.\sV(\crs_\Pi, x_\Pi, \pi_\Pi) = 1 \wedge x_\Pi \not\in Q_\Pi \wedge x_\Pi \not\in \cL_\Pi] \ge \frac{1}{p(\lambda)},
\end{equation*}
where $Q_\Pi$ is the list of queries, originating from $\cA$, that $\Sim_1$ makes to $\Pi.\Sim_1$.
We define a reduction to the simulation soundness property of $\Pi$ as follows:
\begin{addmargin}[2em]{2em} 
    \noindent \underline{Reduction}: to simulation soundness of $\Pi$ given oracle access to $\cA$.

    \noindent \textbf{(1)} Compute $(\pk, \sk) \gets \Gen(1^\lambda)$.

    \noindent \textbf{(2)} Receive $\crs_\Pi$ from the challenger.

    \noindent \textbf{(3)} Send $\crs = (\pk, \crs_\Pi)$ to $\cA$. 
    
    \noindent \textbf{(4)} On query $x$ from $\cA$: compute $\ct = \Enc(\pk, 0; r)$ for $r$ samples uniformly at random, send $x_\Pi = (\pk, x, \ct)$ to the challenger, receives $\pi_\Pi$ from the challenger, send $\pi = (\ct, \pi_\Pi)$ to $\cA$.

    \noindent \textbf{(5)} Receive $(x, \pi=(\ct, \pi_\Pi))$ from $\cA$. Define $x_\Pi = (\pk, x, \ct)$.

    \noindent \textbf{(6)} Output $(x_\Pi, \pi_\Pi)$.
\end{addmargin}
The view of $\cA$ matches that of $\Sim_0$ and $\Sim_1$. As such, this reduction should have the same advantage at breaking the simulation soundness property of $\Pi$. We reach a contradiction, hence our protocol must be simulation extractable.
\end{proof}

\begin{corollary}[Post-Quantum Simulation-Extractable NIZK for $\NP$ in the CRS Model]
    \label{cor:simext-nizk-crs}
    Assuming the polynomial quantum hardness of LWE, there exists a simulation-extractable, adaptively multi-theorem computationally zero-knowledge argument for $\NP$ in the common reference string model (\cref{def:simext-nizk-crs}).
\end{corollary}

\begin{proof}
    This follows from \cref{cor:simsound-nizk} and \cref{thm:simext-nizk-crs}.
\end{proof}

\subsection{Unclonability Definitions}
\label{sec:unc-defs-crs}

We consider two definitions of unclonability for NIZKs. The first one, motivated by simplicity, informally guarantees that no adversary given honestly proofs for ``hard'' instances is able to output more than one accepting proof for the same instance. 

\begin{definition}[(Quantum) Hard Distribution]
    \label{def:hard-dist}
    Let an $\NP$ relation $\cR$ be given.
    $(\cX, \cW)$ is a (quantum) hard distribution over $\cR$ if the following properties hold.

    \begin{itemize}
        \item \textbf{Syntax.} 
        $(\cX, \cW)$ is indexable by a security parameter $\lambda\in \bbN$. For every choice of $\lambda \in \bbN$, the support of $(\cX_\lambda, \cW_\lambda)$ is over instance and witness pairs $(x, w)$ such that $x \in \cL$, $|x| = \lambda$, and $(x, w) \in \cR$. 

        \item \textbf{Hardness.}
        For every polynomial-sized (quantum) circuit family $\cA = \{\cA_\lambda\}_{\lambda \in \bbN}$,
        \begin{equation*}
            \Pr_{(x,w) \leftarrow (\cX_\lambda,\cW_\lambda)}[(x, \cA_\lambda(x)) \in \cR] \leq \mathsf{negl}(\lambda).
        \end{equation*}
    \end{itemize}
\end{definition}

\begin{definition} (Unclonable Security for Hard Instances).
\label{def:uncnizk-alt}
A proof  $(\mathsf{Setup},\sP,\sV)$ satisfies unclonable security for a language $\mathcal{L}$ with corresponding relation $\cR_\cL$ if
for every polynomial-sized quantum circuit family $\{C_\lambda\}_{\lambda \in \mathbb{N}}$,
and for every hard distribution $\{\mathcal{X}_\lambda,\mathcal{W}_\lambda\}_{\lambda \in \bbN}$ over $\cR_\cL$,
there exists a negligible function $\negl(\cdot)$ such that
for every $\lambda \in \bbN$,
\begin{equation*}
    \Pr_{(x,w) \leftarrow (\mathcal{X}_\lambda,\mathcal{W}_\lambda)}\Bigg[
    \sV(\mathsf{crs},x,\pi_1) = 1 \bigwedge
    \sV(\mathsf{crs},x,\pi_2) = 1
    \Bigg|
    \substack{\mathsf{crs} \leftarrow \mathsf{Setup}(1^\lambda)\\
    \pi \leftarrow \sP(\mathsf{crs},x,w)\\
    \pi_1, \pi_2 \leftarrow C_\lambda(x, \pi)
    }
    \Bigg]
    \leq \negl(\lambda).
\end{equation*}
\end{definition}

We will now strengthen this definition to consider a variant where 
from any adversary $\cA$ that on input a single proof of membership of $x \in \cL$ outputs two proofs for $x$, we can extract a valid witness $w$ for $x$ with high probability.
In fact, we can further generalize this definition to a setting where the adversary obtains an even larger number (say $k-1$) input proofs on instances $x_1, \ldots, x_{k-1}$, and outputs $k$ or more proofs. Then we require the extraction of an NP witness corresponding to any proofs that are {\em duplicated} (i.e. two or more proofs w.r.t. the same instance $x_i \in \{x_1, \ldots, x_{k-1}\}$). 
We write this definition below.

\begin{definition}[$(k-1) \text{-to-} k$-Unclonable Extractable NIZK]
\label{def:unizk-crs}
Let security parameter $\lambda \in \bbN$ and $\NP$ relation $\cR$ with corresponding language $\cL$ be given.
Let $\Pi = (\Setup, \sP, \sV)$ be given such that $\Setup, \sP$ and $\sV$ are $\poly(\lambda)$-size quantum algorithms. We have that for any $(x, w) \in \cR$, $(\crs, \td)$ is the output of $\Setup$ on input $1^\lambda$, $\sP$ receives an instance and witness pair $(x, w)$ along with $\crs$ as input and outputs $\pi$, and $\sV$ receives an instance $x$, $\crs$, and proof $\pi$ as input and outputs a value in $\zo$.

$\Pi$ is a non-interactive $(k-1) \text{-to-} k$-unclonable zero-knowledge quantum protocol for language $\cL$ if the following holds:
\begin{itemize}
    \item $\Pi$ is a quantum non-interactive zero-knowledge protocol for language $\cL$ (\cref{def:nizk-crs}).

    \item {\bf $(k-1) \text{-to-} k$-Unclonable with Extraction}: There exists an oracle-aided polynomial-size quantum circuit $\cE$ such that for every polynomial-size quantum circuit $\cA$, for every tuple of $k-1$ instance-witness pairs $(x_1, \omega_1), \ldots, (x_{k-1}, \omega_{k-1}) \in \cR$, for every instance $x$,
    if there exists a polynomial $p(\cdot)$ such that
    \begin{equation*}
        \Pr_{\substack{(\crs, \td) \gets \Setup(1^\lambda) \\ \forall \iota \in [k-1], \: \pi_\iota \gets \sP(\crs, x_\iota, w_\iota) \\ \{\widetilde{{x}_\iota},\widetilde{\pi_\iota}\}_{\iota \in [k]} \gets \cA(\crs, \{x_\iota, \pi_\iota\}_{\iota \in [k-1]})}}
        \left[
        \begin{array}{cc}
        & \exists~ \cJ \subseteq \{j:\widetilde{x}_j = x\} \text{ s.t. } |\cJ| > |\{i:x_i = x\}| \\
        & \text{ and } {\forall \iota \in \cJ}, \sV(\crs, x, \widetilde{\pi_\iota}) = 1 
        \end{array}
        \right]
        \geq \frac{1}{p(\lambda)},
    \end{equation*}
    then there is also a polynomial $q(\cdot)$ such that
    \begin{equation*}
        \Pr_{w \gets \cE^\cA(x_1, \ldots, x_{k-1}, x)}\left[(x, w) \in \cR \right] \geq \frac{1}{q(\lambda)}.
    \end{equation*}
\end{itemize}
\end{definition}

We observe in \cref{def:unizk-crs} that we can generically boost the extractor’s success probability to $1 - \negl(\lambda)$ with respect to a security parameter $\lambda$.

\begin{definition}[$(k-1) \text{-to-} k$-Unclonable Strong-Extractable NIZK]
\label{def:unizk-crs2}
Let security parameter $\lambda \in \bbN$ and $\NP$ relation $\cR$ with corresponding language $\cL$ be given.
Let $\Pi = (\Setup, \sP, \sV)$ be given such that $\Setup, \sP$ and $\sV$ are $\poly(\lambda)$-size quantum algorithms. We have that for any $(x, w) \in \cR$, $(\crs, \td)$ is the output of $\Setup$ on input $1^\lambda$, $\sP$ receives an instance and witness pair $(x, w)$ along with $\crs$ as input and outputs $\pi$, and $\sV$ receives an instance $x$, $\crs$, and proof $\pi$ as input and outputs a value in $\zo$.

$\Pi$ is a non-interactive $(k-1) \text{-to-} k$-unclonable zero-knowledge quantum protocol for language $\cL$ if the following holds:
\begin{itemize}
    \item $\Pi$ is a quantum non-interactive zero-knowledge protocol for language $\cL$ (\cref{def:nizk-crs}).

    \item {\bf $(k-1) \text{-to-} k$-Unclonable with Strong-Extraction}: There exists an oracle-aided polynomial-size quantum circuit $\cE$ such that for every polynomial-size quantum circuit $\cA$ with non-uniform quantum advice $\aux$, for every tuple of $k-1$ instance-witness pairs $(x_1, \omega_1), \ldots, (x_{k-1}, \omega_{k-1}) \in \cR$, for every instance $x$ 
    if there is a polynomial $p(\cdot)$ where 
    \begin{equation*}
        \Pr_{\substack{(\crs, \td) \gets \Setup(1^\lambda) \\ \forall \iota \in [k-1], \: \pi_\iota \gets \sP(\crs, x_\iota, w_\iota) \\ \{\widetilde{x_\iota}, \widetilde{\pi_\iota}\}_{\iota \in [k]} \gets \cA(\crs, \{x_\iota, \pi_\iota\}_{\iota \in [k-1]}, \aux)}}\left[
        \begin{array}{cc}
        & \exists~ \cJ \subseteq \{j:\widetilde{x}_j = x\} \text{ s.t. } |\cJ| > |\{i:x_i = x\}| \\
        & \text{ and } {\forall \iota \in \cJ}, \sV(\crs, x, \widetilde{\pi_\iota}) = 1 
        \end{array}
        \right]
        \geq \frac{1}{p(\lambda)},
    \end{equation*}
    then there is also a polynomial $\poly(\cdot)$ and a negligible function $\negl(\cdot)$ such that
    \begin{equation*}
        \Pr_{w \gets \cE^\cA(x_1, \ldots, x_{k-1}, x, \aux^{\tensor \poly(\lambda)})}\left[(x, w) \in \cR \right] \geq 1 - \negl(\lambda).
    \end{equation*}
\end{itemize}
\end{definition}

We describe two useful lemmas to compare the above definitions.

\begin{lemma}
    \label{lem:uncdef}
    Let $\Pi = (\Setup, \sP, \sV)$ be a $1 \text{-to-} 2$-unclonable with extraction, non-interactive zero-knowledge quantum protocol (\cref{def:unizk-crs}). Then, $\Pi$ satisfies \cref{def:uncnizk-alt}.
\end{lemma}

For a proof of \cref{lem:uncdef}, we refer to \cref{app:defs-reduct}.

\begin{lemma}
    \label{lem:uncdef2}
    Let $\Pi = (\Setup, \sP, \sV)$ be a $(k-1) \text{-to-} k$-unclonable with extraction, non-interactive zero-knowledge quantum protocol (\cref{def:unizk-crs}). Then, $\Pi$ satisfies \cref{def:unizk-crs2}.
\end{lemma}

\begin{proof}[Proof Sketch]
    Given an extractor $\Pi.\cE$ from \cref{def:unizk-crs}, we define a new extractor $\cE$. According to \cref{def:unizk-crs2}, $\cE$ receives multiple copies of the adversary's quantum advice string $\aux$. $\cE$ runs $\Pi.\cE$ on the adversary multiple times, each time using a fresh copy of $\aux$. 
    
    Formally, for every $\cA$ with $\aux$, $(x_1, w_1), \ldots, (x_{k-1}, w_{k-1}) \in \cR$, $x$, polynomial $p(\cdot)$, and polynomial $q(\cdot)$ such that
    \begin{align*}
        &\Pr_{\substack{(\crs, \td) \gets \Setup(1^\lambda) \\ \forall \iota \in [k-1], \: \pi_\iota \gets \sP(\crs, x_\iota, w_\iota) \\ \{\widetilde{x_\iota}, \widetilde{\pi_\iota}\}_{\iota \in [k]} \gets \cA(\crs, \{x_\iota, \pi_\iota\}_{\iota \in [k-1]}, \aux)}}\left[\begin{array}{cc}
        & \exists~ \cJ \subseteq \{j:\widetilde{x}_j = x\} \text{ s.t. } |\cJ| > |\{i:x_i = x\}| \\
        & \text{ and } {\forall \iota \in \cJ}, \sV(\crs, x, \widetilde{\pi_\iota}) = 1 
        \end{array}\right] \geq \frac{1}{p(\lambda)}, \text{ and}\\
        &\Pr_{w \gets \Pi.\cE^\cA(x_1, \ldots, x_{k-1}, x)}\left[(x, w) \in \cR \right] \ge \frac{1}{q(\lambda)},
    \end{align*}
    there exists a polynomial $\poly(\cdot)$ and a negligible function $\negl(\cdot)$ such that the extractor $\cE$ will succeed with probability
    \begin{align*}
        &\Pr_{w \gets \cE^\cA(x_1, \ldots, x_{k-1}, x, \aux^{\tensor \poly(\lambda)})}\left[(x, w) \in \cR \right] \\
        &\geq \left(\Pr_{w \gets \Pi.\cE^{\cA(\cdot, \cdot, \aux)}(x_1, \ldots, x_{k-1}, x)}\left[(x, w) \in \cR \right] \right)^{\poly(\lambda)} \\
        &\ge 1 - \left(1 - \frac{1}{q(\lambda)}\right)^{\poly(\lambda)} \ge 1 - \negl(\lambda).
    \end{align*}
    Thus, $\cE$ satisfies \cref{def:unizk-crs2}.
\end{proof}

From the above lemmas, we conclude that \cref{def:unizk-crs} is the strongest definition. In the following sections, we construct a protocol that satisfies \cref{def:unizk-crs}.

\subsection{Unclonable NIZK Implies Public-Key Quantum Money Mini-Scheme}

\begin{figure}[!ht]
\begin{framed}
\centering
\begin{minipage}{1.0\textwidth}
\begin{center}
    \underline{Public-Key Quantum Money Mini-Scheme}
\end{center}

\vspace{2mm}

Let $(\cX, \cW)$ be a hard distribution over a language $\cL \in \NP$. Let $\Pi = (\Setup, \sP, \sV)$ be an unclonable non-interactive zero-knowledge protocol for $\cL$.

\vspace{2mm}
\noindent {\underline{\textsc{Gen}}$(1^\lambda)$}:
Sample a hard instance-witness pair $(x, w) \gets (\cX, \cY)$, a common reference string $(\crs, \td) \gets \Setup(1^\lambda, x)$, and a proof $\pi \gets \sP(\crs, x, w)$. Output $(\rho_\$ = \pi, s = (\crs, x))$.

\vspace{1mm}
\noindent {\underline{\textsc{Verify}}$(\rho_\$, s)$}:
Parse $\rho_\$ = \pi$ and $s = (\crs, x)$. Output $\sV(\crs, x, \pi)$.

\end{minipage}
\end{framed}
\caption{Public-Key Quantum Money Mini-Scheme from an Unclonable Non-Interactive Quantum Protocol}
\label{fig:qmoney-from-unizk}
\end{figure}

\begin{theorem}
\label{thm:unizk-implies-qmoney}
Let $(\cX, \cW)$ be a hard distribution over a language $\cL \in \NP$. Let $\Pi = (\Setup, \sP, \sV)$ satisfy \cref{def:uncnizk-alt}.
Then $(\Setup, \sP, \sV)$ implies a public-key quantum money mini-scheme (\cref{def:qmoney}) as described in \cref{fig:qmoney-from-unizk}.
\end{theorem}

\begin{proof}
    \noindent \textbf{Perfect Correctness.}
    This follows directly from the perfect completeness of $\Pi$.

    \noindent \textbf{Unforgeability.}
    Let $p(\cdot)$ be a polynomial and $\cA$ be a quantum polynomial-time adversary such that for an infinite number of $\lambda \in \bbN^+$,
    \begin{equation*}
        \Pr_{\substack{(\rho_\$, s) \gets \Gen(1^\lambda) \\ (\rho_{\$, 0}, s_0, \rho_{\$, 1}, s_1) \gets \cA(\rho_\$, s)}}[s_0 = s_1 = s \: \land \: \Ver(\rho_{\$, 0}, s_0) = 1 \: \land \: \Ver(\rho_{\$, 1}, s_1) = 1] \ge \frac{1}{p(\lambda)}.
    \end{equation*}

    We construct a reduction that breaks the uncloneability definition.
    The challenger samples a hard instance-witness pair $(x, w) \gets (\cX, \cY)$, a common reference string with a trapdoor $(\crs, \td) \gets \Setup(1^\lambda, x)$, and a proof $\pi \gets \sP(\crs, x, w)$. The challenger then forwards $(\crs, x, \pi)$ to the reduction. The reduction then sets $\rho_\$ = \pi$ and $s = (\crs, x)$. The reduction sends $(\rho_\$, s)$ to the adversary $\cA$ who returns back $(\rho_{\$, 0}, s_0, \rho_{\$, 1}, s_1)$. The reduction then parses and sets $\pi_i = \rho_{\$, i}$ for $i \in \{0, 1\}$. The reduction then sends $\pi_0$ and $\pi_1$ back to the challenger.

    When the serial numbers are the same, $s = s_0 = s_1$, we have that the common reference string and instance will be the same for all the proofs $\pi, \pi_0, \pi_1$. The quantum money state can be parsed as the proof as shown in the construction. When the verification algorithm of the quantum money algorithm accepts both quantum money states $\rho_{\$, 0}$ and $\rho_{\$, 1}$ with respect to $s$, we know that that $\sV$ would accept both proofs $\pi_0$ and $\pi_1$ with respect to $(\crs, x)$. As such, we will have that the advantage that $\cA$ has at breaking the unforgeability of our quantum money scheme directly translates to the advantage of the reduction at breaking the uncloneability of $\Pi$.
\end{proof}

\subsection{Construction and Analysis of Unclonable-Extractable NIZK in CRS Model}
\label{sec:unizk-crs}

\begin{figure}[!ht]
\begin{framed}
\centering
\begin{minipage}{1.0\textwidth}
\begin{center}
    \underline{Unclonable Non-Interactive ZK for $\cL \in \NP$}
\end{center}

\vspace{2mm}

Let $\Pi = (\Setup, \sP, \sV)$ be a non-interactive simulation-extractable, adaptively multi-theorem computationally zero-knowledge protocol for $\NP$, $\Com$ be a post-quantum perfectly binding, computationally hiding commitment scheme, and $(\NoteGen, \Ver)$ be a public-key quantum money scheme. Let $\cR$ be the relation with respect to $\cL \in \NP$.

\vspace{2mm}
\noindent {\underline{\textsc{Setup}}$(1^\lambda)$}:
Sample the common reference string $(\crs_\Pi, \td_\Pi) \gets \Pi.\Setup(1^\lambda)$, and $s^*, r^*$ uniformly at random.
Define $c = \Com(s^*; r^*)$ and output $(\crs = (\crs_\Pi, c), \td = \td_\Pi)$.

\vspace{1mm}
\noindent {\underline{\textsc{Prove}}$(\crs, x, w)$}:
\begin{itemize}
    \item Compute a quantum note and associated serial number $(\rho_\$, s) \leftarrow \NoteGen$.
    \item Let $x_\Pi = (c, x, s)$ be an instance of the following language $\cL_\Pi$:
    \begin{equation*}
    \hspace{-7mm}
        \{ (c, x, s) \st \exists z \st (x, z) \in \cR \:\vee\: c = \Com(s;z) \}.
    \end{equation*}
    \item Compute proof $\pi_\Pi \gets \Pi.\sP(\crs_\Pi, x_\Pi, w)$ for language $\cL_\Pi$.
    \item Output $\pi = (\rho_\$,s,\pi_\Pi)$.
\end{itemize}
\noindent {\underline{\textsc{Verify}}$(\crs, x, \pi)$}:
\begin{itemize}
    \item Check that $\Ver(\rho_\$, s)$ outputs $1$ and that $\Pi.\sV(\crs_\Pi, x_\Pi, \pi_\Pi)$ outputs $1$.
    \item If both checks pass, output $1$. Otherwise, output $0$.
\end{itemize}

\end{minipage}
\end{framed}
\caption{Unclonable Non-Interactive Quantum Protocol for $\cL \in \NP$}
\label{fig:unizk-crs}
\end{figure}

\begin{theorem}
\label{thm:main-crs}
Let $k(\cdot)$ be a polynomial.
Let $\NP$ relation $\cR$ with corresponding language $\cL$ be given.

Let $(\NoteGen, \Ver)$ be a public-key quantum money mini-scheme (\cref{def:qmoney}) 
and $\Com$ be a post-quantum commitment scheme (\cref{def:com}).
Let $\Pi = (\Setup, \sP, \sV)$ be a non-interactive post-quantum simulation-extractable, adaptive multi-theorem computational zero-knowledge protocol for $\NP$ (\cref{def:simext-nizk-crs}).

$(\Setup, \sP, \sV)$ as defined in \cref{fig:unizk-crs} will be a non-interactive quantum simulation-extractable, adaptive multi-theorem computationally zero-knowledge, and $(k-1)$-to-$k$-unclonable argument with extraction protocol for $\cL$ in the common reference string model (\cref{def:unizk-crs}).
\end{theorem}

\begin{proof}
\noindent \textbf{Perfect Completeness.}
Completeness follows from perfect correctness of the public key quantum money scheme, and perfect completeness of $\Pi$.

\noindent \textbf{Adaptive Multi-Theorem Computational Zero-Knowledge.}
Let $\Pi.\Sim = (\Pi.\Sim_0, \Pi.\Sim_1)$ be the adaptive multi-theorem computationally zero-knowledge simulator of $\Pi$. We define $\Sim_0$ with oracle access to $\Pi.\Sim_0$ as follows:
\begin{addmargin}[2em]{2em} 
    \noindent \emph{Input}: $1^\lambda$.

    \noindent \textbf{(1)} Send $1^\lambda$ to $\Pi.\Sim_0$. Receive $(\crs_\Pi, \td_\Pi)$ from $\Pi.\Sim_0$.

    \noindent \textbf{(2)} Sample $s^*, r^*$ uniformly at random. Define $c = \Com(s^*; r^*)$.

    \noindent \textbf{(3)} Output $\crs = (\crs_\Pi, c)$ and $\td = \td_\Pi$.
\end{addmargin}
We define $\Sim_1$ with oracle access to $\Pi.\Sim_1$ as follows:
\begin{addmargin}[2em]{2em} 
    \noindent \emph{Input}: $\crs=(\crs_\Pi, c)$, $\td = \td_\Pi$, $x$.

    \noindent \textbf{(1)} Sample $(\rho_\$, s) \gets \NoteGen(1^\lambda)$.

    \noindent \textbf{(2)} Define $x_\Pi = (c, x, s)$. Send $(\crs_\Pi, \td_\Pi, x_\Pi)$ to $\Pi.\Sim_1$. Receive $\pi_\Pi$ from $\Pi.\Sim_1$.

    \noindent \textbf{(3)} Output $\pi = (\rho_\$, s, \pi_\Pi)$.
\end{addmargin}

Let a polynomial $p(\cdot)$ and an oracle-aided polynomial-size quantum circuit $\cA$ be given such that
\begin{equation*}
    \left\vert \Pr_{\substack{(\crs, \td) \gets \Setup(1^\lambda)}}[\cA^{\sP(\crs, \cdot, \cdot)}(\crs) = 1] - \Pr_{\substack{(\crs, \td) \gets \Sim_0(1^\lambda)}}[\cA^{\Sim_1(\crs, \td, \cdot)}(\crs) = 1]\right\vert \ge \frac{1}{p(\lambda)}.
\end{equation*}
We define a reduction to the multi-theorem zero-knowledge property of $\Pi$ as follows:
\begin{addmargin}[2em]{2em} 
    \noindent \underline{Reduction}: to zero-knowledge of $\Pi$ given oracle access to $\cA$.

    \noindent \textbf{(1)} Receive (real or simulated) $\crs_\Pi$ from the challenger.

    \noindent \textbf{(2)} Sample $s^*, r^*$ uniformly at random. Define $c = \Com(s^*; r^*)$ and $\crs = (\crs_\Pi, c)$.

    \noindent \textbf{(3)} Send $\crs$ to $\cA$. 
    
    \noindent \textbf{(4)} On query $(x, w)$ from $\cA$: sample $(\rho_\$, s) \gets \NoteGen(1^\lambda)$, define $x_\Pi = (c, x, s)$ and $w_\Pi = w$, send $(x_\Pi, w_\Pi)$ to the challenger, receive (real or simulated) $\pi_\Pi$ from the challenger, define $\pi = (\rho_\$, s, \pi_\Pi)$, send $\pi$ to $\cA$.

    \noindent \textbf{(5)} Output the result of $\cA$.
\end{addmargin}
The view of $\cA$ matches that of our protocol in \cref{fig:unizk-crs} or $\Sim_0$ and $\Sim_1$. As such, this reduction should have the same advantage at breaking the adaptive multi-theorem computational zero-knowledge property of $\Pi$. We reach a contradiction, hence our protocol must be multi-theorem zero-knowledge.\\

\noindent \textbf{Simulation-Extractability.}
Let $\Pi.\Sim = (\Pi.\Sim_0, \Pi.\Sim_1)$ be the adaptive multi-theorem computationally zero-knowledge simulator of $\Pi$.
Let $\Pi.\Ext$ be the simulation-extraction extractor of $\Pi$ with respect to $\Pi.\Sim$.
Let $\Sim = (\Sim_0, \Sim_1)$ be the simulator, with oracle access to $\Pi.\Sim$, as defined in the proof that \cref{fig:unizk-crs} is adaptive multi-theorem computational zero-knowledge.
We define $\Ext$ with oracle access to $\Pi.\Ext$ as follows:
\begin{addmargin}[2em]{2em} 
    \noindent \emph{Input}: $\crs=(\crs_\Pi, c)$, $\td = \td_\Pi$, $x$, $\pi = (\rho_\$, s, \pi_\Pi)$.

    \noindent \textbf{(1)} Define $x_\Pi = (c, x, s)$. Send $(\crs_\Pi, \td_\Pi, x_\Pi, \pi_\Pi)$ to $\Pi.\Ext$. Receive $w_\Pi$ from $\Pi.\Ext$.

    \noindent \textbf{(2)} Output $w_\Pi$ as $w$.
\end{addmargin}

Let a polynomial $p(\cdot)$ and an oracle-aided polynomial-size quantum circuit $\cA$ be given such that
\begin{equation*}
    \Pr_{\substack{(\crs, \td) \gets \Sim_0(1^\lambda) \\ (x, \pi) \gets \cA^{\Sim_1(\crs, \td, \cdot)}(\crs) \\ w \gets \Ext(\crs, \td, x, \pi)}}[\sV(\crs, x, \pi) = 1 \wedge x \not\in Q \wedge (x, w) \not\in \cR] \ge \frac{1}{p(\lambda)},
\end{equation*}
where $Q$ is the list of queries from $\cA$ to $\Sim_1$.
Since $\Sim_1$ forwards oracle queries to $\Pi.\Sim_1$ which contain any query it receives from $\cA$, we know that $x_\Pi \not\in Q_\Pi$ where $Q_\Pi$ is the list of queries from $\Sim_1$ to $\Pi.\Sim_1$. Furthermore, since $\sV$ accepts the output $\pi$ from $\cA$, then $\Pi.\sV$ must accept the proof $\pi_\Pi$. As such, we have that
\begin{align}
    \label{eqn:simext}
    \Pr_{\substack{(\crs, \td) \gets \Sim_0(1^\lambda) \\ (x, \pi) \gets \cA^{\Sim_1(\crs, \td, \cdot)}(\crs) \\ w \gets \Ext(\crs, \td, x, \pi)}}[\Pi.\sV(\crs_\Pi, x_\Pi, \pi_\Pi) = 1 \wedge x_\Pi \not\in Q_\Pi \wedge (x, w) \not\in \cR] \ge \frac{1}{p(\lambda)}.
\end{align}
However, we make the following claim which is in direct contradiction with \cref{eqn:simext}.

\begin{claim}
    \label{claim:simext}
    Let $\Ext$ be as defined earlier, in the current proof of simulation-extractability.
    There exists a negligible function $\negl(\cdot)$ such that for every polynomial-size quantum circuit $\cB$,
    \begin{equation*}
        \Pr_{\substack{(\crs, \td) \gets \Sim_0(1^\lambda) \\ (x, \pi) \gets \cB^{\Sim_1(\crs, \td, \cdot)}(\crs) \\ w \gets \Ext(\crs, \td, x, \pi)}}[\Pi.\sV(\crs_\Pi, x_\Pi, \pi_\Pi) = 1 \wedge x_\Pi \not\in Q_\Pi \wedge (x, w) \not\in \cR] \le \negl(\lambda)
    \end{equation*}
    where $Q_\Pi$ is the list of queries forwarded by $\Sim_1$ to $\Pi.\Sim_1$.
\end{claim}

\begin{proof}[Proof of \cref{claim:simext}]
We proceed by contradiction.
Let a polynomial $p(\cdot)$ and an oracle-aided polynomial-size quantum circuit $\cB$ be given such that
\begin{align}
    \label{eqn:simext-B}
    \Pr_{\substack{(\crs, \td) \gets \Sim_0(1^\lambda) \\ (x, \pi) \gets \cB^{\Sim_1(\crs, \td, \cdot)}(\crs) \\ w \gets \Ext(\crs, \td, x, \pi)}}[\Pi.\sV(\crs_\Pi, x_\Pi, \pi_\Pi) = 1 \wedge x_\Pi \not\in Q_\Pi \wedge (x, w) \not\in \cR] \ge \frac{1}{p(\lambda)}
\end{align}
where $Q_\Pi$ is the list of queries forwarded by $\Sim_1$ to $\Pi.\Sim_1$.
Given \cref{eqn:simext-B}, we may be in one of the two following cases: either the extractor $\Pi.\Ext$ extracts $w_\Pi$ from $\cB$ such that $(x_\Pi, w_\Pi) \not\in \cR_\Pi$ (for an infinite set of $\lambda$), or the extractor $\Pi.\Ext$ extracts $w_\Pi$ from $\cB$ such that $(x_\Pi, w_\Pi) \in \cR_\Pi$ (for an infinite set of $\lambda$).
We consider that either of these two scenarios occur with at least $1/(2p(\lambda))$ probability and show that each reaches a contradiction.

\noindent \underline{Scenario One}

Say that (for an infinite set of $\lambda$) the extractor $\Pi.\Ext$ extracts $w_\Pi$ from $\cB$ such that $(x_\Pi, w_\Pi) \not\in \cR_\Pi$ with at least $1/(2p(\lambda))$ probability. Symbolically,
\begin{align}
    \label{eqn:simext1}
    \Pr_{\substack{(\crs, \td) \gets \Sim_0(1^\lambda) \\ (x, \pi) \gets \cB^{\Sim_1(\crs, \td, \cdot)}(\crs) \\ w \gets \Ext(\crs, \td, x, \pi)}}[\Pi.\sV(\crs_\Pi, x_\Pi, \pi_\Pi) = 1 \wedge x_\Pi \not\in Q_\Pi \wedge (x_\Pi, w_\Pi) \not\in \cR_\Pi] \ge \frac{1}{2p(\lambda)}.
\end{align}
By using the advantage of $\cB$ in this game, we can show a reduction that breaks the simulation-extractability of $\Pi$. We will now outline this reduction.
\begin{addmargin}[2em]{2em} 
    \noindent \underline{Reduction}: to simulation-extractability of $\Pi$ given oracle access to $\cB$.

    \noindent \textbf{(1)} Receive $\crs_\Pi$ from the challenger.

    \noindent \textbf{(2)} Sample $s^*, r^*$ uniformly at random. Define $c = \Com(s^*; r^*)$.

    \noindent \textbf{(3)} Define $\crs = (\crs_\Pi, c)$ and $\td = \td_\Pi$. Send $\crs$ to $\cB$.

    \noindent \textbf{(4)} On query $x$ from $\cB$: sample $(\rho_\$, s) \gets \NoteGen(1^\lambda)$, define $x_\Pi = (c, x, s)$, send $x_\Pi$ to the challenger, receive $\pi_\Pi$ from the challenger, define $\pi = (\rho_\$, s, \pi_\Pi)$, and send $\pi$ to $\cB$.

    \noindent \textbf{(5)} Receive $(x, \pi = (\rho_\$, s, \pi_\Pi))$ from $\cB$. Define $x_\Pi = (c, x, s)$.

    \noindent \textbf{(6)} Output $(x_\Pi, \pi_\Pi)$.
\end{addmargin}
Given the event in \cref{eqn:simext1} holds, then the reduction will return an accepting proof $\pi_\Pi$ for an instance $x_\Pi$ which it has not previously queried on and, yet, the extraction $\Pi.\Ext$ will fail.
With advantage $1/(2p(\lambda))$, the reduction will succeed at breaking simulation-extractability of $\Pi$, thus reaching a contradiction.

\noindent \underline{Scenario Two}

Alternatively, say that (for an infinite set of $\lambda$) the extractor $\Pi.\Ext$ extracts $w_\Pi$ from $\cB$ such that $(x_\Pi, w_\Pi) \in \cR_\Pi$ with at least $1/(2p(\lambda))$ probability. In summary, we have that
\begin{align}
    \label{eqn:simext2}
    \Pr_{\substack{(\crs, \td) \gets \Sim_0(1^\lambda) \\ (x, \pi) \gets \cB^{\Sim_1(\crs, \td, \cdot)}(\crs) \\ w \gets \Ext(\crs, \td, x, \pi)}}[\Pi.\sV(\crs_\Pi, x_\Pi, \pi_\Pi) = 1 \wedge x_\Pi \not\in Q_\Pi \wedge (x, w) \not\in \cR \wedge (x_\Pi, w_\Pi) \in \cR_\Pi] \ge \frac{1}{2p(\lambda)}.
\end{align}
Since $\Ext$ outputs $w = w_\Pi$, by the definition of $\cL_\Pi$ and the perfect binding of $\Com$, we must have that $\cB$ has found an opening to the commitment $c$ in the crs, that is that $s = s^*$ and $w_\Pi = r^*$. We can use $\cB$ to break the hiding of the commitment. We will now outline this reduction.
\begin{addmargin}[2em]{2em} 
    \noindent \underline{Reduction}: to hiding of $\Com$ given oracle access to $\cB$.

    \noindent \textbf{(1)} Compute $(\crs_\Pi, \td_\Pi) \gets \Pi.\Sim_0(1^\lambda)$ from the challenger.

    \noindent \textbf{(2)} Sample $s_0, s_1$ uniformly at random. Send $(s_0, s_1)$ to the challenger. Receive $c$.

    \noindent \textbf{(3)} Define $\crs = (\crs_\Pi, c)$ and $\td = \td_\Pi$. Send $\crs$ to $\cB$.

    \noindent \textbf{(4)} On query $x$ from $\cB$: compute $\pi \gets \Sim_1(\crs, \td, x)$, and send $\pi$ to $\cB$.

    \noindent \textbf{(5)} Receive $(x, \pi = (\rho_\$, s, \pi_\Pi))$ from $\cB$.

    \noindent \textbf{(6)} Compute $w \gets \Ext(\crs, \td, x, \pi)$.

    \noindent \textbf{(7)} If $s = s_b$ for $b \in \zo$, then output $b$. Else, output $s_b$ for $b$ chosen uniformly at random.
\end{addmargin}
Given the event in \cref{eqn:simext1} holds, then the reduction will, with advantage $1/q(\lambda)$ for some polynomial $q(\cdot)$, succeed at breaking the hiding of $\Com$, thus reaching a contradiction.  
\end{proof}

Since \cref{eqn:simext} directly contradicts \cref{claim:simext} which we have proven, then we have reached a contradiction. Therefore, the protocol must be simulation extractable.\\

\noindent \textbf{Unclonable Extractability.}
Let $\Pi.\Sim = (\Pi.\Sim_0, \Pi.\Sim_1)$ be the adaptive multi-theorem computationally zero-knowledge simulator of $\Pi$.
Let $\Pi.\Ext$ be the simulation-extraction extractor of $\Pi$ with respect to $\Pi.\Sim$.
Let $\Sim = (\Sim_0, \Sim_1)$ be the simulator, with oracle access to $\Pi.\Sim$, as defined in the proof that \cref{fig:unizk-crs} is adaptive multi-theorem computational zero-knowledge.
Let $\Ext$ be the extractor, based on $\Sim$, as defined in the proof that \cref{fig:unizk-crs} is simulation-extractable.
We define $\cE$ with oracle access to $\Sim$, $\Ext$, and some $\cA$ as follows:
\begin{addmargin}[2em]{2em} 
    \noindent \emph{Hardwired}: $x_1, \ldots, x_{k-1}$, $x$
    
    \noindent \textbf{(1)} Send $1^\lambda$ to $\Sim_0$. Receive $(\crs, \td)$ from $\Sim_0$.

    \noindent \textbf{(2)} For $\iota \in [k-1]$: send $(\crs, \td, x_\iota)$ to $\Sim_1$, and receive $\pi_\iota$ from $\Sim_1$.

    \noindent \textbf{(3)} Send $(\crs, \{x_\iota, \pi_\iota\}_{\iota \in [k-1]})$ to $\cA$. Receive $\{\widetilde{x_\iota}, \widetilde{\pi_\iota}\}_{\iota \in [k]}$ from $\cA$. 
    
    \noindent \textbf{(4)} 
    Define $j'$ uniformly at random from $[k]$.

    \noindent \textbf{(5)} Output $\Ext(\crs, \td, x, \widetilde{\pi_{j'}})$ as $w$.
\end{addmargin}

Let $\cA$, $(x_1, w_1), \ldots, (x_{k-1}, w_{k-1}) \in \cR$, $x$, polynomial $p(\cdot)$, and negligible function $\negl(\cdot)$ be given such that $\cA$ outputs more accepting proofs for $x$ than $\cA$ received, and yet the extractor $\cE$ is unable to extract a valid witness for $x$ from $\cA$. Restated more formally, that is that
\begin{align}
    &\Pr_{\substack{(\crs, \td) \gets \Setup(1^\lambda) \\ \forall \iota \in [k-1], \: \pi_\iota \gets \sP(\crs, x_\iota, w_\iota) \\ \{\widetilde{x_\iota}, \widetilde{\pi_\iota}\}_{\iota \in [k]} \gets \cA(\crs, \{x_\iota, \pi_\iota\}_{\iota \in [k-1]})}}\left[\begin{array}{cc}
        & \exists~ \cJ \subseteq \{j:\widetilde{x}_j = x\} \text{ s.t. } |\cJ| > |\{i:x_i = x\}| \\
        & \text{ and } {\forall \iota \in \cJ}, \sV(\crs, x, \widetilde{\pi_\iota}) = 1 
        \end{array}\right] \geq \frac{1}{p(\lambda)}, \label{eqn:unc1-crs}
\end{align}
and for all polynomials $p'(\cdot)$ (there are infinitely many $\lambda$) such that
\begin{align}
    &\Pr_{w \gets \cE^\cA(x_1, \ldots, x_{k-1}, x)}\left[(x, w) \in \cR \right] \le \frac{1}{p'(\lambda)}.\label{eqn:unc2-crs}
\end{align}
We parse the output of the adversary $\cA$ as $\widetilde{\pi_\iota} = (\widetilde{\rho_{\$,\iota}}, \widetilde{s_\iota}, \widetilde{\pi_{\Pi,\iota}})$ for all $\iota \in [k]$.

Given \cref{eqn:unc1-crs}, we may be in one of the two following cases: either $\cA$ generates two accepting proofs  
which have the same serial number as an honestly generated proof (for an infinite set of $\lambda$), 
or $\cA$ does not (for an infinite set of $\lambda$).
We consider that either of these two scenarios occur with at least $1/(2p(\lambda))$ probability and show that each reaches a contradiction.

\noindent \underline{Scenario One}

Say that (for an infinite set of $\lambda$) $\cA$ generates two accepting proofs which have the same serial number as an honestly generated proof with at least $1/(2p(\lambda))$ probability. 
Symbolically,
\begin{equation}
    \label{eqn:unc3-crs}
    \Pr_{\substack{(\crs, \td) \gets \Setup(1^\lambda) \\ \forall \iota \in [k-1], \: \pi_\iota \gets \sP(\crs, x_\iota, w_\iota) \\ \{\widetilde{x_\iota}, \widetilde{\pi_\iota}\}_{\iota \in [k]} \gets \cA(\crs, \{x_\iota, \pi_\iota\}_{\iota \in [k-1]})}}\left[\begin{array}{cc}
        & \exists~ \cJ \subseteq \{j:\widetilde{x}_j = x\} \text{ s.t. } |\cJ| > |\{i:x_i = x\}| \\
        & \text{ and } {\forall \iota \in \cJ}, \sV(\crs, x, \widetilde{\pi_\iota}) = 1 \\
        & \text{ and } \exists i^* \in [k-1] \:\: \exists j^*, \ell^* \in \cJ \text{ s.t. } s_{i^*} = \widetilde{s_{j^*}} = \widetilde{s_{\ell^*}}
        \end{array}\right]
    \geq \frac{1}{2p(\lambda)}.
\end{equation}
Through a hybrid argument, we can get a similar event with fixed indices $i^*$, $j^*$, and $\ell^*$ which belong to their respective sets with an advantage of $1/(2k^3p(\lambda))$.
By using the advantage of $\cA$ in this game, we can show a reduction that breaks the unforgeability of the quantum money scheme. We will now outline this reduction.
\begin{addmargin}[2em]{2em} 
    \noindent \underline{Reduction}: to unforgeability of quantum money scheme given oracle access to $\cA$.

    \noindent \emph{Hardwired with}: $(x_1, w_1), \ldots, (x_{k-1}, w_{k-1})$, $x$, $i^*$, $j^*$, $\ell^*$.

    \noindent \textbf{(1)} Compute $(\crs, \td) \gets \Setup(1^\lambda)$ where $\crs = (\crs_\Pi, c)$ and $\td = \td_\Pi$.

    \noindent \textbf{(2)} Receive $(\rho_\$, s) \gets \NoteGen$ from the challenger.

    \noindent \textbf{(3)} Define $\rho_{\$, i^*} = \rho_\$$, $s_{i^*} = s$, and $x_\Pi = (c, x_{i^*}, s_{i^*})$. Compute $\pi_{\Pi, \ell} \gets \Pi.\sP(\crs_\Pi, x_\Pi, w_{i^*})$. Define $\pi_{i^*} = (\rho_{\$, i^*}, s_{i^*}, \pi_{\Pi, i^*})$.
    
    \noindent \textbf{(4)} Define $\pi_\iota \gets \sP(\crs, x_\iota, w_\iota)$ for $\iota \in [k-1] \setminus \{i^*\}$.

    \noindent \textbf{(5)} Send $\{x_\iota, \pi_\iota\}_{\iota \in [k-1]}$ to $\cA$.

    \noindent \textbf{(6)} Receive $\{\widetilde{x_\iota}, \widetilde{\pi_\iota}\}_{\iota \in [k]}$ from $\cA$. 
    
    \noindent \textbf{(7)} Parse $\widetilde{\pi_{j^*}} = (\widetilde{\rho_{\$, j^*}}, \widetilde{s_{j^*}}, \widetilde{\pi_{\Pi, j^*}})$ and $\widetilde{\pi_{\ell^*}} = (\widetilde{\rho_{\$, \ell^*}}, \widetilde{s_{\ell^*}}, \widetilde{\pi_{\Pi, \ell^*}})$.

    \noindent \textbf{(7)} Send $(\widetilde{\rho_{\$, j^*}}, \widetilde{\rho_{\$, \ell^*}})$ to the challenger.
\end{addmargin}
Given the event in \cref{eqn:unc3-crs} holds (for the afore mentioned fixed indices), then the reduction will return two quantum money states with the same serial number as the challenger sent.
With advantage $1/(2k^3p(\lambda))$, the reduction will succeed at breaking unforgeability of the quantum money scheme, thus reaching a contradiction.

\noindent \underline{Scenario Two}

Alternatively, say that (for an infinite set of $\lambda$) $\cA$ does not generate two accepting proofs which have the same serial number as an honestly generated proof with at least $1/(2p(\lambda))$ probability. By the pigeon-hole principle, this means that $\cA$ generates an accepting proof with a serial number which is not amongst the ones it received. 
In summary, we have that
\begin{equation}
    \label{eqn:unc4-crs}
    \Pr_{\substack{(\crs, \td) \gets \Setup(1^\lambda) \\ \forall \iota \in [k-1], \: \pi_\iota \gets \sP(\crs, x_\iota, w_\iota) \\ \{\widetilde{x_\iota}, \widetilde{\pi_\iota}\}_{\iota \in [k]} \gets \cA(\crs, \{x_\iota, \pi_\iota\}_{\iota \in [k-1]})}}\left[\begin{array}{cc}
        & \exists~ \cJ \subseteq \{j:\widetilde{x}_j = x\} \text{ s.t. } |\cJ| > |\{i:x_i = x\}| \\
        & \text{ and } {\forall \iota \in \cJ}, \sV(\crs, x, \widetilde{\pi_\iota}) = 1 \\
        & \text{ and } \exists j^* \in \cJ \text{ s.t. } \widetilde{s_{j^*}} \not\in \{s_\iota\}_{\iota \in [k-1]} 
        \end{array} \right]
    \geq \frac{1}{2p(\lambda)}.
\end{equation}
Through an averaging argument, we can get a similar event with a fixed index $j^*$ that belongs to the event's set $\cJ$ with an advantage of $1/(2kp(\lambda))$. We will now switch to a hybrid where we provide $\cA$ with simulated proofs.

\begin{claim}
    \label{claim:mult-sim-crs}
    There exists a polynomial $q(\cdot)$ such that
    \begin{equation}
        \label{eqn:unc5-crs}
        \Pr_{\substack{(\crs, \td) \gets \Sim_0(1^\lambda) \\ \forall \iota \in [k-1], \: \pi_\iota \gets \Sim_1(\crs, \td, x_\iota) \\ \{\widetilde{x_\iota}, \widetilde{\pi_\iota}\}_{\iota \in [k]} \gets \cA(\crs, \{x_\iota, \pi_\iota\}_{\iota \in [k-1]})}}\left[\begin{array}{cc}
        & \exists~ \cJ \subseteq \{j:\widetilde{x}_j = x\} \text{ s.t. } |\cJ| > |\{i:x_i = x\}| \\
        & \text{ and } {\forall \iota \in \cJ}, \sV(\crs, x, \widetilde{\pi_\iota}) = 1 \\
        & \text{ and } j^* \in \cJ \\
        & \text{ and } \widetilde{s_{j^*}} \not\in \{s_\iota\}_{\iota \in [k-1]} 
        \end{array} \right] \geq \frac{1}{q(\lambda)}.
    \end{equation}
\end{claim}

We will later see a proof of \cref{claim:mult-sim-crs}. For now, assuming that this claim holds, by the definition of $\cE$, \cref{eqn:unc2-crs}, and \cref{eqn:unc5-crs}, there exists a polynomial $q'(\cdot)$ such that
\begin{equation*}
    \Pr_{\substack{(\crs, \td) \gets \Sim_0(1^\lambda) \\ \forall \iota \in [k-1], \: \pi_\iota \gets \Sim_1(\crs, \td, x_\iota) \\ \{\widetilde{x_\iota}, \widetilde{\pi_\iota}\}_{\iota \in [k]} \gets \cA(\crs, \{x_\iota, \pi_\iota\}_{\iota \in [k-1]})\\ j' \urand [k] \\ w \gets \Ext(\crs, \td, x, \widetilde{\pi_{j'}})}}\left[\begin{array}{cc}
        & \exists~ \cJ \subseteq \{j:\widetilde{x}_j = x\} \text{ s.t. } |\cJ| > |\{i:x_i = x\}| \\
        & \text{ and } {\forall \iota \in \cJ}, \sV(\crs, x, \widetilde{\pi_\iota}) = 1 \\
        & \text{ and } j^* \in \cJ \\
        & \text{ and } \widetilde{s_{j^*}} \not\in \{s_\iota\}_{\iota \in [k-1]} \\
        & \text{ and } (x, w) \not\in \cR
        \end{array} \right]
    \geq \frac{1}{q'(\lambda)}.
\end{equation*}
We will additionally have that $j' = j^*$ with advantage  at least $1/(kq'(\lambda))$.
Since $\sV$ accepts $\widetilde{\pi_{j^*}}$ with respect to $x$, $\Pi.\sV$ must accept $\widetilde{\pi_{\Pi, j^*}}$ with respect to $\widetilde{x_{\Pi, j^*}} = (c, x, \widetilde{s_{j^*}})$. Since $\widetilde{s_{j^*}} \not\in \{s_\iota\}_{\iota \in [k-1]}$, we have that $\Pi.\Sim_1$, through $\Sim_1$, has not previously received $\widetilde{x_{\Pi, j^*}}$ as a query. As such, we have that
\begin{equation}
    \label{eqn:unc6-crs}
    \Pr_{\substack{(\crs, \td) \gets \Sim_0(1^\lambda) \\ \forall \iota \in [k-1], \: \pi_\iota \gets \Sim_1(\crs, \td, x_\iota) \\ \{\widetilde{x_\iota}, \widetilde{\pi_\iota}\}_{\iota \in [k]} \gets \cA(\crs, \{x_\iota, \pi_\iota\}_{\iota \in [k-1]}) \\ w \gets \Ext(\crs, \td, \widetilde{x_{j^*}}, \widetilde{\pi_{j^*}})}}\left[\begin{array}{cc}
        & \Pi.\sV(\crs_\Pi, (c, x, \widetilde{s_{j^*}}), \widetilde{\pi_{\Pi, j^*}}) = 1 \\
        & \text{ and } (c, x, \widetilde{s_{j^*}}) \not\in Q_\Pi \\
        & \text{ and } (x, w) \not\in \cR
        \end{array}\right]
    \geq \frac{1}{kq'(\lambda)}
\end{equation}
where $Q_\Pi$ is the set of queries asked through $\Sim_1$ to $\Pi.\Sim_1$.
We now define $\cB$ with oracle access to $\cA$ and $\Sim_1$~\footnote{Here, $\cB$ is given oracle access to $\Sim_1$ which has the terms $(\crs, \td)$ fixed by the output of $\Sim_0$.}:
\begin{addmargin}[2em]{2em} 
    \noindent \emph{Hardwired}: $x_1, \ldots, x_{k-1}$, $x$ $j^*$

    \noindent \emph{Input}: $\crs = (\crs_\Pi, c)$
    
    \noindent \textbf{(1)} For $\iota \in [k-1]$: send $x_\iota$ to $\Sim_1$, and receive $\pi_\iota$ from $\Sim_1$.

    \noindent \textbf{(2)} Send $(\crs, \{x_\iota, \pi_\iota\}_{\iota \in [k-1]})$ to $\cA$. Receive $\{\widetilde{x_\iota}, \widetilde{\pi_\iota}\}_{\iota \in [k]}$ from $\cA$. 
    
    \noindent \textbf{(3)} Output $((c, x, \widetilde{s_{j^*}}), \widetilde{\pi_{j^*}})$.
\end{addmargin}
Given that the event in \cref{eqn:unc6-crs} holds, then $\cB$ contradicts \cref{claim:simext}. Thus, all that remains to be proven is \cref{claim:mult-sim-crs}.

\begin{proof}[Proof of \cref{claim:mult-sim-crs}]
    We proceed by contradiction.
    Let $\negl'(\cdot)$ be a negligible function such that
    \begin{equation}
        \label{eqn:unc8-crs}
        \Pr_{\substack{(\crs, \td) \gets \Sim_0(1^\lambda) \\ \forall \iota \in [k-1], \: \pi_\iota \gets \Sim_1(\crs, \td, x_\iota) \\ \{\widetilde{x_\iota}, \widetilde{\pi_\iota}\}_{\iota \in [k]} \gets \cA(\crs, \{x_\iota, \pi_\iota\}_{\iota \in [k-1]})}}\left[\begin{array}{cc}
        & \exists~ \cJ \subseteq \{j:\widetilde{x}_j = x\} \text{ s.t. } |\cJ| > |\{i:x_i = x\}| \\
        & \text{ and } {\forall \iota \in \cJ}, \sV(\crs, x, \widetilde{\pi_\iota}) = 1 \\
        & \text{ and } j^* \in \cJ \\
        & \text{ and } \widetilde{s_{j^*}} \not\in \{s_\iota\}_{\iota \in [k-1]} 
        \end{array} \right]
        \le \negl'(\lambda).
    \end{equation}
    By \cref{eqn:unc4-crs} and \cref{eqn:unc8-crs}, there exists a polynomial $q^*(\cdot)$ such that 
    \begin{align}
        &\left\vert \Pr_{\substack{(\crs, \td) \gets \Setup(1^\lambda) \\ \forall \iota \in [k-1], \: \pi_\iota \gets \sP(\crs, x_\iota, w_\iota) \\ \{\widetilde{x_\iota}, \widetilde{\pi_\iota}\}_{\iota \in [k]} \gets \cA(\crs, \{x_\iota, \pi_\iota\}_{\iota \in [k-1]})}}\left[\begin{array}{cc}
        & \exists~ \cJ \subseteq \{j:\widetilde{x}_j = x\} \text{ s.t. } |\cJ| > |\{i:x_i = x\}| \\
        & \text{ and } {\forall \iota \in \cJ}, \sV(\crs, x, \widetilde{\pi_\iota}) = 1 \\
        & \text{ and } j^* \in \cJ \\
        & \text{ and } \widetilde{s_{j^*}} \not\in \{s_\iota\}_{\iota \in [k-1]} 
        \end{array}\right] \right. \nonumber\\
        &\left. - \Pr_{\substack{(\crs, \td) \gets \Sim_0(1^\lambda) \\ \forall \iota \in [k-1], \: \pi_\iota \gets \Sim_1(\crs, \td, x_\iota) \\ \{\widetilde{x_\iota}, \widetilde{\pi_\iota}\}_{\iota \in [k]} \gets \cA(\crs, \{x_\iota, \pi_\iota\}_{\iota \in [k-1]})}}\left[\begin{array}{cc}
        & \exists~ \cJ \subseteq \{j:\widetilde{x}_j = x\} \text{ s.t. } |\cJ| > |\{i:x_i = x\}| \\
        & \text{ and } {\forall \iota \in \cJ}, \sV(\crs, x, \widetilde{\pi_\iota}) = 1 \\
        & \text{ and } j^* \in \cJ \\
        & \text{ and } \widetilde{s_{j^*}} \not\in \{s_\iota\}_{\iota \in [k-1]} 
        \end{array}\right]\right\vert \ge \frac{1}{q^*(\lambda)}.\label{eqn:unc7-crs}
    \end{align}
    By using the advantage of $\cA$ in this game, we can show a reduction that breaks the multi-theorem zero-knowledge of \cref{fig:unizk-crs}. We will now outline this reduction.
    \begin{addmargin}[2em]{2em} 
        \noindent \underline{Reduction}: to multi-theorem zero-knowledge of our protocol given oracle access to $\cA$.
        
        \noindent \emph{Hardwired}: $(x_1, w_1), \ldots, (x_{k-1}, w_{k-1})$, $x$, $j^*$
    
        \noindent \textbf{(1)} Receive (real or simulated) $\crs$ from the challenger.
        
        \noindent \textbf{(2)} For $\iota \in [k-1]$: send $(x_\iota, w_\iota)$ to the challenger, and receive (real or simulated) $\pi_\iota$ from the challenger.
    
        \noindent \textbf{(3)} Send $(\crs, \{x_\iota, \pi_\iota\}_{\iota \in [k-1]})$ to $\cA$. Receive $\{\widetilde{x_\iota}, \widetilde{\pi_\iota}\}_{\iota \in [k]}$ from $\cA$. 
        
        \noindent \textbf{(4)} Parse $\widetilde{\pi_b} = (\widetilde{\rho_{\$, b}}, \widetilde{s_b}, \widetilde{\pi_{\Pi, b}})$.
        
        \noindent \textbf{(5)} Output $\exists~ \cJ \subseteq \{j:\widetilde{x}_j = x\} \text{ s.t. } |\cJ| > |\{i:x_i = x\}| \text{ and } {\forall \iota \in \cJ}, \sV(\crs, x, \widetilde{\pi_\iota}) = 1 \text{ and } j^* \in \cJ \text{ and } \widetilde{s_{j^*}} \not\in \{s_\iota\}_{\iota \in [k-1]}$.
    \end{addmargin}
    Given that $\cA$ is able to change its output dependent on which of the two worlds in \cref{eqn:unc7-crs} that it is in, then the reduction will be able to distinguish between receiving honest proofs or simulated proofs.
    With advantage $1/q^*(\lambda)$, the reduction will succeed at breaking the adaptive multi-theorem computational zero-knowledge of our protocol, thus reaching a contradiction.
\end{proof}

By completing the proofs of our claim, we have concluding the proof of our theorem statement.
\end{proof}

\begin{corollary}
    \label{cor:main-crs}
    Assuming the polynomial quantum hardness of LWE, injective one-way functions exist, and post-quantum iO exists, there exists a non-interactive adaptive argument of knowledge, adaptive computationally zero-knowledge, and $(k-1)$-to-$k$-unclonable argument with extraction protocol for $\NP$ in the common reference string model (\cref{def:unizk-crs}).
\end{corollary}

\begin{proof}
    This follows from \cref{thm:com}, \cref{cor:simext-nizk-crs}, \cref{thm:qmoney}, and \cref{thm:main-crs}.
\end{proof}

We have thus shown that \cref{fig:unizk-crs} is an unclonable NIZK AoK in the CRS model as defined according to our proposed unclonability definition, \cref{def:unizk-crs}.

In the upcoming sections, we will consider unclonable proof systems in the QROM.

%% file: rom-unclonable-nizk.tex
\section{Unclonable NIZK in the Quantum Random Oracle Model}
\label{sec:qro}

\subsection{A Modified Sigma Protocol}

We will begin by introducing a slightly modified sigma protocol. In the coming sections, our construction will involve applying Fiat-Shamir to this modified protocol.

\begin{theorem}
    \label{thm:mod-sigma}
    Let a post-quantum sigma protocol with unpredictable commitments $\Pi$ be given (see \cref{def:sigma}).
    Let $\cR_\Pi$ be an $\NP$ relation.
    Let $\cR = \{((x, \cS), w) \st (x, w) \in \cR_\Pi \land \cS \ne \emptyset\}$. 
    We argue that the following protocol will be a post-quantum sigma protocol with unpredictable commitments (see \cref{def:sigma}):
    \begin{itemize}
        \item $\sP.\Com(1^\lambda, (x, \cS), w)$: Sends $(x, \alpha, s)$ to $\sV$ where $(\alpha, \state) \gets \Pi.\sP.\Com(1^\lambda, x, w)$ and $s$ is sampled from $\cS$.
        \item $\sV.\Ch(1^\lambda, (x, \cS), (x, \alpha, s))$: Sends $\beta$ to $\sP$ where $\beta \gets \Pi.\sV.\Ch(1^\lambda, x, \alpha)$.
        \item $\sP.\Com(1^\lambda, (x, \cS), w, \state, \beta)$: Sends $\gamma$ to $\sV$ where $\gamma \gets \Pi.\sP.\Prove(1^\lambda, x, w, \state, \beta)$.
        \item $\sV.\Ver(1^\lambda, (x, \cS), (x, \alpha, s), \beta, \gamma)$: Outputs $1$ iff $s \in \mathsf{Support}(\cS)$ and $\Pi.\sV.\Ver(1^\lambda, x, \alpha, \beta, \gamma) = 1$.
    \end{itemize}
\end{theorem}

\begin{proof}
\noindent \textbf{Perfect completeness} 
This follows directly from the perfect completeness of $\Pi$.\\

\noindent \textbf{Proof of Argument with Quantum Extractor.}
Let $\Pi.\Ext$ be the proof of argument quantum extractor for $\Pi$.
Let constant $c_\Pi$, polynomial $p_\Pi(\cdot)$, and negligible functions $\negl_{0, \Pi}(\cdot), \negl_{1, \Pi}(\cdot)$ be given such that for any quantum $\cA_\Pi = (\cA_{0, \Pi}, \cA_{1, \Pi})$ where
\begin{itemize}
    \item $\cA_{0, \Pi}(x)$ is a unitary $U_x$ followed by a measurement and
    \item $\cA_{1, \Pi}(x), \ket{\state}, \beta)$ is a unitary $V_{x,\beta}$ onto the state $\ket{\state}$ followed by a measurement,
\end{itemize}
and any $x$ with associated $\lambda \in \bbN$ satisfying
\begin{equation}
    \label{eqn:sigma-mod}
    \Pr_{\substack{(\alpha, \ket{\state}) \gets \cA_{0, \Pi}(x) \\ \beta \gets \zo^\lambda \\ \gamma \gets \cA_{1, \Pi}(x, \ket{\state}, \beta)}}[\Pi.\sV.\Ver(x, \alpha, \beta, \gamma) = 1] \ge \negl_{0, \Pi}(\lambda)
\end{equation}
we have
\begin{align*}
    &\Pr[(x, \Pi.\Ext^{\cA_\Pi(x)}(x)) \in \cR_\Pi]  \\
    &\ge \frac{1}{p(\lambda)} \cdot \left( \Pr_{\substack{(\alpha, \ket{\state}) \gets \cA_{0, \Pi}(x) \\ \beta \gets \zo^\lambda \\ \gamma \gets \cA_{1, \Pi}(x, \ket{\state}, \beta)}}[\Pi.\sV.\Ver(x, \alpha, \beta, \gamma) = 1] - \negl_{0, \Pi}(\lambda)\right)^{c_\Pi} - \negl_{1, \Pi}(\lambda).
\end{align*}
We define $\Ext$~\footnote{An extractor whose local code is implementable as a simple unitary which allows for straightforward rewinding.} with oracle-access to $\Pi.\Ext$ and some $\cA$ as follows:
\begin{addmargin}[2em]{2em} 
    \noindent \emph{Input}: $x$, $\cS$.

    \noindent \textbf{(1)} Given $(x, \alpha, s)$ from $\cA_\Pi$: send $\alpha$ to $\Pi.\Ext$, receive $\beta$ from $\Pi.\Ext$, and send $\beta$ to $\cA_\Pi$.

    \noindent \textbf{(2)} Upon receiving $\gamma$ from $\cA_\Pi$: send $\gamma$ to $\Pi.\Ext$.

    \noindent \textbf{(3)} Output the result of $\Pi.\Ext$ as $w$.
\end{addmargin}

We define the following set of parameters: $c = c_\Pi$, $p(\cdot) = p_\Pi(\cdot)$, $\negl_0(\cdot) = \negl_{0, \Pi}(\cdot)$ and $\negl_1(\cdot) = \negl_{1, \Pi}(\cdot)$.

\vspace{0.3cm}

Let polynomial-size quantum circuit $\cA=(\cA_0, \cA_1)$ and $(x, \cS)$ be given such that
\begin{equation*}
    \Pr_{\substack{((x, \alpha, s), \ket{\state}) \gets \cA_0(x, \cS) \\ \beta \gets \zo^\lambda \\ \gamma \gets \cA_1((x, \cS), \ket{\state}, \beta)}}[\sV.\Ver((x, \cS), (x, \alpha, s), \beta, \gamma) = 1] \ge \negl_0(\lambda).
\end{equation*}
We now define $\cA_\Pi = (\cA_{0, \Pi}, \cA_{1, \Pi})$ with oracle-access to $\cA$. $\cA_{0, \Pi}$ is hardwired with $\cS$, takes input $x$, sends $(x, \cS)$ to $\cA_0$, receives $((x, \alpha, s), \ket{\state})$ from $\cA_0$, and outputs $(\alpha, \ket{\state})$.
$\cA_{1, \Pi}$ is hardwired with $\cS$, takes input $(x, \ket{\state}, \beta)$, sends $((x, \cS), \ket{\state}, \beta)$ to $\cA_1$, receives $\gamma$ from $\cA_1$, outputs $\gamma$.
By the structure of our proof and definition of our verifier, this means that
\begin{align*}
    &\Pr_{\substack{(\alpha, \ket{\state}) \gets \cA_{0, \Pi}^{\cA_0}(x, \cS) \\ \beta \gets \zo^\lambda \\ \gamma \gets \cA_{1, \Pi}^{\cA_1}((x, \cS), \ket{\state}, \beta)}}[\Pi.\sV.\Ver(x, \alpha, \beta, \gamma) = 1] \\
    &\ge \Pr_{\substack{((x, \alpha, s), \ket{\state}) \gets \cA_0(x, \cS) \\ \beta \gets \zo^\lambda \\ \gamma \gets \cA_1((x, \cS), \ket{\state}, \beta)}}[\sV.\Ver((x, \cS), (x, \alpha, s), \beta, \gamma) = 1] \ge \negl_0(\lambda)
\end{align*}
which satisfies the constraint in \cref{eqn:sigma-mod}. This means we have, when combined with our definition of $\Ext$, that
\begin{align*}
    &\Pr[((x, \cS), \Ext^{\cA(x, \cS)}(x, \cS)) \in \cR] = \Pr[(x, \Pi.\Ext^{\cA_\Pi(x, \cS)}(x)) \in \cR_\Pi] \\
    & \ge \frac{1}{p_\Pi(\lambda)} \cdot \left( \Pr_{\substack{((x, \alpha, s), \ket{\state}) \gets \cA_{0, \Pi}^{\cA_0}(x, \cS) \\ \beta \gets \zo^\lambda \\ \gamma \gets \cA_{1, \Pi}^{\cA_1}((x, \cS), \ket{\state}, \beta)}}[\Pi.\sV.\Ver(x, \alpha, \beta, \gamma) = 1] - \negl_{0, \Pi}(\lambda)\right)^{c_\Pi} - \negl_{1, \Pi}(\lambda)\\
    &\ge \frac{1}{p_\Pi(\lambda)} \cdot \left( \Pr_{\substack{((x, \alpha, s), \ket{\state}) \gets \cA_{0}(x, \cS) \\ \beta \gets \zo^\lambda \\ \gamma \gets \cA_{1}((x, \cS), \ket{\state}, \beta)}}[\sV.\Ver((x, \cS), (x, \alpha, s), \beta, \gamma) = 1] - \negl_{0, \Pi}(\lambda)\right)^{c_\Pi} - \negl_{1, \Pi}(\lambda)\\
    &\ge \frac{1}{p(\lambda)} \cdot \left( \Pr_{\substack{((x, \alpha, s), \ket{\state}) \gets \cA_0(x, s) \\ \beta \gets \zo^\lambda \\ \gamma \gets \cA_1((x, s), \ket{\state}, \beta)}}[\sV.\Ver((x, s), (x, \alpha, s), \beta, \gamma) = 1] - \negl_0(\lambda)\right)^c - \negl_1(\lambda).
\end{align*}
Thus showing that our protocol is an argument of knowledge protocol.\\

\noindent \textbf{Computational Honest-Verifier Zero-Knowledge with Quantum Simulator.}
Let $\Pi.\Sim$ be the computational honest-verifier zero-knowledge quantum simulator for $\Pi$. We define $\Sim$ with oracle access to $\Pi.\Sim$ as follows:
\begin{addmargin}[2em]{2em} 
    \noindent \emph{Input}: $x$, $\cS$.

    \noindent \textbf{(1)} Compute $(\alpha, \beta, \gamma) \gets \Pi.\Sim(1^\lambda, x)$.

    \noindent \textbf{(2)} Sample $s$ from $\cS$.

    \noindent \textbf{(3)} Output $((x, \alpha, s), \beta, \gamma)$.
\end{addmargin}
Let a polynomial $p(\cdot)$, a polynomial-size quantum circuit $\cD$, $\lambda \in \bbN$, and $((x, \cS), w) \in \cR$ be given such that
\begin{align*}
    &\left\vert \Pr_{\substack{ ((x, \alpha, s), \state) \gets \sP.\Com((x, \cS), w) \\ \beta \gets \sV.\Ch((x, \cS), (x, \alpha, s)) \\ \gamma \gets \sP.\Prove((x, \cS), w, \state, \beta)}}[\cD((x, \cS), (x, \alpha, s), \beta, \gamma) = 1] \right.\\
    &\left. - \Pr_{\substack{((x, \alpha, s), \beta, \gamma) \gets \Sim(1^\lambda, (x, \cS))}}[\cD((x, \cS), (x, \alpha, s), \beta, \gamma) = 1]\right\vert \ge \frac{1}{p(\lambda)}.
\end{align*}
We define a reduction to the zero-knowledge property of $\Pi$ as follows:
\begin{addmargin}[2em]{2em} 
    \noindent \underline{Reduction}: to zero-knowledge of $\Pi$ given oracle access to $\cD$.

    \noindent \textit{Hardwired with}: $x$, $\cS$.
    
    \noindent \textbf{(1)} Receive (real or simulated) $(\alpha, \beta, \gamma)$ from the challenger. 

    \noindent \textbf{(2)} Sample $s$ from $\cS$.

    \noindent \textbf{(3)} Send $((x, \alpha, s), \beta, \gamma)$ to $\cD$. Receive $b$ from $\cD$.

    \noindent \textbf{(4)} Output $b$.
\end{addmargin}
When the challenger sends a real (or simulated) proof for $\Pi$, the reduction generates a proof that is identical to the real (resp. simulated) proof.
As such, this reduction preserves the distinguishing advantage of $\cD$. This reaches a contradiction against the zero-knowledge property of $\Pi$. Hence, our protocol must be zero-knowledge. 
\\

\noindent \textbf{Unpredictable Commitment.}
Let $\negl_\Pi(\cdot)$ be a negligible function for the unpredictable commitment property of $\Pi$.

Let a polynomial function $p(\cdot)$, $\lambda \in \bbN$, and $((x, \cS), w) \in \cR$ be given such that
\begin{equation*}
    \Pr_{\substack{ ((x, \alpha, s), \state) \gets \sP.\Com((x, \cS), w) \\ ((x, \alpha', s'), \state') \gets \sP.\Com((x, \cS), w)}}[(\alpha, s) = (\alpha', s')] \ge \frac{1}{p(\lambda)}.
\end{equation*}
By the definition of the honest prover $\sP.\Com$, 
\begin{align*}
    \Pr_{\substack{ (\alpha, \state) \gets \Pi.\sP.\Com(x, w) \\ (\alpha', \state') \gets \Pi.\sP.\Com(x, w)}}[\alpha = \alpha'] \ge 
    \Pr_{\substack{ ((x, \alpha, s), \state) \gets \sP.\Com((x, \cS), w) \\ ((x, \alpha', s'), \state') \gets \sP.\Com((x, \cS), w)}}[(\alpha, s) = (\alpha', s')] \ge \frac{1}{p(\lambda)}
\end{align*}
which is a contradiction. Hence our protocol must have unpredictable commitments.
\end{proof}

\begin{corollary}
    \label{cor:mod-sigma}
    The Fiat-Shamir transform applied to the post-quantum sigma protocol defined in \cref{thm:mod-sigma} yields a classical post-quantum NIZKAoK $\Pi'$ in the QROM (\cref{def:nizkpok-qro}).
\end{corollary}

\begin{proof}
    This follows by \cref{thm:mod-sigma} and \cref{thm:pq-fs}.
\end{proof}

\subsection{Unclonability Definitions}
\label{sec:unc-defs-qro}

Unclonable NIZKs in the quantum random oracle model are defined analogously to the CRS model -- we repeat these definitions in the QRO model for completeness below.

\begin{definition} (Unclonable Security for Hard Instances).
\label{def:uncnizk-alt-rom}
A proof  $(\sP, \sV)$ satisfies unclonable security with respect to a quantum random oracle $\cO$ if
for every language $\cL$ with corresponding relation $\cR_\cL$,
for every polynomial-sized quantum circuit family $\{C_\lambda\}_{\lambda \in \mathbb{N}}$,
and for every hard distribution $\{\mathcal{X}_\lambda,\mathcal{W}_\lambda\}_{\lambda \in \bbN}$ over $\cR_\cL$,
there exists a negligible function $\negl(\cdot)$ such that
for every $\lambda \in \bbN$,
\begin{equation*}
    \Pr_{\substack{(x,w) \leftarrow (\mathcal{X}_\lambda,\mathcal{W}_\lambda) \\
    \pi \leftarrow \sP^\cO(x,w)\\
    \pi_1, \pi_2 \leftarrow C_\lambda(x, \pi)}}\Bigg[
    \sV^\cO(x, \pi_1) = 1 \bigwedge
    \sV^\cO(x, \pi_2) = 1
    \Bigg]
    \leq \negl(\lambda).
\end{equation*}
\end{definition}

\begin{definition}[$(k-1) \text{-to-} k$-Unclonable Extractable NIZK in QROM]
\label{def:unizk-rom}
Let security parameter $\lambda \in \bbN$ and $\NP$ relation $\cR$ with corresponding language $\cL$ be given.
Let $\Pi = (\sP, \sV)$ be given such that $\sP$ and $\sV$ are $\poly(\lambda)$-size quantum algorithms. We have that for any $(x, \omega) \in \cR$, $\sP$ receives an instance and witness pair $(x, \omega)$ as input and outputs $\pi$, and $\sV$ receives an instance $x$ and proof $\pi$ as input and outputs a value in $\zo$.

$\Pi$ is a non-interactive $(k-1) \text{-to-} k$-unclonable NIZKAoK protocol for language $\cL$ with respect to a random oracle $\cO$ if the following holds:
\begin{itemize}
    \item $\Pi$ is a NIZKAoK protocol for language $\cL$ in the quantum random oracle model (\cref{def:nizkpok-qro}).

    \item \textbf{$(k-1) \text{-to-} k$-Unclonable with Extraction}: 
    There exists an oracle-aided polynomial-size quantum circuit $\cE$ such that for every polynomial-size quantum circuit $\cA$ with non-uniform quantum advice $\aux$, for every tuple of $k-1$ instance-witness pairs $(x_1, \omega_1), \ldots, (x_{k-1}, \omega_{k-1}) \in \cR$, for every instance $x$,
    if there exists a polynomial $p(\cdot)$ such that
    \begin{equation*}
        \Pr_{\substack{\cO \\ \forall \iota \in [k-1], \: \pi_\iota \gets \sP^\cO(x_\iota, w_\iota) \\ \{\widetilde{{x}_\iota},\widetilde{\pi_\iota}\}_{\iota \in [k]} \gets \cA^\cO(\{x_\iota, \pi_\iota\}_{\iota \in [k-1]})}}
        \left[
        \begin{array}{cc}
        & \exists~ \cJ \subseteq \{j:\widetilde{x}_j = x\} \text{ s.t. } |\cJ| > |\{i:x_i = x\}| \\
        & \text{ and } {\forall \iota \in \cJ}, \sV^\cO(x, \widetilde{\pi_\iota}) = 1 
        \end{array}
        \right]
        \geq \frac{1}{p(\lambda)},
    \end{equation*}
    then there is a polynomial $q(\cdot)$ such that
    \begin{equation*}
        \Pr_{w \gets \cE^\cA(x_1, \ldots, x_{k-1}, x)}\left[(x, w) \in \cR \right] \geq \frac{1}{q(\lambda)}.
    \end{equation*}
\end{itemize}
\end{definition}

As we did in the previous section, we observe in \cref{def:unizk-rom} that we can generically boost the extractor’s success probability to $1 - \negl(\lambda)$ with respect to a security parameter $\lambda$.

\begin{definition}[$(k-1) \text{-to-} k$-Unclonable Strong-Extractable NIZK in QROM]
\label{def:unizk-rom2}
Let security parameter $\lambda \in \bbN$ and $\NP$ relation $\cR$ with corresponding language $\cL$ be given.
Let $\Pi = (\sP, \sV)$ be given such that $\sP$ and $\sV$ are $\poly(\lambda)$-size quantum algorithms. We have that for any $(x, \omega) \in \cR$, $\sP$ receives an instance and witness pair $(x, \omega)$ as input and outputs $\pi$, and $\sV$ receives an instance $x$ and proof $\pi$ as input and outputs a value in $\zo$.

$\Pi$ is a non-interactive $(k-1) \text{-to-} k$-unclonable NIZKAoK protocol for language $\cL$ with respect to a random oracle $\cO$ if the following holds:
\begin{itemize}
    \item $\Pi$ is a NIZKAoK protocol for language $\cL$ in the quantum random oracle model (\cref{def:nizkpok-qro}).

    \item \textbf{$(k-1) \text{-to-} k$-Unclonable with Extraction}: 
    There exists an oracle-aided polynomial-size quantum circuit $\cE$ such that for every polynomial-size quantum circuit $\cA$, for every tuple of $k-1$ instance-witness pairs $(x_1, \omega_1), \ldots, (x_{k-1}, \omega_{k-1}) \in \cR$, for every instance $x$,
    if there exists a polynomial $p(\cdot)$ such that
    \begin{equation*}
        \Pr_{\substack{\cO \\ \forall \iota \in [k-1], \: \pi_\iota \gets \sP^\cO(x_\iota, w_\iota) \\ \{\widetilde{{x}_\iota},\widetilde{\pi_\iota}\}_{\iota \in [k]} \gets \cA^\cO(\crs, \{x_\iota, \pi_\iota\}_{\iota \in [k-1]})}}
        \left[
        \begin{array}{cc}
        & \exists~ \cJ \subseteq \{j:\widetilde{x}_j = x\} \text{ s.t. } |\cJ| > |\{i:x_i = x\}| \\
        & \text{ and } {\forall \iota \in \cJ}, \sV^\cO(x, \widetilde{\pi_\iota}) = 1 
        \end{array}
        \right]
        \geq \frac{1}{p(\lambda)},
    \end{equation*}
    then there is also a polynomial $\poly(\cdot)$ and a negligible function $\negl(\cdot)$ such that
    \begin{equation*}
        \Pr_{w \gets \cE^\cA(x_1, \ldots, x_{k-1}, x, \aux^{\tensor \poly(\lambda)})}\left[(x, w) \in \cR \right] \geq 1 - \negl(\lambda).
    \end{equation*}
\end{itemize}
\end{definition}

Similar to the previous section, we have the following two lemmas.

\begin{lemma}
    \label{lem:uncdef-rom}
    Let $\Pi = (\Setup, \sP, \sV)$ be a a non-interactive $1 \text{-to-} 2$-unclonable zero-knowledge quantum protocol (\cref{def:unizk-rom}). Then, $\Pi$ satisfies \cref{def:uncnizk-alt-rom}.
\end{lemma}

For a proof of \cref{lem:uncdef-rom}, we refer to \cref{app:defs-reduct}.

\begin{lemma}
    \label{lem:uncdef-rom2}
    Let $\Pi = (\Setup, \sP, \sV)$ be a $(k-1) \text{-to-} k$-unclonable with extraction, non-interactive zero-knowledge quantum protocol (\cref{def:unizk-rom}). Then, $\Pi$ satisfies \cref{def:unizk-rom2}.
\end{lemma}

\begin{proof}[Proof Sketch]
    Given an extractor $\Pi.\cE$ from \cref{def:unizk-rom}, we define a new extractor $\cE$. According to \cref{def:unizk-rom2}, $\cE$ receives multiple copies of the adversary's quantum advice string $\aux$. $\cE$ runs runs $\Pi.\cE$ on the adversary multiple times, each time using a fresh copy of $\aux$. 
    
    Formally, for every $\cA$ with $\aux$, $(x_1, w_1), \ldots, (x_{k-1}, w_{k-1}) \in \cR$, $x$, 
    polynomial $p(\cdot)$, and polynomial $q(\cdot)$ such that
    \begin{align*}
        &\Pr_{\substack{\cO \\ \forall \iota \in [k-1], \: \pi_\iota \gets \sP^\cO(x_\iota, \omega_\iota) \\ \{\widetilde{\pi_\iota}\}_{\iota \in [k]} \gets \cA^\cO(\{x_\iota, \pi_\iota\}_{\iota \in [k-1]}, \aux)}}\left[\begin{array}{cc}
        & \exists~ \cJ \subseteq \{j:\widetilde{x}_j = x\} \text{ s.t. } |\cJ| > |\{i:x_i = x\}| \\
        & \text{ and } {\forall \iota \in \cJ}, \sV^\cO(x, \widetilde{\pi_\iota}) = 1 
        \end{array}\right]
        \geq \frac{1}{p(\lambda)}, \text{ and}\\
        &\Pr_{w \gets \Pi.\cE^\cA(x_1, \ldots, x_{k-1}, x)}\left[(x, w) \in \cR \right] \ge \frac{1}{q(\lambda)},
    \end{align*}
    there exists a polynomial $\poly(\cdot)$ and a negligible function $\negl(\cdot)$ such that the extractor $\cE$ will succeed with probability
    \begin{align*}
        &\Pr_{w \gets \cE^\cA(x_1, \ldots, x_{k-1}, x, \aux^{\tensor \poly(\lambda)})}\left[(x, w) \in \cR \right] \\
        &\geq \left(\Pr_{w \gets \Pi.\cE^{\cA(\cdot, \cdot, \aux)}(x_1, \ldots, x_{k-1}, x)}\left[(x, w) \in \cR \right] \right)^{\poly(\lambda)} \\
        &\ge 1 - \left(1 - \frac{1}{q(\lambda)}\right)^{\poly(\lambda)} \ge 1 - \negl(\lambda).
    \end{align*}
    Thus, $\cE$ satisfies \cref{def:unizk-rom2}.
\end{proof}

From the above lemmas, we conclude that \cref{def:unizk-rom} is the strongest definition. In the following sections, we construct a protocol that satisfies \cref{def:unizk-rom}.

\subsection{Unclonable NIZK Implies Public-Key Quantum Money Mini-Scheme in QROM}

\begin{figure}[!ht]
\begin{framed}
\centering
\begin{minipage}{1.0\textwidth}
\begin{center}
    \underline{Public-Key Quantum Money Mini-Scheme}
\end{center}

\vspace{2mm}

Let $\cO$ be a quantum random oracle. Let $(\cX, \cW)$ be a hard distribution over a language $\cL \in \NP$. Let $\Pi = (\sP, \sV)$ be an unclonable non-interactive zero-knowledge protocol for $\cL$ in the QROM.

\vspace{2mm}
\noindent {\underline{\textsc{Gen}}$^\cO(1^\lambda)$}:
Sample a hard instance-witness pair $(x, w) \gets (\cX, \cY)$ and a proof $\pi \gets \sP^\cO(x, w)$. Output $(\rho_\$ = \pi, s = x)$.

\vspace{1mm}
\noindent {\underline{\textsc{Verify}}$^\cO(\rho_\$, s)$}:
Parse $\rho_\$ = \pi$ and $s = x$. Output $\sV^\cO(x, \pi)$.

\end{minipage}
\end{framed}
\caption{Public-Key Quantum Money Mini-Scheme from an Unclonable Non-Interactive Quantum Protocol}
\label{fig:qmoney-from-unizk-rom}
\end{figure}

\begin{theorem}
\label{thm:unizk-implies-qmoney-qro}
Let $\cO$ be a quantum random oracle. Let $(\cX, \cW)$ be a hard distribution over a language $\cL \in \NP$. Let $\Pi = (\sP, \sV)$ be a $1$-to-$2$ unclonable non-interactive perfectly complete, computationally zero-knowledge protocol for $\cL$ in the QRO model (\cref{def:unizk-rom}).

Then $(\sP, \sV)$ implies a public-key quantum money mini-scheme in the QRO model (\cref{def:qmoney-qro}) as described in \cref{fig:qmoney-from-unizk-rom}.
\end{theorem}

\begin{proof}
    \noindent \textbf{Perfect Correctness.}
    This follows directly from the perfect completeness of $\Pi$.

    \noindent \textbf{Unforgeability.}
    Let $p(\cdot)$ be a polynomial and $\cA$ be a quantum polynomial-time adversary such that for an infinite number of $\lambda \in \bbN^+$,
    \begin{equation*}
        \Pr_{\substack{(\rho_\$, s) \gets \Gen^\cO(1^\lambda) \\ (\rho_{\$, 0}, s_0, \rho_{\$, 1}, s_1) \gets \cA^\cO(\rho_\$, s)}}[s_0 = s_1 = s \: \land \: \Ver^\cO(\rho_{\$, 0}, s_0) = 1 \: \land \: \Ver^\cO(\rho_{\$, 1}, s_1) = 1] \ge \frac{1}{p(\lambda)}.
    \end{equation*}

    We construct a reduction that breaks the uncloneability definition (\cref{def:uncnizk-alt-rom}) which we show (in  \cref{app:defs-reduct}) is implied by our definition (\cref{def:unizk-rom}). The challenger, with access to random oracle $\cO$, samples a hard instance-witness pair $(x, w) \gets (\cX, \cY)$ and a proof $\pi \gets \sP^\cO(x, w)$. The challenger then forwards $(x, \pi)$ to the reduction, which also has oracle access to $\cO$. The reduction then sets $\rho_\$ = \pi$ and $s = x$. The reduction sends $(\rho_\$, s)$ to the adversary $\cA$ who returns back $(\rho_{\$, 0}, s_0, \rho_{\$, 1}, s_1)$. The reduction then parses and sets $\pi_i = \rho_{\$, i}$ for $i \in \{0, 1\}$. The reduction then sends $\pi_0$ and $\pi_1$ back to the challenger.

    When the serial numbers are the same, $s = s_0 = s_1$, we have that the instance will be the same for all the proofs $\pi, \pi_0, \pi_1$. The quantum money state can be parsed as the proof as shown in the construction. When the verification algorithm of the quantum money algorithm accepts both quantum money states $\rho_{\$, 0}$ and $\rho_{\$, 1}$ with respect to $s$, we know that that $\sV^\cO$ would accept both proofs $\pi_0$ and $\pi_1$ with respect to $x$. As such, we will have that the advantage that $\cA$ has at breaking the unforgeability of our quantum money scheme directly translates to the advantage of the reduction at breaking the uncloneability of $\Pi$.
\end{proof}

\subsection{Construction and Analysis of Unclonable-Extractable NIZK in QROM}

\begin{lemma}
    \label{lem:distinct}
    Let $\lambda, k \in \bbN$ and a public-key quantum money mini-scheme $(\NoteGen, \Ver)$ be given. Let points $q_1, \ldots, q_k$ with the following structure be given: a point $q$ contains a serial number $s$ sampled according to $\NoteGen(1^\lambda)$.

    The points $q_1, \ldots, q_k$ must be distinct with overwhelming probability.
\end{lemma}

\begin{proof}
    Each point contains a serial number sampled according to the quantum money generation algorithm, $\NoteGen(1^\lambda)$. By the unpredictability of the serial numbers of quantum money (\cref{def:qmoney}),  all $k$ honestly generated serial numbers must be distinct with overwhelming probability.
    Hence, these $k$ points will be distinct with overwhelming probability.
\end{proof}

\begin{figure}[!ht]
\begin{framed}
\centering
\begin{minipage}{1.0\textwidth}
\begin{center}
    \underline{Unclonable NIZK for $\NP$ in the QROM}
\end{center}

\vspace{2mm}

Let $\cO$ be a random oracle. Let $\Pi = (\sP=(\sP.\Com, \sP.\Prove), \sV=(\sV.\Ch, \sV.\Ver))$ be a post-quantum sigma protocol with unpredictable commitments (see \cref{def:sigma}), and $(\NoteGen, \Ver)$ be a public-key quantum money mini-scheme (see \cref{def:qmoney}). Let $\cR$ be the relation with respect to $\cL \in \NP$.

\vspace{1mm}
\noindent {\underline{\textsc{Prove}}$^\cO(x, \omega)$}:
\begin{itemize}
    \item Compute a quantum note and associated serial number $(\rho_\$, s) \leftarrow \NoteGen(1^\lambda)$.
    \item Compute $(\alpha, \zeta) \gets \sP.\Com(x, \omega)$.
    \item Query $\cO$ at $(x, \alpha, s)$ to get $\beta$.
    \item Compute $\gamma \gets \sP.\Prove(x, \omega, \beta, \zeta)$.
    \item Output $\pi = (\rho_\$, s, \alpha, \beta, \gamma)$.
\end{itemize}

\vspace{1mm}
\noindent {\underline{\textsc{Verify}}$^\cO(x, \pi)$}:
\begin{itemize}
    \item Check that $\Ver(\rho_\$, s)$ outputs $1$.
    \item Check that $\cO$ outputs $\beta$ when queried at $(x, \alpha, s)$.
    \item Output the result of $\sV.\Ver(x, \alpha, \beta, \gamma)$.
\end{itemize}

\end{minipage}
\end{framed}
\caption{Unclonable Non-Interactive Quantum Protocol for $\cL \in \NP$ in the Quantum Random Oracle Model}
\label{fig:unizk-rom}
\end{figure}

We now introduce our construction in \cref{fig:unizk-rom} and prove the main theorem of this section.

\begin{theorem}
\label{thm:main-qro}
Let $k(\cdot)$ be a polynomial. Let $\NP$ relation $\cR$ with corresponding language $\cL$ be given.

Let $(\NoteGen, \Ver)$ be a public-key quantum money mini-scheme (\cref{def:qmoney}) and $\Pi = (\sP, \sV)$ be a post-quantum sigma protocol (\cref{def:sigma}).

$(\sP, \sV)$ as defined in \cref{fig:unizk-rom} will be a non-interactive knowledge sound, computationally zero-knowledge, and $(k-1)$-to-$k$-unclonable argument with extraction protocol for $\cL$ in the quantum random oracle model (\cref{def:nizkpok-qro}).
\end{theorem}

\begin{proof}
    Let the parameters and primitives be as given in the theorem statement. We argue that completeness follows from the protocol construction in~\cref{fig:unizk-rom}, and we prove the remaining properties below.\\

    \noindent \textbf{Argument of Knowledge.}
    Let $\Ext_{FS}$ be the extractor for $\Pi'$ in \cref{cor:mod-sigma} (where $\Pi$ instantiates \cref{thm:mod-sigma}).
    Let $\cR_{FS}$ be the relation for $\Pi'$ with respect to $\cR$.
    Let constant $c_{FS}$, polynomial $p_{FS}(\cdot)$, and negligible functions $\negl_{0, FS}(\cdot), \negl_{1, FS}(\cdot)$ be given such that for any quantum $\cA_{FS}$ and any $(x, \cS)$ with associated $\lambda \in \bbN$ satisfying
    \begin{equation}
        \Pr_{\substack{\cO \\ \pi_{FS} \gets \cA_{FS}^{\ket{\cO}}(x, \cS)}}[\sV_{FS}^\cO((x, \cS), \pi_{FS}) = 1] \ge \negl_{0, FS}(\lambda) \label{eqn:sigma-pok}
    \end{equation}
    we have
    \begin{align*}
        &\Pr[(x, \Ext_{FS}^{\cA_{FS}^{\ket{\cO}}(x, \cS)}(x, \cS)) \in \cR_{FS}]  \\
        &\ge \frac{1}{p_{FS}(\lambda)} \cdot \left( \Pr_{\substack{\cO \\ \pi_{FS} \gets \cA_{FS}^{\ket{\cO}}(x, \cS)}}[\sV_{FS}^\cO((x, \cS), \pi_{FS}) = 1] - \negl_{0, FS}(\lambda)\right)^{c_{FS}} - \negl_{1, FS}(\lambda).
    \end{align*}
    
    Let $\cS$ be the distribution of serial numbers as output by $\NoteGen(1^\lambda)$.
    We define $\Ext$~\footnote{An extractor whose local code is implementable as a simple unitary which allows for straightforward rewinding.} with oracle-access to $\Ext_{FS}$, $\cO$, and some $\cA$ as follows:
    \begin{addmargin}[2em]{2em} 
        \noindent \emph{Hardwired with}: $\cS$.
    
        \noindent \emph{Input}: $x$.

        \noindent \textbf{(1)} Given an oracle-query $(x, \alpha, s)$ from $\cA$: send $(x, \alpha, s)$ to $\cO$, receive $\beta$ from $\cO$, and send $\beta$ to $\cA$.

        \noindent \textbf{(2)} Upon receiving $\pi = (\rho_\$, s, \alpha, \beta, \gamma)$ from $\cA$: send $\pi_{FS} = ((x, \alpha, s), \beta, \gamma)$ to $\Ext_{FS}$.

        \noindent \textbf{(3)} Output the result of $\Ext_{FS}$ as $w$.
    \end{addmargin}

    We define the following set of parameters: $c = c_{FS}$, $p(\cdot) = p_{FS}(\cdot)$, $\negl_0(\cdot) = \negl_{0, FS}(\cdot)$ and $\negl_1(\cdot) = \negl_{1, FS}(\cdot)$.

    \vspace{0.3cm}

    Let polynomial-size quantum circuit $\cA$ and $x$ be given such that
    \begin{equation*}
        \Pr_{\substack{\cO \\ \pi \gets \cA^{\ket{\cO}}(x)}}[\sV^\cO(x, \pi) = 1] \ge \negl_0(\lambda).
    \end{equation*}
    Let $\cA_{FS}$ be defined with oracle-access to some $\cA$ and $\cO$ as follows:
    \begin{addmargin}[2em]{2em} 
        \noindent \emph{Input}: $x$, $\cS$.

        \noindent \textbf{(1)} Given a query $(x, \alpha, s)$ from $\cA$: send $(x, \alpha, s)$ to $\cO$, receive $\beta$ from $\cO$, and send $\beta$ to $\cA$.

        \noindent \textbf{(2)} Upon receiving $\pi = (\rho_\$, s, \alpha, \beta, \gamma)$ from $\cA$: output $\pi_{FS} = ((x,\alpha, s), \beta, \gamma)$.
    \end{addmargin}
    By the structure of our proof and definition of our verifier, this means that
    \begin{align*}
        \Pr_{\substack{\cO \\ \pi_{FS} \gets \cA_{FS}^{\cA(x), \ket{\cO}}(x, \cS)}}[\sV_{FS}^{\cO}((x, \cS), \pi_{FS}) = 1] &\ge \Pr_{\substack{\cO \\ (\rho_\$, \pi_{FS}) \gets \cA^{\ket{\cO}}(x, \cS)}}[\sV_{FS}^{\cO}((x, \cS), \pi_{FS}) = 1 \: \wedge \: \Ver(\rho_\$, s) = 1]\\
        &= \Pr_{\substack{\cO \\ \pi \gets \cA^{\ket{\cO}}(x)}}[\sV^\cO(x, \pi) = 1] \ge \negl_0(\lambda) = \negl_{0, FS}(\lambda)
    \end{align*}
    which satisfies the constraint in \cref{eqn:sigma-pok}. This means we have, when combined with our definition of $\Ext$ and $\cS$, that
    \begin{align*}
        &\Pr[(x, \Ext^{\Ext_{FS}(x), \ket{\cO}, \cA(x)}(x)) \in \cR] = \Pr[((x, \cS), \Ext_{FS}^{\cA_{FS}^{\cA(x), \ket{\cO}}(x, \cS)}(x, \cS)) \in \cR_{FS}]  \\
        &\ge \frac{1}{p_{FS}(\lambda)} \cdot \left( \Pr_{\substack{\cO \\ \pi_{FS} \gets \cA_{FS}^{\cA(x), \ket{\cO}}(x, \cS)}}[\sV_{FS}^\cO((x, \cS), \pi_{FS}) = 1] - \negl_{0, FS}(\lambda)\right)^{c_{FS}} - \negl_{1, FS}(\lambda)\\
        &\ge \frac{1}{p_{FS}(\lambda)} \cdot \left( \Pr_{\substack{\cO \\ \pi \gets \cA^{\ket{\cO}}(x)}}[\sV^\cO(x, \pi) = 1] - \negl_{0, FS}(\lambda)\right)^{c_{FS}} - \negl_{1, FS}(\lambda)\\
        &= \frac{1}{p(\lambda)} \cdot \left( \Pr_{\substack{\cO \\ \pi \gets \cA^{\ket{\cO}}(x)}}[\sV^\cO(x, \pi) = 1] - \negl_0(\lambda)\right)^c - \negl_1(\lambda).
    \end{align*}
    Thus showing that our protocol is an argument of knowledge protocol.\\

    \noindent \textbf{Zero-Knowledge.} Let $\Sim_{FS}$ be the simulator for $\Pi'$ in \cref{cor:mod-sigma} (where $\Pi$ instantiates \cref{thm:mod-sigma}).
    Let $\cR_{FS}$ be the relation for $\Pi'$ with respect to $\cR$.
    We define $\Sim$ with oracle-access to $\Sim_{FS}$ and program access to some random oracle $\cO$ as follows:
    \begin{addmargin}[2em]{2em} 
        \noindent \emph{Input}: $x$ (ignores any witnesses it may receive).

        \noindent \textbf{(1)} Sample $(\rho_\$, s) \gets \NoteGen(1^\lambda)$. 
        
        \noindent \textbf{(2)} Let $\cS$ be the distribution where all probability mass is on $s$.

        \noindent \textbf{(3)} Compute $((x, \alpha, s), \beta, \gamma) \gets \Pi.\Sim(x, \cS)$. Allow $\Pi.\Sim$ to program $\cO$ at $(x, \alpha, s)$ to return $\beta$.

        \noindent \textbf{(5)} Output $\pi = (\rho_\$, s, \alpha, \beta, \gamma)$.
    \end{addmargin}

    Let an oracle-aided distinguisher $\cD$ which can only make queries $(x, w) \in \cR$, and a polynomial $p(\cdot)$ be given such that
    \begin{equation}
        \left\vert \Pr\left[\cD^{\Sim, \ket{\cO}}(1^\lambda) = 1\right] - \Pr_\cO\left[\cD^{\sP^{\cO}, \ket{\cO}}(1^\lambda) = 1\right] \right\vert \ge \frac{1}{p(\lambda)}.
    \end{equation}
    We define a reduction to the zero-knowledge property of $\Pi'$ as follows:
    \begin{addmargin}[2em]{2em} 
        \noindent \underline{Reduction}: to zero-knowledge of $\Pi'$ given oracle access to $\cD$ and program access to $\cO$.

        \noindent For every $(x, w)$ from $\cD$:

        \noindent \textbf{(1)} Sample $(\rho_\$, s) \gets \NoteGen(1^\lambda)$. 
        
        \noindent \textbf{(2)} Let $\cS$ be the distribution where all probability mass is on $s$.

        \noindent \textbf{(3)} Send $((x, \cS), w)$ to the challenger. Receive $((x, \alpha, s), \beta, \gamma)$ from the challenger. 
        The challenger will have already programmed $\cO$ at $(x, \alpha, s)$ to return $\beta$.

        \noindent \textbf{(4)} Output $\pi = (\rho_\$, s, \alpha, \beta, \gamma)$.

        \noindent Output the result of $\cD$.
    \end{addmargin}
    The view of $\cD$ matches that of our protocol in \cref{fig:unizk-rom} or $\Sim$. As such, our reduction should have the same advantage at breaking the zero-knowledge property of $\Pi'$. We reach a contradiction, hence our protocol must be zero-knowledge.\\

    \noindent \textbf{Unclonable Extractability.}
    Let $\Ext$ be the quantum circuit of the extractor we defined earlier (in our proof that \cref{fig:unizk-rom} is an argument of knowledge). Let $\Sim$ be the quantum circuit of the simulator that we defined earlier (in our proof that \cref{fig:unizk-rom} is a zero-knowledge protocol). We define a simulator for our extractor, $\Sim\Ext$, which interacts with some $\cA$ and has oracle-access to $\cO$ as follows:
    \begin{addmargin}[2em]{2em} 
        \noindent \emph{Hardwired with}: $x_1, \ldots, x_{k-1}$, $x$

        \noindent \textbf{(1)} Compute $\pi_\iota \gets \Sim(x_\iota)$ for $\iota \in [k-1]$ where we store all points $\Sim$ would program into a list $\cP$.

        \noindent \textbf{(2)} Send $\{x_\iota, \pi_\iota\}_{\iota \in [k-1]}$ to $\cA$.

        \noindent \textbf{(3)} For every query from $\cA$, if the query is in $\cP$, then reply with the answer from $\cP$. Else, forward the query to $\cO$ and send the answer back to $\cA$.


    \end{addmargin}
    
    \noindent We now define our extractor $\cE$ with oracle-access to some $\cA$ as follows:
    \begin{addmargin}[2em]{2em} 
        \noindent \emph{Hardwired with}: some choice of $x_1, \ldots, x_{k-1}$, $x$.


        \noindent \textbf{(1)} Instantiates a simulatable and extractable random oracle $\cO$. Runs $\Ext$ on $\cO$ throughout the interaction with $\cA$ (which may involve rewinding, in which case we would rewind $\cA$ and repeat the following steps). 




        \noindent \textbf{(2)} Run $\Sim\Ext^\cO(x_1, \ldots, x_{k-1}, x)$ which interacts with $\cA$.

        \noindent \textbf{(3)} Receive $\{\widetilde{x_\iota}, \widetilde{\pi_\iota}\}_{\iota \in [k]}$ from $\cA$. 

        \noindent \textbf{(4)} Samples $\ell \in [k]$ uniformly at random.
        Send $\widetilde{\pi_\ell}$ to $\Ext$.

        \noindent \textbf{(5)} Outputs the result of $\Ext$ as $w$.
    \end{addmargin}

    Let $\cA$, $(x_1, w_1), \ldots, (x_{k-1}, w_{k-1}) \in \cR$, $x$, polynomial $p(\cdot)$, and negligible function $\negl(\cdot)$ be given such that $\cA$ outputs more accepting proofs for $x$ than $\cA$ received, and yet the extractor $\cE$ is unable to extract a valid witness for $x$ from $\cA$.
    Restated more formally, that is that
    \begin{align}
        &\Pr_{\substack{\cO \\ \forall \iota \in [k-1], \: \pi_\iota \gets \sP^\cO(x_\iota, w_\iota) \\ \{\widetilde{{x}_\iota},\widetilde{\pi_\iota}\}_{\iota \in [k]} \gets \cA^\cO(\{x_\iota, \pi_\iota\}_{\iota \in [k-1]})}}
        \left[
        \begin{array}{cc}
        & \exists~ \cJ \subseteq \{j:\widetilde{x}_j = x\} \text{ s.t. } |\cJ| > |\{i:x_i = x\}| \\
        & \text{ and } {\forall \iota \in \cJ}, \sV^\cO(x, \widetilde{\pi_\iota}) = 1 
        \end{array}
        \right]
        \geq \frac{1}{p(\lambda)},\label{eqn:vaccept-unclon-qro}
    \end{align}
    and for all polynomials $p'(\cdot)$ (there are infinitely many $\lambda$) such that
    \begin{align}
        &\Pr_{w \gets \cE^\cA(x_1, \ldots, x_{k-1}, x)}\left[(x, w) \in \cR \right] \le \frac{1}{p'(\lambda)}. \label{eqn:ext-unclon}
    \end{align}
    We parse the output of the adversary $\cA$ as $\widetilde{\pi_\iota} = (\widetilde{\rho_{\$,\iota}}, \widetilde{s_\iota}, \widetilde{\alpha_{\iota}}, \widetilde{\beta_{\iota}}, \widetilde{\gamma_{\iota}})$ for all $\iota \in [k]$.

    Given \cref{eqn:vaccept-unclon-qro}, we may be in one of the two following cases: either $\cA$ generates two accepting proofs which have the same serial number as a honestly generated proof (for an infinite set of $\lambda$), or $\cA$ does not (for an infinite set of $\lambda$).
    We consider that either of these two scenarios occur with at least $1/(2p(\lambda))$ probability and show that each reaches a contradiction.

    \noindent \underline{Scenario One}

    Say that (for an infinite set of $\lambda$) $\cA$ generates two accepting proofs which have the same serial number as an honestly generated proof with at least $1/(2p(\lambda))$ probability. Symbolically,
\begin{equation}
    \label{eqn:vaccept+serialsame-qro}
    \Pr_{\substack{\cO \\ \forall \iota \in [k-1], \: \pi_\iota \gets \sP^\cO(x_\iota, w_\iota) \\ \{\widetilde{x_\iota}, \widetilde{\pi_\iota}\}_{\iota \in [k]} \gets \cA^\cO(\{x_\iota, \pi_\iota\}_{\iota \in [k-1]})}}\left[\begin{array}{cc}
        & \exists~ \cJ \subseteq \{j:\widetilde{x}_j = x\} \text{ s.t. } |\cJ| > |\{i:x_i = x\}| \\
        & \text{ and } {\forall \iota \in \cJ}, \sV^\cO(x, \widetilde{\pi_\iota}) = 1 \\
        & \text{ and } \exists i^* \in [k-1] \:\: \exists j^*, \ell^* \in \cJ \text{ s.t. } s_{i^*} = \widetilde{s_{j^*}} = \widetilde{s_{\ell^*}}
        \end{array}\right]
    \geq \frac{1}{2p(\lambda)}.
\end{equation}
Through a hybrid argument, we can get a similar event with fixed indices $i^*$, $j^*$, and $\ell^*$ which belong to their respective sets with an advantage of $1/(2k^3p(\lambda))$.
By using the advantage of $\cA$ in this game, we can show a reduction that breaks the unforgeability of the quantum money scheme. We will now outline this reduction.
\begin{addmargin}[2em]{2em} 
        \noindent \underline{Reduction}: to unforgeability of quantum money scheme given oracle access to $\cA$ and $\cO$.

        \noindent \emph{Hardwired with}: $(x_1, w_1), \ldots, (x_{k-1}, w_{k-1})$, $x$, $i^*, j^*, \ell^*$.

        \noindent \textbf{(1)} Receive $(\rho_\$, s)$ from the challenger.

        \noindent \textbf{(2)} Define $\rho_{\$, i^*} = \rho_\$$ and $s_{i^*} = s$. Sample $(\rho_{\$, \iota}, s_\iota) \gets \NoteGen(1^\lambda)$ for $\iota \in [k-1] \setminus \{i^*\}$. Compute $(\alpha_\iota, \zeta_\iota) \gets \Pi.\sP.\Com(x_\iota, w_\iota)$, query $\cO$ at $(x_\iota, \alpha_\iota, s_\iota)$ to get $\beta_\iota$, compute $\gamma_\iota \gets \Pi.\sP.\Prove(x_\iota, w_\iota, \beta_\iota, \zeta_\iota)$, and define $\pi_\iota = (\rho_{\$, \iota}, s_\iota, \alpha_\iota, \beta_\iota, \gamma_\iota)$ for $\iota \in [k-1]$.

        \noindent \textbf{(3)} Send $\{x_\iota, \pi_\iota\}_{\iota \in [k-1]}$ to $\cA$.

        \noindent \textbf{(4)} Receive $\{\widetilde{\pi_\iota}\}_{\iota \in [k]}$ from $\cA$.

        \noindent \textbf{(5)} Send $(\widetilde{\rho_\$}_{j^*}, \widetilde{\rho_\$}_{\ell^*})$ to the challenger.
    \end{addmargin}
Given the event in \cref{eqn:vaccept+serialsame-qro} holds (for the afore mentioned fixed indices), then the reduction will return two quantum money states with the same serial number as the challenger sent.
With advantage $1/(2k^3p(\lambda))$, the reduction will succeed at breaking unforgeability of the quantum money scheme, thus reaching a contradiction.

    \noindent \underline{Scenario Two}

    Alternatively, say that (for an infinite set of $\lambda$) $\cA$ does not generate two accepting proofs which have the same serial number as an honestly generated proof with at least $1/(2p(\lambda))$ probability. By the pigeon-hole principle, this means that $\cA$ generates an accepting proof with a serial number which is not amongst the ones it received. 
    In summary, we have that
    \begin{equation}
        \label{eqn:vaccept+serialdiff-qro}
        \Pr_{\substack{\cO \\ \forall \iota \in [k-1], \: \pi_\iota \gets \sP^\cO(x_\iota, w_\iota) \\ \{\widetilde{x_\iota}, \widetilde{\pi_\iota}\}_{\iota \in [k]} \gets \cA^\cO(\{x_\iota, \pi_\iota\}_{\iota \in [k-1]})}}\left[\begin{array}{cc}
        & \exists~ \cJ \subseteq \{j:\widetilde{x}_j = x\} \text{ s.t. } |\cJ| > |\{i:x_i = x\}| \\
        & \text{ and } {\forall \iota \in \cJ}, \sV^\cO(x, \widetilde{\pi_\iota}) = 1 \\
        & \text{ and } \exists j^* \in \cJ \text{ s.t. } \widetilde{s_{j^*}} \not\in \{s_\iota\}_{\iota \in [k-1]} 
        \end{array} \right]
        \geq \frac{1}{2p(\lambda)}.
    \end{equation}
    Through an averaging argument, we can get a similar event with a fixed index $j^*$ that belongs to the event's set $\cJ$ with an advantage of $1/(2kp(\lambda))$. We will now switch to a hybrid where we provide $\cA$ with simulated proofs.

    \begin{claim}
        \label{claim:mult-sim-qro}
        There exists a polynomial $q(\cdot)$ such that
        \begin{equation}
            \label{eqn:mult-sim-qro}
            \Pr_{\substack{\cO \\ \{\pi_\iota\}_{\iota \in [k-1]} \gets \Sim\Ext^\cO(x_1, \ldots, x_{k-1}) \\ \{\widetilde{x_\iota}, \widetilde{\pi_\iota}\}_{\iota \in [k]} \gets \cA^{\Sim\Ext^\cO}(\{x_\iota, \pi_\iota\}_{\iota \in [k-1]})}}\left[\begin{array}{cc}
            & \exists~ \cJ \subseteq \{j:\widetilde{x}_j = x\} \text{ s.t. } |\cJ| > |\{i:x_i = x\}| \\
            & \text{ and } {\forall \iota \in \cJ}, \sV^{\Sim\Ext^\cO}(x, \widetilde{\pi_\iota}) = 1 \\
            & \text{ and } j^* \in \cJ \\
            & \text{ and } \widetilde{s_{j^*}} \not\in \{s_\iota\}_{\iota \in [k-1]} 
            \end{array} \right] \geq \frac{1}{q(\lambda)}.
        \end{equation}
    \end{claim}

    We will later see a proof of \cref{claim:mult-sim-qro}. For now, assuming that this claim holds, we can define an adversary from which $\Ext$ can extract a valid witness for $x$.

    \begin{claim}
        \label{claim:mult-sim-ext-qro}
        There exists a polynomial $q'(\cdot)$ such that
        \begin{equation}
            \label{eqn:mult-sim-ext-qro}
            \Pr_{w \gets \cE^\cA(x_1, \ldots, x_{k-1}, x)}\left[(x, w) \in \cR \right] \ge \frac{1}{q'(\lambda)}.
        \end{equation}
    \end{claim}

    We will soon see a proof for \cref{claim:mult-sim-ext-qro}. Meanwhile, if this claim is true, then we will have a direct contradiction with \cref{eqn:ext-unclon}. Thus, all that remains to be proven are the two claims: \cref{claim:mult-sim-qro} and \cref{claim:mult-sim-ext-qro}. We start by proving the former claim.

    \begin{proof}[Proof of \cref{claim:mult-sim-qro}]
        We first need to argue that our strategy is well-defined, that we will be able to independently program these $k$ points. Then we can argue the indistinguishability of switching one-by-one to simulated proofs.
        We will argue that our simulator will run in expected polynomial time.
        By \cref{lem:distinct}, the $k$ points which our simulator will program will be distinct with overwhelming probability.
        Furthermore, since we assumed that our quantum random oracle can be programmed at multiple distinct points~\cref{def:qro}, our simulator is well-defined.

        We now argue indistinguishability of the simulated proofs from the honestly generated proofs via a hybrid argument. 
        Suppose for sake of contradiction that the probability difference between \cref{eqn:vaccept+serialdiff-qro} and \cref{eqn:mult-sim-qro} was $1/p'(\lambda)$ for some polynomial $p'(\cdot)$. We construct a series of consecutive hybrids for each $i \in [k-1]$ where we switch the $i$th proof from prover generated to simulated.
        By this hybrid argument, there must be some position $\ell^* \in [k-1]$ where switching the $\ell^*$th proof has a probability difference of at least $1/(kp'(\lambda))$. We now formalize a reduction which can distinguish between these two settings:
        \begin{addmargin}[2em]{2em} 
            \noindent \underline{Reduction}: to zero-knowledge of our protocol given oracle access to $\cA$ and some $\cO$.

            \noindent \emph{Hardwired with}: $(x_1, w_1), \ldots, (x_{\ell-1}, w_{\ell-1})$, $x$, $j^*$, $\ell^*$.

            \noindent \textbf{(1)} Receive (real or simulated) $\pi$ from the challenger and mediated query access to a (real or simulated, respectively) random oracle $\cO$.

            \noindent \textbf{(2)} Define $\pi_{\ell^*} = \pi$. Compute $\pi_\iota \gets \sP^\cO(x_\iota, w_\iota)$ for $\iota \in [\ell^* - 1]$. Compute $\pi_\iota \gets \Sim(x_\iota)$ for $\iota \in \{\ell^*+1, \ldots, k-1\}$ where we store all points $\Sim$ would program into a list $\cP$.

            \noindent \textbf{(3)} Send $\{x_\iota, \pi_\iota\}_{\iota \in [k-1]}$ to $\cA$.

            \noindent \textbf{(4)} For every query from $\cA$, if the query is in $\cP$, then reply with the answer from $\cP$. Else, forward the query to $\cO$ and send the answer back to $\cA$. Let $\widehat{\cO}$ denote this modified random oracle.

            \noindent \textbf{(5)} Receive $\{\widetilde{\pi_\iota}\}_{\iota \in [k]}$ from $\cA$.

            \noindent \textbf{(6)} If the following event holds:
            \begin{equation*}
                \begin{array}{cc}
                & \exists~ \cJ \subseteq \{j:\widetilde{x}_j = x\} \text{ s.t. } |\cJ| > |\{i:x_i = x\}| \\
                & \text{ and } {\forall \iota \in \cJ}, \sV^{\widehat{\cO}}(x, \widetilde{\pi_\iota}) = 1 \\
                & \text{ and } j^* \in \cJ \\
                & \text{ and } \widetilde{s_{j^*}} \not\in \{s_\iota\}_{\iota \in [k-1]} 
                \end{array},
            \end{equation*}
            then output $1$. Else, output $0$.
        \end{addmargin}

        We first argue that the view that the reduction provides to $\cA$ matches one of the games: where all proofs up to the $\ell^*$th are simulated or where all proofs up to and including the $\ell^*$th are simulated.
        By \cref{lem:distinct}, the point computed or programmed by the challenger will be distinct from the points which the reduction programs.
        As such, the reduction is allowed to modify~\footnote{In more detail, the reduction can construct a unitary which runs the classical code in step \noindent \textbf{(4)}. This can then be applied in superposition to a query sent by $\cA$.} the oracle which $\cA$ interfaces with (see step \noindent \textbf{(4)}).
        In summary, $\cA$ will be provided access to an oracle that is consistent with all of the proofs it receives.

        We also mention that $\sV$ always receives access to the same oracle that $\cA$ receives. 

        Given that $\cA$ has a view which directly matches its expected view in either game, then the reduction's advantage is the same as $\cA$'s advantage which is at least $1/(kp'(\lambda))$. This is a contradiction with the zero-knowledge property of our protocol. Thus, our original claim must be true.
    \end{proof}

    Now, we continue on to proving the latter claim.

    \begin{proof}[Proof of \cref{claim:mult-sim-ext-qro}]
        We define
        a reduction to the argument of knowledge property of our protocol, which we subsequently will refer to as $\widetilde{\sP}$ (borrowing notation from the definition of AoK in \cref{def:nizkpok-qro}),
        as follows:
        \begin{addmargin}[2em]{2em} 
            \noindent \underline{Reduction}: to argument of knowledge given oracle access to $\cA$ and some $\cO$.

            \noindent \emph{Hardwired with}: $x_1, \ldots, x_{k-1}$, $x$, $j^*$

            \noindent \textbf{(1)} Receive query access to a (real or extractable) random oracle $\cO$.




            \noindent \textbf{(2)} Run $\Sim\Ext^\cO(x_1, \ldots, x_{k-1}, x)$ which interacts with $\cA$.

            \noindent \textbf{(3)} Receive $\{\widetilde{\pi_\iota}\}_{\iota \in [k]}$ from $\cA$.

            \noindent \textbf{(4)} Output $\widetilde{\pi_{j^*}}$.
        \end{addmargin}

        First we must argue that $\cA$'s and $\sV$'s view remains identical to \cref{eqn:mult-sim-qro}.
        The oracle which $\cA$ interfaces with (see $\Sim\Ext$) will be consistent with all of the proofs it receives. Hence
        \begin{equation}
            \label{eqn:mult-sim-qro2}
            \Pr_{\substack{\cO \\ \widetilde{\pi_{j^*}} \gets \widetilde{\sP}^{\cA,\ket{\cO}}(1^\lambda)}}\left[\begin{array}{cc}
            & \exists~ \cJ \subseteq \{j:\widetilde{x}_j = x\} \text{ s.t. } |\cJ| > |\{i:x_i = x\}| \\
            & \text{ and } {\forall \iota \in \cJ}, \sV^{\Sim\Ext^\cO}(x, \widetilde{\pi_\iota}) = 1 \\
            & \text{ and } j^* \in \cJ \\
            & \text{ and } \widetilde{s_{j^*}} \not\in \{s_\iota\}_{\iota \in [k-1]} 
            \end{array} \right] \geq \frac{1}{q(\lambda)}.
        \end{equation}

        Now, by the event in \cref{eqn:mult-sim-qro2}, the proof output by $\widetilde{P}$ has a serial number which differs from the serial numbers in the proofs provided to $\cA$. As such, even when given oracle acccess to the unmodified random oracle, the verification algorithm $\sV$ should continue to accept the proof output by $\widetilde{P}$. Thus, we have that
        \begin{align*}
            &\Pr_{\substack{\cO \\ \widetilde{\pi_{j^*}} \gets \widetilde{\sP}^{\cA,\ket{\cO}}(1^\lambda)}}\left[\sV^{\cO}(x, \widetilde{\pi_{j^*}}) = 1\right] \\
            &\ge \Pr_{\substack{\cO \\ \widetilde{\pi_{j^*}} \gets \widetilde{\sP}^{\cA,\ket{\cO}}(1^\lambda)}}\left[\begin{array}{cc}
            & \exists~ \cJ \subseteq \{j:\widetilde{x}_j = x\} \text{ s.t. } |\cJ| > |\{i:x_i = x\}| \\
            & \text{ and } {\forall \iota \in \cJ}, \sV^{\Sim\Ext^\cO}(x, \widetilde{\pi_\iota}) = 1 \\
            & \text{ and } j^* \in \cJ \\
            & \text{ and } \widetilde{s_{j^*}} \not\in \{s_\iota\}_{\iota \in [k-1]} 
            \end{array} \right] \geq \frac{1}{q(\lambda)}.
        \end{align*}

        By the definition of an argument of knowledge (\cref{def:nizkpok-qro}) which have some parameters polynomial $p^*(\cdot)$ and negligible functions $\negl_0(\cdot)$ and $\negl_1(\cdot)$, we have that there exists some polynomial $q'(\cdot)$ such that
        \begin{align*}
            \Pr[(x, \Ext^{\widetilde{\sP}^{\cA, \ket{\cO_\Ext}}(x)}(x)) \in \cR_\lambda]  &\ge \frac{1}{p^*(\lambda)} \cdot \left(\Pr_{\substack{\cO \\ \widetilde{\pi_{j^*}} \gets \widetilde{\sP}^{\cA,\ket{\cO}}(1^\lambda)}}\left[\sV^{\cO}(x, \widetilde{\pi_{j^*}}) = 1\right]  - \negl_0(\lambda)\right)^c - \negl_1(\lambda) \\
            &\ge \frac{1}{q'(\lambda)}
        \end{align*}
        where $\Ext$ simulates the random oracle $\ket{\cO_\Ext}$.
        We now compare the reduction $\widetilde{\sP}$ to the extractor $\cE$. 
        The extractor $\Ext$ with oracle access to $\widetilde{\sP}$ has an identical view to that in $\cE$ when $\ell = j^*$. We will have that
        \begin{align*}
            \Pr_{w \gets \cE^\cA(x_1, \ldots, x_{k-1}, x)}\left[(x, w) \in \cR \: \vert \: \ell = j^*\right] &= \Pr[(x, \Ext^{\widetilde{\sP}^{\cA, \ket{\cO_\Ext}}(x)}(x)) \in \cR_\lambda] \ge \frac{1}{q'(\lambda)}
        \end{align*}
        and hence
        \begin{align*}
            &\Pr_{w \gets \cE^\cA(x_1, \ldots, x_{k-1}, x)}\left[(x, w) \in \cR \right] \\
            &\ge \Pr_{w \gets \cE^\cA(x_1, \ldots, x_{k-1}, x)}\left[(x, w) \in \cR \: \vert \: \ell = j^*\right] \cdot \Pr[\ell = j^*] \ge \frac{1}{kq'(\lambda)}
        \end{align*}
        which completes the proof of our claim.
    \end{proof}

    By completing the proofs of our claims, we have concluding the proof of our theorem statement.
\end{proof}

\begin{corollary}
    Assuming the injective one-way functions exist, and post-quantum iO exists, there exists a non-interactive knowledge sound, computationally zero-knowledge, and $(k-1)$-to-$k$-unclonable with extraction protocol for $\NP$ in the quantum random oracle model (\cref{def:unizk-rom}).
\end{corollary}

\begin{proof}
    This follows from \cref{thm:qmoney} and \cref{thm:main-qro}. 
\end{proof}

We have thus shown that \cref{fig:unizk-rom} is an unclonable NIZK AoK in the ROM model as defined according to our unclonability definition, \cref{def:unizk-rom}.

%% file: unclonable-sigs.tex
\section{Applications}

\subsection{Unclonable Signatures of Knowledge}

\begin{definition}[Unclonable Extractable SimExt-secure Signatures of Knowledge]
    \label{def:usign}
    Let $\NP$ relation $\cR$ with corresponding language $\cL$ be given such that they can be indexed by a security parameter $\lambda \in \bbN$. Let a message space $\cM$ be given such that it can be indexed by a security parameter $\lambda \in \bbN$.

    $(\Setup, \Sign, \Verify)$ is an unclonable signature of knowledge of a witness with respect to $\cL$ and $\cM$ if it has the following properties:    
    \begin{itemize}
        \item $(\Setup, \Sign, \Verify)$ is a quantum Sim-Ext signature of knowledge (\cref{def:sig}).

        \item {\bf $(k-1) \text{-to-} k$-Unclonable with Extraction}: There exists an oracle-aided polynomial-size quantum circuit $\cE$ such that for every polynomial-size quantum circuit $\cA$, for every tuple of $k-1$ instance-witness pairs $(x_1, \omega_1), \ldots, (x_{k-1}, \omega_{k-1}) \in \cR$, every $\{m_\iota \in \cM_\lambda\}_{\iota \in [k-1]}$, 
        for every $(x, m)$, if there is a polynomial $p(\cdot)$ where
        \begin{equation*}
            \Pr_{\substack{(\crs, \td) \gets \Setup(1^\lambda) \\ 
            \forall \iota \in [k-1], \: \sigma_\iota \gets \Sign(\crs, x_\iota, \omega_\iota, m_\iota) \\ 
            \{\widetilde{\sigma_\iota}\}_{\iota \in [k]} \gets \cA(\crs, \{x_\iota, m_\iota, \sigma_\iota\}_{\iota \in [k-1]})
            }}\left[
            \begin{array}{cc}
                & \exists~ \cJ \subseteq \{j:(\widetilde{x}_j, \widetilde{m}_j) = (x, m)\} \\
                & \text{ s.t. } |\cJ| > |\{i:(x_i, m_i) = (x, m)\}| \\
                & \text{ and } {\forall \iota \in \cJ}, \Verify(\crs, x, m, \widetilde{\sigma_\iota}) = 1
            \end{array}\right]
            \geq \frac{1}{p(\lambda)},
        \end{equation*}
        then there is also a polynomial $q(\cdot)$ such that
        \begin{equation*}
            \Pr_{w \gets \cE^\cA(\{x_\iota, m_\iota\}_{\iota \in [k-1]}, x, m)}\left[(x, w) \in \cR \right] \geq \frac{1}{q(\lambda)}.
        \end{equation*}
    \end{itemize}
\end{definition}

\begin{figure}[!ht]
\begin{framed}
\centering
\begin{minipage}{1.0\textwidth}
\begin{center}
    \underline{Unclonable Signature of Knowledge with CRS}
\end{center}

\vspace{-0.2cm}

Let $(\Setup, \sP, \sV)$ be non-interactive simulation-extractable, adaptive multi-theorem computational zero-knowledge, unclonable-extractable protocol for $\NP$.
Let $\cR$ be the relation with respect to $\cL \in \NP$.

\vspace{1mm}
\noindent {\underline{\textsc{Setup}}$(1^\lambda)$}: $(\crs, \td) \gets \Pi.\Setup(1^\lambda)$.

\vspace{1mm}
\noindent {\underline{\textsc{Sign}}$(\crs, x, w, m)$}:
\vspace{-0.3cm}
\begin{itemize}
    \item Let $x_\Pi = (x, m)$ be an instance and $w_\Pi = w$ be its corresponding witness for the following language $\cL_\Pi$:
    \begin{equation*}
        \{(x, m) \st \exists w \:\st\: (x, w) \in \cR\}.
    \end{equation*}
    \item Compute $\pi_\Pi \gets \Pi.\sP(\crs, x_\Pi, w_\Pi)$.
    \item Output $\sigma = \pi_\Pi$.
\end{itemize}

\vspace{-0.2cm}
\noindent {\underline{\textsc{Verify}}$(\crs, x, m, \sigma)$}: Output $\Pi.\sV(\crs, (x, m), \pi_\Pi)$.

\end{minipage}
\end{framed}
\caption{Unclonable Signature of Knowledge in CRS model}
\label{fig:usign}
\end{figure}

\begin{theorem}
    \label{thm:usig}
    Let $\Pi = (\Setup, \sP, \sV)$ be a non-interactive simulation-extractable, adaptive multi-theorem computational zero-knowledge, unclonable-extractable protocol for $\NP$ (\cref{def:unizk-crs}).
    
    $(\Setup, \Sign, \Verify)$ in \cref{fig:usign} is an unclonable-extractable SimExt-secure signature of knowledge (\cref{def:usign}).
\end{theorem}

\begin{proof}[Proof of \cref{thm:usig}]
Correctness follows naturally. It remains to prove simulateability, extractability, and unclonable extractability.\\

\noindent \textbf{Simulateable.}
Let $\Pi.\Sim = (\Pi.\Sim_0, \Pi.\Sim_1)$ be the adaptive multi-theorem computationally zero-knowledge simulator of $\Pi$. We define $\Sim_0$ with oracle access to $\Pi.\Sim_0$ as follows:
\begin{addmargin}[2em]{2em} 
    \noindent \emph{Input}: $1^\lambda$.
    
    \noindent \textbf{(1)} Send $1^\lambda$ to $\Pi.\Sim_0$. Receive $(\crs, \td)$ from $\Pi.\Sim_0$.

    \noindent \textbf{(2)} Output $\crs$ and $\td$.
\end{addmargin}
We define $\Sim_1$ with oracle access to $\Pi.\Sim_1$ as follows:
\begin{addmargin}[2em]{2em} 
    \noindent \emph{Input}: $\crs$, $\td$, $x$, $m$.

    \noindent \textbf{(1)} Define $x_\Pi = (x, m)$. Send $(\crs, \td, x_\Pi)$ to $\Pi.\Sim_1$. Receive $\pi$ from $\Pi.\Sim_1$.

    \noindent \textbf{(2)} Output $\sigma = \pi$.
\end{addmargin}

Let a polynomial $p(\cdot)$ and an oracle-aided polynomial-size quantum circuit $\cA$ be given such that
\begin{equation*}
    \left\vert \Pr_{(\crs, \td) \gets \Sim_0(1^\lambda)}[\cA^{\Sim_1(\crs, \td, \cdot, \cdot)}(\crs) = 1] - \Pr_{\substack{(\crs, \td) \gets \Setup(1^\lambda)}}[\cA^{\Sign(\crs, \cdot, \cdot, \cdot)}(\crs) = 1] \right\vert \ge \frac{1}{p(\lambda)}.
\end{equation*}
We define a reduction to the multi-theorem zero-knowledge property of $\Pi$ as follows:
\begin{addmargin}[2em]{2em} 
    \noindent \underline{Reduction}: to zero-knowledge of $\Pi$ given oracle access to $\cA$.

    \noindent \textbf{(1)} Receive (real or simulated) $\crs$ from the challenger.

    \noindent \textbf{(2)} Send $\crs$ to $\cA$. 
    
    \noindent \textbf{(3)} On query $(x, w, m)$ from $\cA$: 
    send $x_\Pi = (x, m)$ to the challenger, receives (real or simulated) $\pi$ from the challenger, send $\sigma = \pi$ to $\cA$.

    \noindent \textbf{(4)} Output the result of $\cA$.
\end{addmargin}

The view of $\cA$ matches that of our protocol in \cref{fig:usign} or $\Sim_0$ and $\Sim_1$. As such, this reduction should have the same advantage at breaking the adaptive multi-theorem computational zero-knowledge property of $\Pi$. We reach a contradiction, hence our protocol must be simulateable.\\

\noindent \textbf{Extractable.}
Let $\Pi.\Sim = (\Pi.\Sim_0, \Pi.\Sim_1)$ be the adaptive multi-theorem computationally zero-knowledge simulator of $\Pi$.
Let $\Pi.\Ext$ be the simulation extractable extractor of $\Pi$ defined relative to $\Pi.\Sim$.
Let $\Sim = (\Sim_0, \Sim_1)$ be the simulator given by the simulation property which uses $\Pi.\Sim$.
We define $\Ext$ with oracle access to $\Pi.\Ext$ as follows:
\begin{addmargin}[2em]{2em} 
    \noindent \emph{Input}: $\crs$, $\td$, $x$, $m$, $\sigma=\pi$.

    \noindent \textbf{(1)} Define $x_\Pi = (x, m)$.
    
    \noindent \textbf{(2)} Send $(\crs, \td, x_\Pi, \pi)$ to $\Pi.\Ext$. Receive $w_\Pi = w$ from $\Pi.\Ext$.

    \noindent \textbf{(3)} Output $w$.
\end{addmargin}

Let a polynomial $p(\cdot)$ and an oracle-aided polynomial-size quantum circuit $\cA$ be given such that
\begin{align*}
    \Pr_{\substack{(\crs, \td) \gets \Sim_0(1^\lambda) \\ (x, m, \sigma) \gets \cA^{\Sim_1(\crs, \td, \cdot, \cdot)}(\crs) \\ w \gets \Ext(\crs, \td, x, m, \sigma)}}\left[\Verify(\crs, x, m, \sigma) = 1 \wedge (x, m) \not\in Q \wedge (x, w) \not\in \cR_\lambda \right] \ge \frac{1}{p(\lambda)}
\end{align*}
where $Q$ is the list of queries from $\cA$ to $\Sim_1$.
If $\Verify$ accepts the output of $\cA$, then $\Pi.\sV$ must accept $(\crs, x_\Pi, \pi)$. If $(x, m) \not\in Q$, then since $x_\Pi$ contains $x, m$, $x_\Pi$ must not be in the queries asked to $\Pi.\Sim_1$. Since $(x, w) \not\in \cR$, then $x_\Pi \not\in \cL_\Pi$ by the definition of $\cL_\Pi$. As such, it must necessarily be the case that $(x_\Pi, w_\Pi) \not\in \cR_\Pi$.
Hence, we have that 
\begin{align*}
    \Pr_{\substack{(\crs, \td) \gets \Sim_0(1^\lambda) \\ (x, m, \sigma) \gets \cA^{\Sim_1(\crs, \td, \cdot, \cdot)}(\crs) \\ w \gets \Ext(\crs, \td, x, m, \sigma)}}\left[\Pi.\sV(\crs, x_\Pi, \pi) = 1 \wedge x_\Pi \not\in Q_\Pi \wedge (x_\Pi, w_\Pi) \not\in \cR_\Pi \right] \ge \frac{1}{p(\lambda)}
\end{align*}
where $Q_\Pi$ is the list of queries, originating from $\cA$, that $\Sim_1$ makes to $\Pi.\Sim_1$.
We define a reduction to the simulation extraction property of $\Pi$ as follows:
\begin{addmargin}[2em]{2em} 
    \noindent \underline{Reduction}: to simulation extraction of $\Pi$ given oracle access to $\cA$.

    \noindent \textbf{(1)} Receive $\crs$ from the challenger.

    \noindent \textbf{(2)} Send $\crs$ to $\cA$. 
    
    \noindent \textbf{(3)} On query $(x, w, m)$ from $\cA$: 
    send $x_\Pi = (x, m)$ to the challenger, receives $\pi$ from the challenger, send $\sigma = \pi$ to $\cA$.

    \noindent \textbf{(4)} Receive $(x, m, \sigma = \pi)$ from $\cA$. Define $x_\Pi = (x, m)$.

    \noindent \textbf{(5)} Output $(x_\Pi, \pi)$.
\end{addmargin}

The view of $\cA$ matches that of $\Sim_0$ and $\Sim_1$. As such, this reduction should have the same advantage at breaking the extraction property of $\Pi$. We reach a contradiction, hence our protocol must be extractable.\\

\noindent \textbf{Unclonable Extractability.}
Let $\Pi.\Sim$ be the adaptive multi-theorem computationally zero-knowledge simulator of $\Pi$.
Let $\Pi.\Ext$ be the simulation extractable extractor of $\Pi$ defined relative to $\Pi.\Sim$.
Let $\Pi.\cE$ be the unclonable extractor of $\Pi$.
We define $\cE$ with oracle-access to $\Pi.\cE$,
and some $\cA$ as follows:
\begin{addmargin}[2em]{2em} 
    \noindent \emph{Input}: $\{x_\iota, m_\iota\}_{\iota \in [k-1]}$, $x$, $m$

    \noindent \textbf{(1)} 
    Define $x_{\Pi, \iota} = (x_\iota, m_\iota)$ for $\iota \in [k-1]$.



    
    \noindent \textbf{(2)} Send $(\{x_{\Pi, \iota}\}_{\iota \in [k-1]}, (x, m))$ to $\Pi.\cE$. Receive $(\crs, \{\pi_{\Pi, \iota} = \sigma_{\iota}\}_{\iota \in [k-1]})$ from $\Pi.\cE$.

    \noindent \textbf{(3)} Send $(\crs, \{x_\iota, m_\iota, \sigma_{\iota}\}_{\iota \in [k-1]})$ to $\cA$. Receive $\{(\widetilde{x_\iota}, \widetilde{m_\iota}) = \widetilde{x_{\Pi, \iota}}, \widetilde{\sigma_\iota} = \widetilde{\pi_{\Pi,\iota}}\}_{\iota \in [k]}$ from $\cA$. 

    \noindent \textbf{(4)} Send $\{\widetilde{x_{\Pi, \iota}}, \widetilde{\pi_{\Pi, \iota}}\}_{\iota \in [k]}$ to $\Pi.\cE$. Receive $w_\Pi = w$ from $\Pi.\cE$.

    \noindent \textbf{(5)} Output $w$.



    

    
\end{addmargin}

Let $\cA$, $(x_1, w_1), \ldots, (x_{k-1}, w_{k-1}) \in \cR$, $\{m_\iota \in \cM_\lambda\}_{\iota \in [k-1]}$, $x$, $m$, polynomial $p(\cdot)$, and negligible function $\negl(\cdot)$ be given such $\cA$ outputs more accepting signatures for $(x, m)$ than $\cA$ received, and the extractor $\cE$ is unable to extract a valid witness. Formally, that is that
\begin{align}
    &\Pr_{\substack{(\crs, \td) \gets \Setup(1^\lambda) \\ 
    \forall \iota \in [k-1], \: \sigma_\iota \gets \Sign(\crs, x_\iota, \omega_\iota, m_\iota) \\ 
    \{\widetilde{\sigma_\iota}\}_{\iota \in [k]} \gets \cA(\crs, \{x_\iota, m_\iota, \sigma_\iota\}_{\iota \in [k-1]})
    }}\left[
    \begin{array}{cc}
        & \exists~ \cJ \subseteq \{j:(\widetilde{x}_j, \widetilde{m}_j) = (x, m)\} \\
        & \text{ s.t. } |\cJ| > |\{i:(x_i, m_i) = (x, m)\}| \\
        & \text{ and } {\forall \iota \in \cJ}, \Verify(\crs, x, m, \widetilde{\sigma_\iota}) = 1
    \end{array}\right]
    \geq \frac{1}{p(\lambda)}, \label{eq:usig-ver}
\end{align}
and for all polynomials $p'(\cdot)$ (there are infinitely many $\lambda$) such that
\begin{align}
    &\Pr_{w \gets \cE^\cA(\{x_\iota, m_\iota\}_{\iota \in [k-1]}, x, m)}\left[(x, w) \in \cR_\lambda \right] \le \frac{1}{p'(\lambda)}.\label{eq:usig-ext}
\end{align}

We will now show that we can construct a reduction that breaks the unclonability of the NIZK protocol $\Pi$.
We define a reduction to the unclonability of $\Pi$ as follows:
\begin{addmargin}[2em]{2em} 
    \noindent \underline{Reduction}: to unclonability of $\Pi$ given oracle access to $\cA$.
    
    \noindent \emph{Hardwired with}: $\{x_\iota, m_\iota\}_{\iota \in [k-1]}$, $x$, $m$


    \noindent \textbf{(1)} 
    Define $x_\Pi = (x, m)$ and $x_{\Pi, \iota} = (x_\iota, m_\iota)$ for $\iota \in [k-1]$.
    
    \noindent \textbf{(2)} Send $(\{x_{\Pi, \iota}\}_{\iota \in [k-1]}, x_\Pi)$ to the challenger. 
    
    \noindent We note that the following code is re-windable, as necessary:
    
    \noindent \textbf{(3)} Receive $(\crs, \{\pi_{\iota}\}_{\iota \in [k-1]})$ from the challenger. Define $\sigma_{\iota} = \pi_{\Pi, \iota}$ for $\iota \in [k-1]$.

    \noindent \textbf{(4)} Send $(\crs, \{x_\iota, m_\iota, \sigma_{\iota}\}_{\iota \in [k-1]})$ to $\cA$.

    \noindent \textbf{(5)} Receive $\{\widetilde{x_\iota}, \widetilde{m_\iota}, \widetilde{\sigma_\iota}\}_{\iota \in [k]}$ from $\cA$. Define $\widetilde{\pi_{\Pi,\iota}} = \widetilde{\sigma_\iota}$ and $\widetilde{x_{\Pi, \iota}} = (\widetilde{x_\iota}, \widetilde{m_\iota})$ for $\iota \in [k]$.
    
    \noindent \textbf{(6)} Output $\{\widetilde{x_{\Pi, \iota}}, \widetilde{\pi_\iota}\}_{\iota \in [k]}$.
\end{addmargin}
We note that the reduction does not change the view of $\cA$ from the honest execution.

If the challenger runs $\Pi.\Setup$ and $\Pi.\Prove$, then the reduction clones proofs of $\Pi$ with noticeable probability. This can be seen from \cref{eq:usig-ver} because the reduction should have the same advantage at cloning more accepting proofs for the statement $(x, m)$ as $\cA$ does at cloning signatures for instance and message pair $(x, m)$. 

If the challenger runs $\Pi.\cE$ given oracle access to our reduction, the challenger will succeed in extraction with only negligible probability. 
This can be seen from \cref{eq:usig-ext} because $\cE$ directly uses $\Pi.\cE$ to extract any witness, and the functionality that $\cE$ provides to $\Pi.\cE$ is the same as the functionality that our reduction provides to the challenger.
Thus, the above reduction leads to a contradiction with the unclonability of the NIZK protocol $\Pi$.

\end{proof}

\begin{corollary}
    \label{cor:usig}
    Assuming the polynomial quantum hardness of LWE, injective one-way functions exist, post-quantum iO exists, there exists an unclonable SimExt-secure signature of knowledge (\cref{def:usign}).
\end{corollary}

\begin{proof}
    This follows from \cref{cor:main-crs} and \cref{thm:usig}.
\end{proof}

%% file: revocable-anonymous-credentials.tex
\subsection{Revocable Anonymous Credentials}
\label{sec:rac}

In this section, we will see how to use unclonable signatures of knowledge to construct an anonymous credentials scheme which has a natural revocation property.

\begin{definition}[Revocable Anonymous Credentials]
\label{def:rac}

$(\IGen, \Issue, \VerCred, \Revoke, \Prove,\\ \VerRevoke)$ is a revocable anonymous credentials scheme with respect to some set of accesses $\{\cS_\lambda\}_{\lambda \in \bbN}$ if it has the following syntax and properties.\\

\noindent \textbf{Syntax.}
The input $1^\lambda$ is left out when it is clear from context.
\begin{itemize}
    \item $(\nym, \sk) \gets \IGen(1^\lambda)$: The probabilistic polynomial-time algorithm $\IGen$ is run by the issuer of the credentials. It takes input $1^\lambda$; and outputs a pseudonym $\nym$ with a secret key $\sk$.
    
    \item $\cred \gets \Issue(1^\lambda, \nym, \sk, \access)$: The polynomial-time quantum algorithm $\Issue$ is run by the issuer of the credentials. It takes input the issuer's keys $\nym$ and $\sk$ as well as the requested access $\access \in \cS_\lambda$; and outputs a quantum credential $\cred$ along with a classical identifier $\id$.
    
    \item $\VerCred(1^\lambda, \nym, \access, \cred) \in \zo$: The polynomial-time quantum algorithm $\VerCred$ is run by a verifier of the user's credentials. It takes input the issuer's pseudonym $\nym$, the requested access $\access \in \cS_\lambda$, and a credential $\cred$; and outputs $1$ iff $\cred$ is a valid credential for access $\access$ with respect to $\nym$.

    \item $\revnotice \gets \Revoke(1^\lambda, \nym, \sk, \access)$: The polynomial-time quantum algorithm $\Revoke$ is run by the issuer of the credentials. It takes input the issuer's keys $\nym$ and $\sk$, and the access $\access$ being revoked; and outputs a notice of revocation $\revnotice$.

    \item $\pi \gets \Prove(1^\lambda, \nym, \revnotice, \cred)$: The polynomial-time quantum algorithm $\Prove$ is run by the user of the credentials. It takes input the issuer's pseudonym $\nym$, and a revocation notice $\revnotice$, and the credential to be revoked $\cred$ which is identified by $\revnotice$; and outputs a proof of revocation $\pi$.

    \item $\VerRevoke(1^\lambda, \nym, \sk, \access, \revnotice, \pi) \in \zo$: The polynomial-time quantum algorithm $\VerRevoke$ is run by the issuer of the credentials. It takes input the issuer's keys $\nym$ and $\sk$, the access $\access$  being revoked, the revocation notice $\revnotice$, and a proof of revocation $\pi$; and outputs $1$ iff $\pi$ is a valid proof that the user's access to the credential identified by $\id$ has been revoked.
\end{itemize}

\noindent \textbf{Properties.}
\begin{itemize}
    \item {\bf Correctness}: For every sufficiently large $\lambda \in \bbN$, and every $\access \in \cS_\lambda$,
    \begin{equation*}
        \Pr_{\substack{(\nym, \sk) \gets \IGen(1^\lambda) \\ \cred \gets \Issue(1^\lambda, \nym, \sk, \access)}}[\VerCred(1^\lambda, \nym, \access, \cred) = 1] = 1
    \end{equation*}
    and
    \begin{equation*}
        \Pr_{\substack{(\nym, \sk) \gets \IGen(1^\lambda) \\ \cred \gets \Issue(1^\lambda, \nym, \sk, \access) \\ \revnotice \gets \Revoke(1^\lambda, \nym, \sk, \access) \\ \pi \gets \Prove(1^\lambda, \nym, \revnotice, \cred)}}[\VerRevoke(1^\lambda, \nym, \sk, \access, \revnotice, \pi) = 1] = 1.
    \end{equation*}

    \item {\bf Revocation}: For every polynomial-size quantum circuit $\cA$, there exists a negligible function $\negl(\cdot)$ such that for sufficiently large $\lambda \in \bbN$, and every $\access \in \cM_\lambda$
    \begin{equation*}
        \Pr_{\substack{(\nym, \sk) \gets \IGen(1^\lambda) \\ \cred \gets \Issue(1^\lambda, \nym, \sk, \access) \\ \revnotice \gets \Revoke(1^\lambda, \nym, \sk, \access) \\ \pi, \cred' \gets \cA(1^\lambda, \nym, \revnotice, \cred)}}
        \Bigg[\substack{\VerRevoke(1^\lambda, \nym, \sk, \access, \revnotice, \pi) = 1 \\ 
        \bigwedge \VerCred(1^\lambda, \nym, \access, \cred') = 1}\Bigg] \le \negl(\lambda).
    \end{equation*}
\end{itemize}
\end{definition}

\begin{remark}
    Unlike previous literature, the users that get issued credentials do not have their own identity. We also define algorithms for a three-message revocation process as opposed to the polynomial-message revocation process defined in the literature. 
\end{remark}

We now introduce a construction based on unclonable signatures of knowledge.

\begin{figure}[!ht]
\begin{framed}
\centering
\begin{minipage}{1.0\textwidth}
\begin{center}
    \underline{Revocable Anonymous Credentials}
\end{center}

\vspace{2mm}

Let $(\cX, \cW)$ be a hard-distribution of instance and witness pairs for some $\NP$ relation.
Let $\{\cS_\lambda\}_{\lambda \in \bbN}$ be some set of accesses.
Let $(\Setup, \Sign, \Verify)$ be an unclonable-extractable SimExt-secure signature of knowledge for message space $\{\cS_\lambda\}_{\lambda \in \bbN}$ (\cref{def:usign}).

\vspace{1mm}
\noindent {\underline{\textsc{IssuerKeyGen}}$(1^\lambda)$}:
\begin{itemize}
    \item $(\crs, \td) \gets \Setup(1^\lambda)$.
    \item $(x, w) \gets (\cX, \cW)$.
    \item Output $\nym = (\crs, x)$ and $\sk = (\td, w)$.
\end{itemize}

\vspace{1mm}
\noindent {\underline{\textsc{Issue}}$(\nym, \sk, \access)$}:
\begin{itemize}
    \item $\sigma \gets \Sign(\crs, x, w, \access)$.
    \item Output $\cred = \sigma$.
\end{itemize}

\vspace{1mm}
\noindent {\underline{\textsc{VerifyCred}}$(\nym, \access, \cred)$}:
\begin{itemize}
    \item Output $\Verify(\crs, x, \access, \cred)$.
\end{itemize}

\vspace{1mm}
\noindent {\underline{\textsc{Revoke}}$(\nym, \sk, \access)$}:
\begin{itemize}
    \item Output $\revnotice = \access$.
\end{itemize}

\vspace{1mm}
\noindent {\underline{\textsc{Prove}}$(\nym, \revnotice, \cred)$}:
\begin{itemize}
    \item Output $\pi = \cred$.
\end{itemize}

\vspace{1mm}
\noindent {\underline{\textsc{VerifyRevoke}}$(\nym, \sk, \access, \revnotice, \pi)$}:
\begin{itemize}
    \item Output $\VerCred(\nym, \access, \pi)$.
\end{itemize}

\end{minipage}
\end{framed}
\caption{Revocable Anonymous Credentials}
\label{fig:rac}
\end{figure}

\begin{theorem}
\label{thm:rac}
Let $(\cX, \cW)$ be a hard-distribution of instance and witness pairs for some $\NP$ relation.
Let $\{\cS_\lambda\}_{\lambda \in \bbN}$ be some set of accesses.
Let $(\Setup, \Sign, \Verify)$ be an unclonable-extractable SimExt-secure signature of knowledge for message space $\{\cS_\lambda\}_{\lambda \in \bbN}$ (\cref{def:usign}). 

$(\IGen, \Issue, \VerCred, \Revoke, \Prove, \VerRevoke)$ defined in \cref{fig:rac} is a revocable anonymous credentials scheme (\cref{def:rac}).
\end{theorem}

\begin{proof}[Proof Sketch]
The correctness of this revocable anonymous credentials scheme follows from the correctness of the unclonable signature of knowledge scheme.

We will now sketch the proof of revocation. Say that there exists an adversary $\cA$, access $\access$, and polynomial $p(\cdot)$ such that, with probability at least $1/p(\lambda)$: (1) $\pi$ passes the revocation check, and (2) $\cred'$ passes the credential check. This means that both $\pi$ and $\cred'$ are valid signatures with respect to the same $\crs$, $x$, and $\access$ that the signature $\cred$ was issued under. This satisfies the ``if'' condition of the unclonability property of the unclonable signature of knowledge. As such, there exists a polynomial $q(\cdot)$ such that the unclonable signature of knowledge's extractor can produce a witness $w$ for $x$ with probability at least $1/q(\lambda)$. However, this contradicts the hardness of the distribution $(\cX, \cW)$. Hence, our protocol must have the revocation property.
\end{proof}

\begin{corollary}
    Assuming the polynomial quantum hardness of LWE, injective one-way functions exist, post-quantum iO exists, and the hardness of $\NP$, there exists a revocable anonymous credentials scheme (\cref{def:rac}).
\end{corollary}

\begin{proof}
    This follows from \cref{cor:usig} and \cref{thm:rac}.
\end{proof}

\subsection{Unclonable Anonymous Credentials}

We will show that our revocable anonymous credentials construction in \cref{fig:rac} also satisfies a definition of unclonable anonymous credentials.

\begin{definition}[Unclonable Anonymous Credentials]
\label{def:uac}

$(\IGen, \Issue, \VerCred)$ is an unclonable anonymous credentials scheme with respect to some set of accesses $\{\cS_\lambda\}_{\lambda \in \bbN}$ if it has the following syntax and properties.\\

\noindent \textbf{Syntax.}
The input $1^\lambda$ is left out when it is clear from context.
\begin{itemize}
    \item $(\nym, \sk) \gets \IGen(1^\lambda)$: The probabilistic polynomial-time algorithm $\IGen$ is run by the issuer of the credentials. It takes input $1^\lambda$; and outputs a pseudonym $\nym$ with a secret key $\sk$.
    
    \item $\cred \gets \Issue(1^\lambda, \nym, \sk, \access)$: The polynomial-time quantum algorithm $\Issue$ is run by the issuer of the credentials. It takes input the issuer's keys $\nym$ and $\sk$ as well as the requested access $\access \in \cS_\lambda$; and outputs a quantum credential $\cred$ along with a classical identifier $\id$.
    
    \item $\VerCred(1^\lambda, \nym, \access, \cred) \in \zo$: The polynomial-time quantum algorithm $\VerCred$ is run by a verifier of the user's credentials. It takes input the issuer's pseudonym $\nym$, the requested access $\access \in \cS_\lambda$, and a credential $\cred$; and outputs $1$ iff $\cred$ is a valid credential for access $\access$ with respect to $\nym$.
\end{itemize}

\noindent \textbf{Properties.}
\begin{itemize}
    \item {\bf Correctness}: For every sufficiently large $\lambda \in \bbN$, and every $\access \in \cS_\lambda$,
    \begin{equation*}
        \Pr_{\substack{(\nym, \sk) \gets \IGen(1^\lambda) \\ \cred \gets \Issue(1^\lambda, \nym, \sk, \access)}}[\VerCred(1^\lambda, \nym, \access, \cred) = 1] = 1.
    \end{equation*}

    \item {\bf Unclonable}: For every polynomial-size quantum circuit $\cA$, there exists a negligible function $\negl(\cdot)$ such that for sufficiently large $\lambda \in \bbN$, and every $\access \in \cM_\lambda$
    \begin{equation*}
        \Pr_{\substack{(\nym, \sk) \gets \IGen(1^\lambda) \\ \cred \gets \Issue(1^\lambda, \nym, \sk, \access) \\ \cred_0, \cred_1 \gets \cA(1^\lambda, \nym, \cred)}}
        \Bigg[\substack{\VerCred(1^\lambda, \nym, \access, \cred_0) = 1 \\ 
        \bigwedge \VerCred(1^\lambda, \nym, \access, \cred_1) = 1}\Bigg] \le \negl(\lambda).
    \end{equation*}
\end{itemize}
\end{definition}

\begin{theorem}
\label{thm:uac}
Let $(\cX, \cW)$ be a hard-distribution of instance and witness pairs for some $\NP$ relation.
Let $\{\cS_\lambda\}_{\lambda \in \bbN}$ be some set of accesses.
Let $(\Setup, \Sign, \Verify)$ be an unclonable-extractable SimExt-secure signature of knowledge for message space $\{\cS_\lambda\}_{\lambda \in \bbN}$ (\cref{def:usign}). 

$(\IGen, \Issue, \VerCred)$ defined in \cref{fig:rac} is an unclonable anonymous credentials scheme (\cref{def:uac}).
\end{theorem}

\begin{proof}[Proof Sketch]
The correctness of this unclonable anonymous credentials scheme follows from the correctness of the unclonable signature of knowledge scheme.

We will now sketch the proof of unclonability. Say that there exists an adversary $\cA$, access $\access$, and polynomial $p(\cdot)$ such that, with probability at least $1/p(\lambda)$: (1) $\cred_0$ passes the credential check, and (2) $\cred_1$ passes the credential check. This means that both $\cred_0$ and $\cred_1$ are valid signatures with respect to the same $\crs$, $x$, and $\access$ that the signature $\cred$ was issued under. This satisfies the ``if'' condition of the unclonability property of the unclonable signature of knowledge. As such, there exists a polynomial $q(\cdot)$ such that the unclonable signature of knowledge's extractor can produce a witness $w$ for $x$ with probability at least $1/q(\lambda)$. However, this contradicts the hardness of the distribution $(\cX, \cW)$. Hence, our protocol must have the revocation property.
\end{proof}

\begin{corollary}
    Assuming the polynomial quantum hardness of LWE, injective one-way functions exist, post-quantum iO exists, and the hardness of $\NP$, there exists an unclonable anonymous credentials scheme (\cref{def:uac}).
\end{corollary}

\begin{proof}
    This follows from \cref{cor:usig} and \cref{thm:uac}.
\end{proof}

%% file: acks.tex
\section{Acknowledgments}
The authors were supported in part by DARPA SIEVE, NSF QIS-2112890, NSF CAREER CNS-2238718, and NSF CNS-2247727.
This material is based on work supported by DARPA under Contract No. 
HR001120C0024. Any opinions, findings and conclusions or recommendations expressed in
this material are those of the author(s) and do not necessarily reflect the views of the United States
Government or DARPA.

%% file: udef-app.tex
\section{A Reduction Between  Unclonability Definitions}
\label{app:defs-reduct}

\subsection{In the CRS model}

For completeness, here we repeat the definitions of unclonability.

\begin{definition} (Unclonable Security for Hard Instances).
\label{def:uncnizk-app}
A proof  $(\Setup,\mathsf{Prove},\mathsf{Verify})$ satisfies unclonable security if 
for every language $\mathcal{L}$ with corresponding relation $\cR_\cL$, 
for every polynomial-sized quantum circuit family $\{C_\lambda\}_{\lambda \in \mathbb{N}}$,
and for every hard distribution $\{\mathcal{X}_\lambda,\mathcal{W}_\lambda\}_{\lambda \in \bbN}$ over $\cR_\cL$,
there exists a negligible function $\negl(\cdot)$ such that
for every $\lambda \in \bbN$,
\begin{equation*}
    \Pr_{(x,w) \leftarrow (\mathcal{X}_\lambda,\mathcal{W}_\lambda)}\Bigg[
    \mathsf{Verify}(\crs,x,\pi_1) = 1 \bigwedge 
    \mathsf{Verify}(\crs,x,\pi_2) = 1
    \Bigg|
    \substack{(\crs, \td) \leftarrow \Setup(1^\lambda)\\
    \pi \leftarrow \mathsf{Prove}(\crs,x,w)\\
    \pi_1, \pi_2 \leftarrow C_\lambda(x, \pi)
    }
    \Bigg]
    \leq \negl(\lambda).
\end{equation*}
\end{definition}

\begin{definition}($1$-to-$2$ Unclonable Extractability)
    \label{def:right-app}
    A proof $(\Setup,\mathsf{Prove},\mathsf{Verify})$ satisfies unclonable security 
    there exists a QPT extractor $\mathcal{E}$ which is an oracle-aided circuit such that 
    for every language $\cL$ with corresponding relation $\cR_\cL$ and 
    for every non-uniform polynomial-time quantum adversary $\mathcal{A}$,
    for every instance-witness pair $(x,w) \in \cR_\cL$ 
    and $\lambda = \lambda(|x|)$,
    such that there is a polynomial $p(\cdot)$ satisfying:
    \begin{equation}
    \label{eq:ag}
    \Pr\Bigg[
    \mathsf{Verify}(\crs,x,\pi_1) = 1 \bigwedge 
    \mathsf{Verify}(\crs,x,\pi_2) = 1
    \Bigg|
    \substack{(\crs, \td) \leftarrow \Setup(1^\lambda)\\
    \pi \leftarrow \mathsf{Prove}(\crs,x,w)\\
    \pi_1, \pi_2 \leftarrow \cA_\lambda(\crs, x, \pi, z)
    }
    \Bigg]
    \geq \frac{1}{p(\lambda)},
    \end{equation} 
    there is also a polynomial $q(\cdot)$ such that
    \begin{equation}
        \Pr[(x,w_\cA) \in \cR_\cL| w_\cA \leftarrow \mathcal{E}^{\cA}(x)] \geq \frac{1}{q(\lambda)}.
    \end{equation}
\end{definition}

\begin{claim}
    Any protocol satisfying \cref{def:right-app} also satisfies \cref{def:uncnizk-app}.
\end{claim}

\begin{proof}
Suppose there exists a protocol $\Pi = (\Setup, \mathsf{Prove},\mathsf{Verify})$ satisfying \cref{def:right-app}.

Suppose towards a contradiction that $\Pi$ does not satisfy \cref{def:uncnizk-app}. This implies that there is a QPT adversary $\widehat{\mathcal{A}}$, auxiliary input $\widehat{z} = \{\widehat{z}_\lambda\}_{\lambda \in \mathbb{N}}$, a hard distribution $(\cX, \cW)$ over $\cR_\cL$, and a polynomial $p(\cdot)$ such that 
\begin{equation}
\label{eq:hg}
\Pr_{(x,w) \leftarrow (\mathcal{X},\mathcal{W})}\Bigg[
\mathsf{Verify}(\crs,x,\pi_1) = 1 \bigwedge 
\mathsf{Verify}(\crs,x,\pi_2) = 1
\Bigg|
\substack{(\crs, \td) \leftarrow \Setup(1^\lambda)\\
\pi \leftarrow \mathsf{Prove}(\crs,x,w)\\
\pi_1, \pi_2 \leftarrow \widehat{\cA}_\lambda(\crs, x, \pi, \widehat{z})
}
\Bigg]
\geq \frac{1}{p(\lambda)}.
\end{equation}
Let $S$ denote the set of instance-witness pairs $\{(x, w) \in (\cX, \cW)\}$ that satisfy \begin{equation}
\label{eq:ig}
\Pr \Bigg[
\mathsf{Verify}(\crs,x,\pi_1) = 1 \bigwedge 
\mathsf{Verify}(\crs,x,\pi_2) = 1
\Bigg|
\substack{(\crs, \td) \leftarrow \Setup(1^\lambda)\\
\pi \leftarrow \mathsf{Prove}(\crs,x,w)\\
\pi_1, \pi_2 \leftarrow \widehat{\cA}_\lambda(\crs, x, \pi, \widehat{z})
}
\Bigg]
\geq \frac{1}{2p(\lambda)}.
\end{equation}
First, we claim that
\begin{equation}
\label{eq:new}
\Pr_{(x,w) \leftarrow (\mathcal{X},\mathcal{W})} [(x,w) \in S] \geq \frac{1}{2p(\lambda)}
\end{equation}
Suppose not, then by \cref{eq:ig}, 
\begin{equation*}
\Pr_{(x,w) \leftarrow (\mathcal{X},\mathcal{W})}\Bigg[
\mathsf{Verify}(\crs,x,\pi_1) = 1 \bigwedge 
\mathsf{Verify}(\crs,x,\pi_2) = 1
\Bigg|
\substack{(\crs, \td) \leftarrow \Setup(1^\lambda)\\
\pi \leftarrow \mathsf{Prove}(\crs,x,w)\\
\pi_1, \pi_2 \leftarrow \widehat{\cA}_\lambda(\crs, x, \pi, \widehat{z})
}
\Bigg]
< \frac{1}{2p(\lambda)} + \frac{1}{2p(\lambda)}.
\end{equation*}
contradicting \cref{eq:hg}. Thus, \cref{eq:new} must be true.

Consider the extractor $\cE$ guaranteed by \cref{def:right-app}.
Given a sample $(x,w) \leftarrow (\cX,\cW)$, we will show that there is a polynomial $p'(\cdot)$ such that
\begin{equation}
\label{eq:bg}
\Pr_{(x,w) \leftarrow (\cX, \cW)}[\cE^{\widehat{\cA}}(x, \widehat{z}) \in \cR_\cL(x)] \geq \frac{1}{p'(\lambda)}
\end{equation}
which suffices to contradict hardness of the distribution $(\cX,\cW)$, as desired.

Towards showing that \cref{eq:bg} holds,
recall by \cref{def:right-app} that
for every NP instance-witness pair $(x,w)$ such that there is a polynomial $p(\cdot)$ satisfying:
\begin{equation*}
\Pr\Bigg[
\mathsf{Verify}(\crs,x,\pi_1) = 1 \bigwedge 
\mathsf{Verify}(\crs,x,\pi_2) = 1
\Bigg|
\substack{(\crs, \td) \leftarrow \Setup(1^\lambda)\\
\pi \leftarrow \mathsf{Prove}(\crs,x,w)\\
\pi_1, \pi_2 \leftarrow \widehat{A}_\lambda(\crs, x, \pi, \widehat{z})
}
\Bigg]
\geq \frac{1}{p(\lambda)},
\end{equation*} 
there is also a polynomial $q(\cdot)$ such that
\begin{equation*}
    \Pr[R_L(x,w) = 1| w \leftarrow \mathcal{E}^{\widehat{A}}(x, \widehat{z})] \geq \frac{1}{q(\lambda)}
\end{equation*}
This implies that there is a polynomial $q(\cdot)$ such that for every $(x, w) \in S$, 
\begin{equation*}
    \Pr[R_L(x,w) = 1| w \leftarrow \mathcal{E}^{\widehat{A}}(x, \widehat{z})] \geq \frac{1}{q(\lambda)}
\end{equation*}
This, combined with \cref{eq:new} implies that
\begin{equation*}
    \Pr_{(x, w) \leftarrow (\cX, \cW)}[R_L(x,w) = 1| w \leftarrow \mathcal{E}^{\widehat{A}}(x, \widehat{z})] \geq \frac{1}{2p(\lambda)q(\lambda)}
\end{equation*}
which proves \cref{eq:bg} as desired.
\end{proof}

\subsection{In the QRO model}

For completeness, here we repeat the definitions of unclonability.

\begin{definition} (Unclonable Security for Hard Instances).
\label{def:uncnizk-app-rom}
A proof  $(\mathsf{Prove},\mathsf{Verify})$ satisfies unclonable security with respect to a quantum random oracle $\cO$ if 
for every language $\cL$ with corresponding relation $\cR_\cL$, 
for every polynomial-sized quantum oracle-aided circuit family $\{C_\lambda\}_{\lambda \in \mathbb{N}}$,
and for every hard distribution $\{\mathcal{X}_\lambda,\mathcal{W}_\lambda\}_{\lambda \in \bbN}$ over $\cR_\cL$,
there exists a negligible function $\negl(\cdot)$ such that
for every $\lambda \in \bbN$,
\begin{equation*}
    \Pr_{(x,w) \leftarrow (\mathcal{X}_\lambda,\mathcal{W}_\lambda)}\Bigg[
    \mathsf{Verify}^\cO(x,\pi_1) = 1 \bigwedge 
    \mathsf{Verify}^\cO(x,\pi_2) = 1
    \Bigg|
    \substack{
    \pi \leftarrow \mathsf{Prove}^\cO(x,w)\\
    \pi_1, \pi_2 \leftarrow C_\lambda(x, \pi)}
    \Bigg]
    \leq \negl(\lambda).
\end{equation*}
\end{definition}

\begin{definition}($1$-to-$2$ Unclonable Extractability)
    \label{def:right-app-rom}
    A proof $(\mathsf{Prove},\mathsf{Verify})$ satisfies unclonable security with respect to a quantum random oracle $\cO$
    there exists a QPT extractor $\mathcal{E}$ which is an oracle-aided circuit such that 
    for every language $\cL$ with corresponding relation $\cR_\cL$ and 
    for every non-uniform polynomial-time quantum adversary $\mathcal{A}$,
    for every instance-witness pair $(x,w) \in \cR_\cL$ 
    and $\lambda = \lambda(|x|)$,
    such that there is a polynomial $p(\cdot)$ satisfying:
    \begin{equation}
    \label{eq:ag-rom}
    \Pr\Bigg[
    \mathsf{Verify}^\cO(x,\pi_1) = 1 \bigwedge 
    \mathsf{Verify}^\cO(x,\pi_2) = 1
    \Bigg|
    \substack{
    \pi \leftarrow \mathsf{Prove}^\cO(x,w)\\
    \pi_1, \pi_2 \leftarrow \cA_\lambda^\cO(x, \pi, z)
    }
    \Bigg]
    \geq \frac{1}{p(\lambda)},
    \end{equation} 
    there is also a polynomial $q(\cdot)$ such that
    \begin{equation}
        \Pr[(x,w_\cA) \in \cR_\cL| w_\cA \leftarrow \mathcal{E}^{\cA^{\ket{\cO}}}(x)] \geq \frac{1}{q(\lambda)}.
    \end{equation}
\end{definition}

\begin{claim}
    Any protocol satisfying \cref{def:right-app-rom} also satisfies \cref{def:uncnizk-app-rom}.
\end{claim}

\begin{proof}
Suppose there exists a protocol $\Pi = (\mathsf{Prove},\mathsf{Verify})$ satisfying \cref{def:right-app-rom}.

Suppose towards a contradiction that $\Pi$ does not satisfy \cref{def:uncnizk-app-rom}. This implies that there is a QPT adversary $\widehat{\mathcal{A}}$ with oracle access to some quantum random oracle $\cO$, auxiliary input $\widehat{z} = \{\widehat{z}_\lambda\}_{\lambda \in \mathbb{N}}$, a hard distribution $(\cX, \cW)$ over $\cR_\cL$, and a polynomial $p(\cdot)$ such that 
\begin{equation}
\label{eq:hg-rom}
\Pr_{(x,w) \leftarrow (\mathcal{X},\mathcal{W})}\Bigg[
\mathsf{Verify}^\cO(x,\pi_1) = 1 \bigwedge 
\mathsf{Verify}^\cO(x,\pi_2) = 1
\Bigg|
\substack{\pi \leftarrow \mathsf{Prove}^\cO(x,w)\\
\pi_1, \pi_2 \leftarrow \widehat{\cA}_\lambda^\cO(x, \pi, \widehat{z})
}
\Bigg]
\geq \frac{1}{p(\lambda)}.
\end{equation}
Let $S$ denote the set of instance-witness pairs $\{(x, w) \in (\cX, \cW)\}$ that satisfy \begin{equation}
\label{eq:ig-rom}
\Pr \Bigg[
\mathsf{Verify}^\cO(x,\pi_1) = 1 \bigwedge 
\mathsf{Verify}^\cO(x,\pi_2) = 1
\Bigg|
\substack{
\pi \leftarrow \mathsf{Prove}^\cO(x,w)\\
\pi_1, \pi_2 \leftarrow \widehat{\cA}_\lambda^\cO( x, \pi, \widehat{z})
}
\Bigg]
\geq \frac{1}{2p(\lambda)}.
\end{equation}
First, we claim that
\begin{equation}
\label{eq:new-rom}
\Pr_{(x,w) \leftarrow (\mathcal{X},\mathcal{W})} [(x,w) \in S] \geq \frac{1}{2p(\lambda)}
\end{equation}
Suppose not, then by \cref{eq:ig-rom}, 
\begin{equation*}
\Pr_{(x,w) \leftarrow (\mathcal{X},\mathcal{W})}\Bigg[
\mathsf{Verify}^\cO(x,\pi_1) = 1 \bigwedge 
\mathsf{Verify}^\cO(x,\pi_2) = 1
\Bigg|
\substack{
\pi \leftarrow \mathsf{Prove}^\cO(x,w)\\
\pi_1, \pi_2 \leftarrow \widehat{\cA}_\lambda^\cO( x, \pi, \widehat{z})
}
\Bigg]
< \frac{1}{2p(\lambda)} + \frac{1}{2p(\lambda)}.
\end{equation*}
contradicting \cref{eq:hg-rom}. Thus, \cref{eq:new-rom} must be true.

Consider the extractor $\cE$ guaranteed by \cref{def:right-app-rom}.
Given a sample $(x,w) \leftarrow (\cX,\cW)$, we will show that there is a polynomial $p'(\cdot)$ such that
\begin{equation}
\label{eq:bg-rom}
\Pr_{(x,w) \leftarrow (\cX, \cW)}[\cE^{\widehat{\cA}}(x, \widehat{z}) \in \cR_\cL(x)] \geq \frac{1}{p'(\lambda)}
\end{equation}
which suffices to contradict hardness of the distribution $(\cX,\cW)$, as desired.

Towards showing that \cref{eq:bg-rom} holds,
recall by \cref{def:right-app-rom} that
for every NP instance-witness pair $(x,w)$ such that there is a polynomial $p(\cdot)$ satisfying:
\begin{equation*}
\Pr\Bigg[
\mathsf{Verify}^\cO(x,\pi_1) = 1 \bigwedge 
\mathsf{Verify}^\cO(x,\pi_2) = 1
\Bigg|
\substack{
\pi \leftarrow \mathsf{Prove}^\cO(x,w)\\
\pi_1, \pi_2 \leftarrow \widehat{A}_\lambda^{\ket{\cO}}(\crs, x, \pi, \widehat{z})
}
\Bigg]
\geq \frac{1}{p(\lambda)},
\end{equation*} 
there is also a polynomial $q(\cdot)$ such that
\begin{equation*}
    \Pr[R_L(x,w) = 1| w \leftarrow \mathcal{E}^{\widehat{A}}(x, \widehat{z})] \geq \frac{1}{q(\lambda)}
\end{equation*}
This along with \cref{eq:hg-rom} implies that there is a polynomial $q(\cdot)$ such that for every $(x, w) \in S$, 
\begin{equation*}
    \Pr[R_L(x,w) = 1| w \leftarrow \mathcal{E}^{\widehat{A}}(x, \widehat{z})] \geq \frac{1}{q(\lambda)}
\end{equation*}
This, combined with \cref{eq:new-rom} implies that
\begin{equation*}
    \Pr_{(x, w) \leftarrow (\cX, \cW)}[R_L(x,w) = 1| w \leftarrow \mathcal{E}^{\widehat{A}}(x, \widehat{z})] \geq \frac{1}{2p(\lambda)q(\lambda)}
\end{equation*}
which proves \cref{eq:bg-rom} as desired.
\end{proof}